\documentclass{lmcs} 
\usepackage[utf8]{inputenc}
\pdfoutput=1

\usepackage{lastpage}
\lmcsdoi{18}{4}{4}
\lmcsheading{}{\pageref{LastPage}}{}{}%
{Mar.~22,~2021}{Nov.~09,~2022}{}

\keywords{Event-B, Formal Semantics, Institutions, Refinement, Modularisation}

\usepackage{hyperref}
\theoremstyle{plain} 

\usepackage{mathpartir}
\usepackage{booktabs}
\usepackage{bsymb,newestb2latex}
\usepackage{programs}
\usepackage{amsmath}
\usepackage{amsfonts}
\usepackage{tikz-cd}
\usepackage{graphicx}

\usepackage{mathtools}

\newcommand\CommentNF[1]{\color{gray}{\#\; \hbox{\emph{#1}}}}
\newcommand\Comment[1]{\hfill \hspace{0.1cm}\color{gray}{\#\; \hbox{\emph{#1}}}}
\newcommand\CommentLN[1]{\hfill \hspace{0.1cm}\color{gray}{\; \hbox{\emph{#1}}}}
\newcommand\maplet[2]{\ensuremath{#1\mathop{\mapsto}#2}}

\newcommand{\groundevt}{\phi({\overline{x}, \overline{x}'}) }
\newcommand{\groundevtinit}{\phi({\overline{x}'}) }
\newcommand\BR[1]{\ensuremath{\llbracket #1 \rrbracket}}

\newcommand\setof[1]{\mathcal{P}(#1)} 

\newcommand\EBtoEVT{\textsc{eb2evt}}
\newcommand\Rodin{\textsf{Rodin}}
\newcommand\Hets{\textsc{Hets}}

\usepackage{stmaryrd}
\usepackage{url}

\usepackage{amsthm}
\usepackage{thmtools, thm-restate}

\newtheorem{subprop}[thm]{Sub-Proposition}

\usepackage{chngcntr}
\usepackage{apptools}
\AtAppendix{\counterwithin{defi}{section}}

\InBodyLeftNumberLine%

\begin{document}

\title[Building Specifications in the Event-B Institution]
      {Building Specifications in the Event-B Institution}

\author[M. Farrell]{Marie Farrell\lmcsorcid{0000-0001-7708-3877}}	
\address{Department of Computer Science and Hamilton Institute, Maynooth University, Ireland}	
\email{marie.farrell@mu.ie}  
\thanks{This work was initially funded by a Government of Ireland Postgraduate Grant from the Irish Research Council, and subsequently supported by EPSRC Hubs for Robotics and AI in Hazardous Environments: EP/R026092 (FAIR-SPACE)}
\thanks{Farrell and Monahan dedicate this paper to the memory of Dr.~James F. Power who passed away before he could see this work accepted for publication. We thank him for his contributions and encouragement throughout this project}

\author[R. Monahan]{Rosemary Monahan\lmcsorcid{0000-0003-3886-4675}} 

\author[J.F. Power]{James F. Power\lmcsorcid{0000-0002-6770-3910}}

\begin{abstract}
This paper describes a formal semantics for the Event-B specification
language using the theory of institutions.  We define an institution for
Event-B, {\iEVT}, and prove that it meets the validity requirements
for satisfaction preservation and model amalgamation.  We also present a
series of functions that show how the constructs of the Event-B
specification language can be mapped into our institution.  Our
semantics sheds new light on the structure of the Event-B language,
allowing us to clearly delineate three constituent sub-languages: the
superstructure, infrastructure and mathematical languages.  One of the
principal goals of our semantics is to provide access to the generic
modularisation constructs available in institutions, including
specification-building operators for parameterisation and refinement.
We demonstrate how these features subsume and enhance the corresponding
features already present in Event-B through a detailed study of their use in a
worked example.  We have implemented our approach via a parser and
translator for Event-B specifications, {\EBtoEVT}, which also provides a
gateway to the {\Hets} toolkit for heterogeneous specification.

\end{abstract}

\maketitle

\section{Introduction}%
\label{sec:intro}

When formal methods are used in software development, formal specification languages are used to model software systems and to provide a basis for their verification. Event-B is an industrial-strength formal specification language that  originally evolved from the B-Method~\cite{schneider2001b}, and has been used in a multitude of domains including air-traffic control systems and medical devices~\cite{abrial_modeling_2010}. Even though it is a mature formal specification language, Event-B has a number of drawbacks, particularly its lack of standardised modularisation constructs. In addition, while Event-B has been provided with a semantics in terms of proof obligations, the abstractness of this approach can make it difficult to deal formally with modularisation, or to provide a concrete basis for interoperability with other formalisms.

In this paper, we describe the formal constructions that allow the Event-B language to be embedded into the theory of institutions. Institutions are mathematical structures that make use of category theory and provide a generic framework for formalising a language or logic~\cite{goguen_institutions:_1992, sanella_foundations_2012}.   Having built an institution to support the core concepts in Event-B, we can then use this institution as a target for the semantics of the full Event-B specification language.  This in turn gives access to a range of generic modularisation constructs and provides a foundation for interoperability between Event-B and other formalisms.

\subsection{The structure of this paper}

Sections~\ref{sec:sem} and~\ref{sec:bg} review the background material and basic definitions used in this paper.
The semantics of programming languages is a well-understood concept that has been studied for many years and we use it as a starting point for our research. It has become evident from our investigations that defining a semantics of a specification language, such as Event-B in our case, is a slightly different endeavour.
We provide some background on the semantics of specification languages in Section~\ref{sec:sem}, and we give the basic definitions from the theory of institutions in this context.
In Section~\ref{sec:bg} we give an overview of the Event-B specification language, describe its abstract syntax and introduce an example that we will use throughout the paper.

The core of our semantics is presented in Sections~\ref{sec:iEVT} and~\ref{sec:transsemantics}.  In Section~\ref{sec:iEVT} we define our
institution for Event-B, {\iEVT}, that was originally developed in
our previous work~\cite{farrell_institution_2017}.
The {\iEVT}
institution provides the mathematical structures that will act as a target
for the semantics of constructs from the Event-B language. The proofs of the satisfaction relation and the comorphism between {\iFOPEQ} and {\iEVT} are contained in the paper. All other necessary proofs are contained in the appendix.

In Section~\ref{sec:transsemantics} we present a series of definitions of the  functions that map the abstract syntax of the Event-B language into the corresponding constructs in the {\iEVT} institution.  A key step here is
splitting the Event-B language into a three-layer model, by decomposing it into its superstructure, infrastructure and mathematical sub-languages.  As well as aiding the presentation, this also helps to expose the important features of the Event-B language.  In particular, once we have delineated these sub-languages, we can then employ the specification-building operators made available to us in the theory of institutions in our definition.

In the final sections of our paper, Sections~\ref{sec:impl} and~\ref{sec:refmod}, we explore some of the consequences of our semantics.
We have validated our semantics by building a Haskell translator tool that we call {\EBtoEVT}, and in Section~\ref{sec:impl} we describe the correspondence between this tool and our semantics. To illustrate the use of this tool in practice we apply it to our running example from Section~\ref{sec:ebexample}.

Formal refinement is core to the Event-B development process and, as such, any formalisation of Event-B should take this into account. In Section~\ref{sec:refmod} we discuss what refinement means in the theory of institutions and show how our formalisation of {\iEVT} supports Event-B refinement.  Again, this is illustrated via our running example. The weaknesses that Event-B has in terms of modularisation are also addressed in Section~\ref{sec:refmod} where we show how the generic specification-building operators from the theory of institutions can provide a set of well-defined modularisation constructs for Event-B.

One of the major benefits of these operators is that they are designed to be used in an institution-, and thus formalism-independent fashion. In this way, we provide scope for interoperability between Event-B and other formalisms that have been incorporated into the institution-theoretic framework. Our semantics and its associated {\EBtoEVT} tool are instrumental to this, since they allow us to bridge the gap between the {\Rodin} tool for Event-B and the {\Hets} tool-set used for institutions.

\subsection{Contributions of this paper}
The principal contributions of this paper are summarised as follows:

\begin{enumerate}
\item A formal semantics for the Event-B formal specification
  language, defined as a series of functions from Event-B constructs
  to specifications over {\iEVT}.  This provides clarity on the
  meaning of the language elements and their interaction, and allows
  us to distinguish their differing levels of responsibility using the
  three-layer model presented in Figure~\ref{fig: ebstruc}.

\item A well-defined set of modularisation constructs using the
  specification-building operators available through the theory of
  institutions.  These are built-in to our semantics, they subsume and
  extend the existing Event-B modularisation constructs, and they
  provide a standardised approach to exploring new modularisation
  possibilities.

\item An explication of Event-B refinement in the context of the
  {\iEVT} institution.  Refinement in {\iEVT} incorporates and extends
  the refinement constructs of Event-B, and our semantics provides a
  more general context for exploring other options.

\end{enumerate}

\noindent
In addition, we have implemented an {\EBtoEVT} translator tool that transforms
Event-B specifications that have been developed using the {\Rodin}
Platform into specifications over the {\iEVT} institution.  We have
used this to validate our semantic definitions, and to provide
interoperability with the existing large corpus of specifications
available in this format~\cite{farrell_clones_2017}.

\section{The Semantics of Specification Languages}%
\label{sec:sem}

The study of the semantics of specification languages is closely
related to the semantics of programming languages, since both are
species of formal languages and, in many cases, have closely-related
origins.  For example, the specification language at the heart of the
Vienna Development Method (VDM) emerged from a line of research on the
semantics of programming languages starting with PL/1 and Algol~\cite{jones-vdl-jucs-2001,AstarteJones-Algol-16}.  Providing a
semantics has many practical advantages, since the semantics can
facilitate an increased understanding of the language, can provide a
basis for standardisation documents (e.g., the ISO standard for Z~\cite{iso_z_2002}), can act as a specification for tools that
process the language, and can provide a basis for comparison and
interoperability with similar languages.

There are many approaches to defining a semantics for a
programming language, including denotational, operational and axiomatic
approaches among others~\cite{watt_programming_1991,gunter_semantics_1992,winskel_semantics_1993,goguen1996algebraic}.
A denotational semantics shows how to map the language constructs to
mathematical structures~\cite{scott1971toward,schmidt_denotational_1986}, and much of the
focus of this approach is on the correct selection of these
structures.  In contrast, the axiomatic approach characterises the
language in terms of the logical properties that must hold~\cite{hoare1969axiomatic}, and much of the focus here is on proof
obligations and theorem-proving.  While not exclusive, these perspectives can often be reflected in the specification language itself, with a contrast between \textit{model-oriented} and \textit{property-oriented} approaches~\cite{bjorner_se_2006}.

Some of the best-known model-oriented specification languages are Z
and VDM, both of which use typed first-order logic and set theory
overlaid with a set of language constructs to build specification units.
The Z specification language has been given a denotational semantics
by Spivey~\cite{spivey1988understanding}.  This semantics defines a
mapping from the basic syntactic unit of specification, a \textit{schema}, to a
\textit{variety}, which consists of a signature and set of models.  These models
map names used in the specification to appropriate values from the
underlying logic, such as sets, relations and functions.
In a similar vein, the semantics of VDM presented by Larsen and
Pawlowski~\cite{larsen-vdm-1995} shows how to map a \textit{document} to a
set of \textit{environments} that model it; these environments map names to
their denotations.

In order to specify the behaviour of software systems, many specification languages are based on models that include a concept of \textit{state}.
This has some similarity to the operational semantics used for programming
languages~\cite{hennessy_semantics_1990}, though the overall approach is still essentially denotational.
An early example was the COLD specification language, and its semantics mapped the basic unit of specification, a \textit{class description}, to a set of \textit{states} along with a \textit{transition relation} for each procedure~\cite{feijs-COLD-1994}.  The COLD specification-building operators are based on a class algebra, closely related to the structuring operators found in algebraic specification.  There are now many formalisms and languages which have semantics defined in terms of states and transitions including, for example, sequences of events in Object-Z~\cite{smith2012object,smith_oz_semantics_1995}, the \textit{trace}-based semantics of CSP~\cite{hoare1978communicating,roscoe1998theory} and the \textit{behaviours} of TLA+~\cite{lamporttla,merz2012logic}.

Along a similar vein, Unified Theories of
Programming (UTP) is a framework where specifications are interpreted
in terms of relations between their before and after states~\cite{hoare1998unified}. This work, which is rooted in lattice theory,
 somewhat parallels the approach using institutions, which is rooted in
category theory. More recent work uses UTP to define an algebra of programs and an associated refinement relation to derive denotational and operational models of the Guarded Command Language (GCL) and CSP~\cite{he2018roadmap}.

Translational approaches to formal semantics appear throughout the literature, in particular, the denotational UTP semantics for Circus (combination of CSP and Z) which translates the CSP component of the specification into Z~\cite{oliveira2009utp}. Another related translational approach is the semantics of a Dafny program which is given in terms of its translation to the Boogie intermediate verification language which generates the verification conditions to be proven with the Z3 SMT solver~\cite{leino2010dafny}.

The semantics of Clear, the first algebraic specification language, is a mix of denotational semantics and category-theoretic elements~\cite{burstall1980semantics}. Subsequently, the Common Algebraic Specification Language, CASL, was devised as a basis for future work on algebraic specifications using the unifying framework supplied by the theory of institutions~\cite{caslref_2004}. The semantics of CASL follows a similar appooach to Clear and is thus founded in the theory of institutions. In particular, CASL's two step semantics involves translating CASL constructs to their underlying mathematical concepts. In this way, the CASL static semantics determines signatures and its model semantics determines model-classes over the underlying institution. The Heterogeneous Toolset, \textsc{Hets}, provides tool support for CASL specifications and facilitates the use of heterogeneous specifications in an institution-theoretic fashion~\cite{grumberg_heterogeneous_2007}.

\subsection{Modularisation of specifications}

In his semantics for Z, Spivey notes that while ordinary mathematical reasoning is similar to
software specification, the latter has the additional problem of
scale.  In this context, dealing with modularisation, including naming
and structuring constructs, assumes increased importance.  In modular software engineering, a system is decomposed into multiple independent \textit{modules} that can be recomposed at a later stage to form the entire system~\cite{ghezzi2002fundamentals}. There are a multitude of benefits to this approach including ease of code maintenance, code reuse and enabling multiple engineers to work on different components of the same system in parallel.

Although intended to provide support for provably correct stepwise development, formal refinement has been often used as a modularisation technique in formal specification languages. In these languages, support for other modularisation features is typically added after the language has been defined~\cite{ehrig1992introduction,kesten1998modularization}.
The Z
specification language has a \textit{schema calculus}, whose structuring operations are deliberately designed to look like logical operators, but this does not have to be the case.
One of the purposes of re-engineering the Z specification language in the development of Object-Z was to enhance the support for modular developments~\cite{wordsworth1992software,smith2012object}. Similarly, the VVSL language was developed as an approach to modularisation for VDM~\cite{middelburg1989vvsl}.

When Burstall and Goguen presented the
semantics of the Clear specification language~\cite{burstall1980semantics}, one of their key observations was that
the specification-building operators could be given a meaning in a
very abstract way, almost independently of the notions of the
underlying formalism.  They characterised the meaning of a
Clear specification as an inter-connected collection of \textit{theories}, where a theory is a closed set of properties derivable from the sentences of the specification.  In this context, the semantics of the
specification-structuring operators become theory transformations,
which Burstall and Goguen describe as ``the mathematical essence of
Clear semantics''.

The Clear specification language shares many concepts with other
algebraic specification languages~\cite{bergstra_algebraic_1989,ehrig1992introduction}, most
notably culminating in the Common Algebraic Specification Language
(CASL)~\cite{caslref_2004} and the Heterogeneous Tool Set ({\Hets})~\cite{grumberg_heterogeneous_2007}.
These ideas were generalised in the theory of institutions, which describes these features in terms of category theory~\cite{goguen_institutions:_1992}. It has been observed that most of the formal languages used in Computer Science can be captured as institutions. One notable example is a category-theoretic formalisation of Z and its structuring mechanisms which captures the structure of modular Z specifications in a natural way~\cite{casto2015categorical}.
The semantics of Event-B that we present
in this paper follows a similar school of thought by
first defining an institution for Event-B, called {\iEVT}, in Section~\ref{sec:iEVT}
and then translating Event-B specifications
into structured specifications in the {\iEVT} institution in Section~\ref{sec:transsemantics}.

\subsection{Institutions: basic definitions}%
\label{sec:ins}

The theory of institutions, originally developed by Goguen and Burstall in a series of papers originating from their work on algebraic specification, provides a general framework for defining a logical system~\cite{goguen1983introducing,burstall1977putting,goguen_institutions:_1992}.  Institutions are generally used for representing logics in terms of their vocabulary (signatures), syntax (sentences) and semantics (models and satisfaction condition). The basic maxim of institutions, inspired by Tarski's theory of truth~\cite{tarski1944semantic}, is that
\begin{center}
\textit{``Truth is invariant under change of notation''}.
\end{center}

The use of institutions facilitates the development of specification language concepts, such as structuring of specifications, parameterisation, and refinement. This is achieved independently of the underlying logical system resulting in an environment for heterogeneous specification and combination of logics~\cite{sannella1988specifications, burstall1977putting}.
The theory of institutions is rooted in category theory and our work follows the notation and terminology used by Sannella and Tarlecki~\cite{sanella_foundations_2012}. Thus, we define
an institution as follows~\cite{goguen_institutions:_1992}:

\begin{defi}[\textbf{Institution}]%
\label{def:ins}
An \textbf{institution} $\mathcal{INS}$ for some given logic consists of:
\begin{itemize}
\item A category $\mathbf{Sign}_{\mathcal{INS}}$ whose objects are called \textit{signatures} and  whose arrows are called \textit{signature morphisms}.

\item A functor $\mathbf{Sen}_{\mathcal{INS}}: \mathbf{Sign}_{\mathcal{INS}} \rightarrow \mathbf{Set}$ yielding a set $\mathbf{Sen}_{\mathcal{INS}}(\Sigma)$ of $\Sigma$-\textit{sentences} for each signature $\Sigma \in \lvert \mathbf{Sign}_{\mathcal{INS}}\rvert$ and a function $\mathbf{Sen}_{\mathcal{INS}}(\sigma): \mathbf{Sen}_{\mathcal{INS}}(\Sigma)\rightarrow \mathbf{Sen}_{\mathcal{INS}}(\Sigma ')$ for each signature morphism $\sigma: \Sigma \rightarrow \Sigma'$.

\item A functor $\mathbf{Mod}_{\mathcal{INS}}: \mathbf{Sign}_{\mathcal{INS}}^{op} \rightarrow \mathbf{Cat}$ yielding a category $\mathbf{Mod}_{\mathcal{INS}}(\Sigma)$ of \textit{$\Sigma$-models} for each signature $\Sigma\in\lvert \mathbf{Sign}_{\mathcal{INS}}\rvert$ and a functor $\mathbf{Mod}_{\mathcal{INS}}(\sigma): \mathbf{Mod}_{\mathcal{INS}}(\Sigma')\rightarrow \mathbf{Mod}_{\mathcal{INS}}(\Sigma)$ for each signature morphism $\sigma: \Sigma \rightarrow \Sigma'$.

\item For every signature $\Sigma$, a \textit{satisfaction relation} $\models_{\mathcal{INS}, \Sigma}$ determining the satisfaction of $\Sigma$-sentences by $\Sigma$-models.

\item An institution must uphold the \textbf{satisfaction condition}:
for any signature morphism $\sigma:\Sigma \rightarrow \Sigma'$, $\phi \in \textbf{Sen}_{\mathcal{INS}}(\Sigma)$ and $M' \in \lvert\textbf{Mod}_{\mathcal{INS}}(\Sigma')\rvert$, the translations
$\textbf{Mod}_{\mathcal{INS}}(\sigma)$ of models and $\textbf{Sen}_{\mathcal{INS}}(\sigma)$ of
sentences preserve the following satisfaction condition
\begin{equation*}%
\label{eq:inssatcond}
M' \models _{\mathcal{INS}, \Sigma'} \textbf{Sen}_{\mathcal{INS}}(\sigma)(\phi)  \quad\iff\quad \textbf{Mod}_{\mathcal{INS}}(\sigma)(M') \models _{\mathcal{INS}, \Sigma} \phi
\end{equation*}
where $\mathbf{Sen}_{\mathcal{INS}}(\sigma)(\phi)$ (resp. $\mathbf{Mod}_{\mathcal{INS}}(\sigma)(M')$) indicates the translation of $\phi$ (resp.\ reduction of $M'$)  along the signature morphism $\sigma$.
\end{itemize}
\end{defi}
\noindent
Note that we often denote $\mathbf{Mod}_{\mathcal{INS}}(\sigma): \mathbf{Mod}_{\mathcal{INS}}(\Sigma') \rightarrow \mathbf{Mod}_{\mathcal{INS}}(\Sigma)$ by \[\_\rvert_{\sigma}: \mathbf{Mod}_{\mathcal{INS}}(\Sigma') \rightarrow \mathbf{Mod}_{\mathcal{INS}}(\Sigma)\] to indicate the model reduct along a signature morphism $\sigma: \Sigma \rightarrow \Sigma'$. Model reducts are central to the definition of an institution as they allow us to consider a model over one signature as a model over another via a signature morphism.

Institutions have been defined for many different types of logic~\cite[\S 4.1]{sanella_foundations_2012}.  One example is the institution of first-order predicate logic with equality, {\iFOPEQ}:

\begin{defi}[{\iFOPEQ}]%
\label{def:fopeq}
The institution for first-order predicate logic with equality, {\iFOPEQ}, is composed of the following~\cite[\S 4.1]{sanella_foundations_2012}~\cite{farrell_institution_2017,farrell_phdthesis}:

\begin{description}
\item[Signatures] Tuples of the form, $\Sigma_{{\iFOPEQ}} = \langle S, \Omega, \Pi \rangle$, where $S$ is a set of sort names, $\Omega$ is a set of operation names indexed by arity and sort, and $\Pi$ is a set of predicate names indexed by arity. Signature morphisms are sort/arity-preserving functions that rename sorts, operations and predicates.

\item[Sentences] For any $\Sigma_{{\iFOPEQ}} = \langle S, \Omega, \Pi \rangle$, $\Sigma_{{\iFOPEQ}}$-sentences are closed first-order formulae built out of atomic formulae using $\wedge, \vee, \lnot, \Rightarrow, {\iff}, \exists, \forall$. Atomic formulae are equalities between $\langle S, \Omega \rangle$-terms, predicate formulae of the form $p(t_1,\ldots, t_n)$ where $p \in \Pi$ and $t_1, \ldots, t_n$ are terms (with variables), and the logical constants \texttt{true} and \texttt{false}.

\item[Models] Given a signature $\Sigma_{\iFOPEQ} = \langle S, \Omega, \Pi \rangle$,
a model over {\iFOPEQ} consists of a carrier set $\lvert A\rvert _s$ for each sort
name $s \in S$, a function $f_A: \lvert A \rvert _{s_1} \times \dotsb \times
\lvert A\rvert _{s_n} \rightarrow \lvert A\rvert _s$ for each operation name $f \in
\Omega_{s_1 \dotsc s_n,s}$ and a relation $p_A \subseteq \lvert A\rvert _{s_1}
\times \dotsb \times \lvert A \rvert _{s_n}$ for each predicate name $p \in
\Pi_{s_1\dotsb s_n}$, where $s_1,\dotsc,s_n$, and $s$ are sort names.

\item[Satisfaction Relation] The satisfaction relation in {\iFOPEQ} is the usual satisfaction of
first-order sentences by first-order structures. This amounts to evaluating the truth of the {\iFOPEQ}-formula using the values in the carrier sets and the interpretation of operation and predicate names supplied by the models.

\end{description}

\end{defi}

\noindent
We will use {\iFOPEQ} in Section~\ref{sec:iEVT} as part of the
definition of our Event-B institution.

\subsection{Writing Specifications over an Institution}
Formally, a specification document written in a specification language is a set of sentences from that language, known as a \textit{presentation}, which is the basic unit of specification over an institution.

\begin{defi}[Presentation]%
\label{def:pres}
For any signature $\Sigma$, a $\Sigma$-presentation (sometimes called a \textit{flat specification}) is a pair $\langle \Sigma, \Phi \rangle$ where $\Phi \subseteq \textbf{Sen}(\Sigma)$. $M \in \lvert \textbf{Mod}(\Sigma)\rvert$ is a model of a $\Sigma$-presentation $\langle \Sigma, \Phi \rangle$  if $M \models \Phi$~\cite[\S 4.2]{sanella_foundations_2012}.
\end{defi}

\begin{table}
\centering
\begin{tabular}{ccp{9.2cm}}
\textbf{Operation} & \textbf{Format} & \textbf{Description}\\\toprule
Translation & $SP_1$ \ikw{with} $\sigma$ & Renames the signature components of $SP_1$ (e.g., sort, operation and predicate names in {\iFOPEQ}) using the signature morphism $\sigma: \Sigma_{SP_1} \rightarrow \Sigma'$. The resulting specification is the same as the original with the appropriate renamings carried out when $\sigma$ is bijective. In the case where $\sigma$ is not surjective then translation adds new symbols. If $\sigma$ is not injective then the signature morphism requires that the interpretations of the symbols identified by $\sigma$ coincide.
\begin{itemize}
\item[] $\mathit{Sig}[SP_1$ \ikw{with} $\sigma] = \Sigma'$
\item[] $\mathit{Mod}[SP_1$ \ikw{with} $\sigma]$

$= \{M' \in \lvert \textbf{Mod}(\Sigma')\rvert \ \rvert \ M'  \lvert_{\sigma} \in Mod[SP_1]\}$.
\end{itemize}\vspace{-10pt}\\
\midrule
Sum&$SP_1$ \ikw{and} $SP_2$ & Combines the specifications $SP_1$ and $SP_2$. It is the most straightforward way of combining specifications with different signatures. This is achieved by forming a specification with a signature corresponding to the union of the signatures of $SP_1$ and $SP_2$. Translation should be used before sum to avoid any unintended name clashes.
\begin{itemize}
\item[] $\mathit{SP_1}$ \ikw{and} $\mathit{SP_2} = (\mathit{SP_1}$ \ikw{with} $\iota)$ \ikw{$\cup$} $(\mathit{SP_2}$ \ikw{with} $\iota')$
\end{itemize}
where $\mathit{Sig}[SP_1] = \Sigma$,
$\mathit{Sig}[SP_2] = \Sigma'$,
$\iota: \Sigma \xhookrightarrow{} \Sigma \cup \Sigma'$, $\iota': \Sigma' \xhookrightarrow{} \Sigma \cup \Sigma'$ and \ikw{$\cup$} is applied to specifications ($SP_3$ and $SP_4$) over the same signature ($\Sigma''$) as follows
\begin{itemize}
\item[]  $\mathit{Sig}[SP_3$ \ikw{$\cup$} $SP_4] = \Sigma''$
\item[]  $\mathit{Mod}[SP_3$ \ikw{$\cup$} $SP_4] = \mathit{Mod}[SP_3] \cap \mathit{Mod}[SP_4]$.
\end{itemize}\vspace{-10pt}\\
\midrule
Enrichment&$SP_1$ \ikw{then} \ldots & Extends the specification $SP_1$ by adding new sentences after the \ikw{then} specification-building operator. This operator can be used to represent superposition refinement of Event-B specifications by adding new variables and events. The resulting signature and model construction is similar to \ikw{and}. \\
\midrule
Hiding&$SP_1$ \ikw{hide via} $\sigma$ & Hiding via the signature morphism $\sigma$ allows viewing a specification, $SP_1$, as a specification containing only the signature components of another specified by the signature morphism $\sigma: \Sigma \rightarrow \Sigma_{SP_1}$. This is useful for representing refinement as it allows a concrete specification to be viewed as one written using only those signature items supported by its corresponding abstract specification.
\begin{itemize}
\item[] $\mathit{Sig}[SP_1$ \ikw{hide via} $\sigma] = \Sigma$
\item[] $\mathit{Mod}[SP_1$ \ikw{hide via} $\sigma] = \{ \ M \lvert_{\sigma} \ \lvert \ M \in Mod[SP_1]\}$.
\end{itemize}\vspace{-10pt}\\
\end{tabular}
\caption{The basic institution-theoretic specification-building
  operators that can be used to modularise specifications in a
  formalism-independent manner. Here $SP_1$ and $SP_2$ denote
  specifications written over some institution, and $\sigma$ is a
  signature morphism in the same institution.}%
\label{table: sbo}
\end{table}

To support modularisation, \textit{structured specifications} are then built up from flat specifications using the specification-building operators that are available in the theory of institutions~\cite[\S 5]{sanella_foundations_2012}. These generic specification-building operators are summarised briefly in Table~\ref{table: sbo}, and we will demonstrate their use in Event-B in Section~\ref{sec:refmod}. Note that the intended semantics of a
  specification $SP$ is given by its signature, denoted $Sig[SP]$, and its
  class of models, $Mod[SP] \subseteq Mod(Sig[SP])$.

\subsection{Heterogeneous specification: connecting institutions}
Just as we can combine flat and structured specifications within an institution using the specification-building operators, we can combine specifications from different specification languages by defining a mapping between their institutions. This supports the software engineering notion of separation of concerns, allowing us to model different aspects of a specification using different formalisms, each of which is more suited to specifying particular aspects of the system.
An institution \textit{comorphism} allows the embedding of a primitive institution into a more complex one, and we will use it in Section~\ref{sec:comor} to define a relationship between {\iFOPEQ} and {\iEVT}.

\begin{defi}[\textbf{Institution Comorphism}]%
\label{def:inscomor}
An institution comorphism $\rho:
\mathbf{INS}\rightarrow\mathbf{INS'}$  has three components:
\begin{itemize}
\item A functor $\rho^{Sign}: \mathbf{Sign} \rightarrow
  \mathbf{Sign'}$.
\item A natural transformation $\rho ^{Sen}:
\mathbf{Sen}\rightarrow \rho ^{Sign}; \mathbf{Sen'}$, that is,
for each $\Sigma \in \lvert\mathbf{Sign}\rvert$, a function
$\rho^{Sen}_{\Sigma}: \mathbf{Sen}(\Sigma) \rightarrow
\mathbf{Sen'}(\rho^{Sign}(\Sigma))$.
\item A natural transformation $\rho^{Mod}:(\rho^{Sign})^{op};\mathbf{Mod'}\rightarrow \mathbf{Mod} $, that is, for each $\Sigma \in \lvert\mathbf{Sign}\rvert$, a functor $\rho^{Mod}_{\Sigma}:\mathbf{Mod'}(\rho^{Sign}(\Sigma))\rightarrow \mathbf{Mod}(\Sigma)$.
\end{itemize}

\noindent Further, an institution comorphism must ensure that for any signature $\Sigma
\in \lvert\mathbf{Sign}\rvert$, the translations $\rho^{Sen}_{\Sigma}$
of sentences and $\rho^{Mod}_{\Sigma}$ of models preserve the
satisfaction relation, that is, for any $\psi \in
\mathbf{Sen}(\Sigma)$ and $M' \in
\lvert\mathbf{Mod}(\rho^{Sign}(\Sigma))\rvert$ we have
\[
\rho^{Mod}_{\Sigma}(M') \models_{\Sigma} \psi
\quad\iff\quad
M' \models '_{\rho^{Sign}(\Sigma)}\rho^{Sen}_{\Sigma}(\psi)
\]
and the following diagrams in $\mathbf{Sen}$ and $\mathbf{Mod}$ commute
for each signature morphism in $\mathbf{Sign}$~\cite[\S 10.4]{sanella_foundations_2012}:
\[
\begin{tikzcd}
\Sigma_1 \arrow[d, "\sigma "]\\ \Sigma_2
\end{tikzcd}
\;\;
\begin{tikzcd}
\mathbf{Sen}(\Sigma_1) \arrow[r, "\rho^{Sen}_{\Sigma_1} "] \arrow[d, "\mathbf{Sen}(\sigma)"']
& \mathbf{Sen'}(\rho^{Sign}(\Sigma_1)) \arrow[d, "\mathbf{Sen'}(\rho^{Sign}(\sigma))" ] \\
\mathbf{Sen}(\Sigma_2) \arrow[r, "\rho^{Sen}_{\Sigma_2} "']
& \mathbf{Sen'}(\rho^{Sign}(\Sigma_2))
\end{tikzcd}
\;\;
\begin{tikzcd}
\mathbf{Mod'}(\rho^{Sign}(\Sigma_2)) \arrow[r, "\rho^{Mod}_{\Sigma_2} "] \arrow[d, "\mathbf{Mod'}(\rho^{Sign}(\sigma))"']
& \mathbf{Mod}(\Sigma_2) \arrow[d, "\mathbf{Mod}(\sigma)" ] \\
\mathbf{Mod'}(\rho^{Sign}(\Sigma_2)) \arrow[r, "\rho^{Mod}_{\Sigma_1} "']
& \mathbf{Mod}(\Sigma_1)
\end{tikzcd}
\]
\end{defi}

In summary, this section has outlined the basic definitions that will
be needed to understand our construction of the Event-B institution
that will be presented in Section~\ref{sec:iEVT}.

\section{An Overview of the Event-B Specification Language}%
\label{sec:bg}

In this section, we describe the prerequisite background material that is pertinent to this paper. In particular, Section~\ref{sec:eb} briefly outlines the Event-B formal specification language and  Section~\ref{sec:ebexample} introduces a simple example that we will refer to throughout this document. Finally, Section~\ref{sec:ebrelsem} discusses existing approaches to defining a semantics for Event-B.

\subsection{The Syntax for Event-B}%
\label{sec:eb}

\begin{figure}
\input eb_abstract_syntax
\caption{The Event-B syntax, consisting of the definition of machines,
  contexts and events.  This syntax does not define the
  non-terminals \textit{Predicate} and \textit{Expression}, since these
  will be mapped to {\iFOPEQ}-formulae and terms respectively in our
   semantics.}%
\label{fig1: syntax}
\end{figure}
Event-B evolved from the B Method, often termed Classical-B, as a formal method for the  specification and verification of systems using a set-theoretic notation~\cite{schneider2001b}.
The abstract syntax for Event-B is described briefly in~\cite{abrial_modeling_2010} and we provide a more detailed version in Figure~\ref{fig1: syntax}.   An Event-B specification consists of a set of \textit{machine} and \textit{context} definitions.
\begin{itemize}
\item \textit{Machines} encapsulate the dynamic behaviour of a system in terms of a state space, represented by a set of variables and constrained by the machine-level invariants.  A machine contains a set of events, which can be parameterised by values, and are further constrained by the  predicates in the guards, witnesses and actions. Guards are predicates representing the conditions under which a given event may be triggered. Witnesses are used to relate abstract and refined event parameters. Actions are before-after predicates describing the updates to the state that are caused by the corresponding event.

Events are also equipped with a status which is either \texttt{ordinary}, \texttt{anticipated} or \texttt{convergent} and this is used for proving termination properties with respect to a variant expression defined in the machine. Specifically, \texttt{anticipated} events must not increase the variant expression, whilst \texttt{convergent} events must strictly decrease the variant expression.

\item \textit{Contexts} model the static behaviour of a system using sets and constants which can be constrained by axioms.
\end{itemize}

\noindent
Both machines and contexts allow the user to specify \textit{theorems} which are used to generate a specific set of proof obligations. Since these must be consequences of the specification and do not add any additional constraints, we omit them from further discussion here.

In Figure~\ref{fig1: syntax} we do not give a definition for the non-terminals \textit{predicate} and \textit{expression}, which we intend to be understood as representing the corresponding concepts from first-order logic.  We will ultimately delegate the semantics of these constructs in Section~\ref{sec:interface} to {\iFOPEQ}, the institution for first-order predicate logic with equality.
Thus we immediately distinguish two basic sub-languages in Event-B:\@ the
Event-B \textit{mathematical language}, consisting of predicate logic, set
theory and arithmetic, and the Event-B \textit{modelling language} which uses these in the definition of events, machines and contexts.

In the spirit of the institutional approach, we can further decompose the Event-B modelling language into two sub-languages which (inspired by an early
version of the UML specification~\cite{uml_infrastructure_2011,uml_superstructure_2011}) we call the
\textit{infrastructure language} and the \textit{superstructure
language}. The infrastructure language covers those concepts which are internal to a given context or machine, while the superstructure
language includes the specification-building operators in Event-B.
These constructs allow one machine to \textit{refine} another, or to \textit{see} a context, and allow a context to \textit{extend} other contexts.
It is our thesis that these Event-B operators can be represented and substantially extended using the generic specification-building operators that are part of the institutional framework, and we show how this can be achieved in Section~\ref{sec:super}.

In summary, we divide the constructs of the Event-B language
into three layers, each layer corresponding to one of its three
constituent sub-languages, as illustrated in Figure~\ref{fig: ebstruc}.
The right hand side of Figure~\ref{fig: ebstruc} shows the
institutional constructs that we will use to define the semantics of each of these constituent sub-languages. This three-layer model plays a key role in structuring the definitions of the functions given in Section~\ref{sec:transsemantics}.

\begin{figure}
\hspace{-15pt}
\begin{tikzpicture}[block/.style={
  text width=0.3*\columnwidth
  }]
\matrix [column sep=2mm, row sep=9mm] {
  \node[block, align = center] (super) [] {Modelling Superstructure}; &
  \node[block, align = center] (supc) [draw, shape=rectangle] {refines, sees}; &
  \node[block, align = center] (supsbo) [] {{\iEVT} specification-building operators}; \\
  \node[block, align = center] (ebinf) [] {Modelling Infrastructure}; &
  \node[block, align = center](ebinfc) [draw, shape=rectangle] {variables, invariants,\\ variants, events}; &
  \node[block, align = center](evtsen) [] {{\iEVT}-sentences}; \\
  \node[block, align = center](fopeq) [] {Mathematical Language}; &
  \node[block, align = center](fopeqc) [draw, shape=rectangle] {carrier sets, constants,\\ axioms, extends}; &
  \node[block, align = center] (fopeqsbo) [] {{\iFOPEQ}-sentences and specification-building operators}; \\
};
\draw[->, thick] (ebinfc) -- (supc);
\draw[->, thick] (fopeqc) -- (ebinfc);

\end{tikzpicture}
\caption{We split the Event-B syntax, defined in Figure~\ref{fig1:
    syntax}, into three sub-languages: superstructure, infrastructure
  and a mathematical language.  This figure shows the corresponding
  Event-B constructs, and their representation in our semantics.}%
\label{fig: ebstruc}
\end{figure}

\subsection{A running example illustrating Event-B specification and refinement}%
\label{sec:ebexample}

We illustrate the use of Event-B via a running example of \textit{cars on a bridge} that was originally presented in Abrial's book on modelling in Event-B~\cite[Ch.\ 2]{abrial_modeling_2010}. Figure~\ref{fig:ebm0} contains an Event-B specification corresponding to the first abstract machine in this development. It describes the behaviour of cars entering and leaving the mainland. This abstract model is comprised of a context (lines 1--7) specifying a natural number constant $d$ representing the total number of cars, and a basic machine description (lines 8--29) which has a counter $n$ for the number of cars on the island or on the bridge, along with events to initialise the number of cars, and to increase or decrease the counter as cars leave or enter. In this example, the bridge is only wide enough for one car so it can only be used by cars travelling in one direction at any given time.

\begin{figure}
\begin{programsc}
\CONTEXT{cd}
  \CONSTANTS
    d
  \AXIOMS
    \bLabel{axm1}{d\in\mathbb{N}}
    \bLabel{axm2}{d>0}
\END

\MACHINE{m0}
  \SEES{cd}
  \VARIABLES
    n
  \INVARIANTS
    \bLabel{inv1}{n\in\mathbb{N}}
    \bLabel{inv2}{n \leq d}
  \EVENTS
    \INITIALISATION
      \bkw{then}
        \bLabel{act1}{n := 0}
    \EVT{ML\_out}{ordinary}
      \bkw{when}
        \bLabel{grd1}{n<d}
      \bkw{then}
        \bLabel{act1}{n := n+1}
    \EVT{ML\_in}{ordinary}
      \bkw{when}
        \bLabel{grd1}{n>0}
      \bkw{then}
        \bLabel{act1}{n := n-1}
\END
\end{programsc}
\caption{An Event-B specification of an abstract machine \texttt{m$0$}
  and context \texttt{cd}. This machine specification consists of the events
  \texttt{ML\_out} and \texttt{ML\_in} that model the behaviour of
  cars leaving and entering the mainland respectively.}%
\label{fig:ebm0}

\end{figure}

\begin{figure}

\begin{programsc}
\MACHINE{m1}
  \REFINES{m0}
  \SEES{cd}
  \VARIABLES
    a, b, c
  \INVARIANTS
    \bLabel{inv1}{a\in\mathbb{N}}{//\textit{on bridge to island}}
    \bLabel{inv2}{b\in\mathbb{N}}{//\textit{on island}}
    \bLabel{inv3}{c\in\mathbb{N}}{//\textit{on bridge to mainland}}
    \bLabel{inv4}{n = a+b+c}
    \bLabel{inv5}{a=0 \lor c=0}
    \bLabel{thm1}{a+b+c \in \mathbb{N}} theorem
    \bLabel{thm2}{c>0 \lor a>0 }
           $\lor (a+b < d \land c=0)$
           $\lor (0<b \land a=0)$ theorem
  \VARIANT 2*a + b
  \EVENTS
    \INITIALISATION
      \bkw{then}
        \bLabel{act2}{a := 0}
        \bLabel{act3}{b := 0}
        \bLabel{act4}{c := 0}
    \EVT{ML\_out}{ordinary}
      \REFINES{ML\_out}
      \bkw{when}
        \bLabel{grd1}{a+b<d}
        \bLabel{grd2}{c=0}
      \bkw{then}
        \bLabel{act1}{a := a+1}
    \EVT{IL\_in}{convergent}
      \bkw{when}
        \bLabel{grd1}{a>0}
      \bkw{then}
        \bLabel{act1}{a := a-1}
        \bLabel{act2}{b := b+1}
    \EVT{IL\_out}{convergent}
      \bkw{when}
        \bLabel{grd1}{0<b}
        \bLabel{grd2}{a=0}
      \bkw{then}
        \bLabel{act1}{b := b-1}
        \bLabel{act2}{c := c+1}
    \EVT{ML\_in}{ordinary}
      \REFINES{ML\_in}
      \bkw{when}
        \bLabel{grd1}{c>0}
      \bkw{then}
        \bLabel{act2}{c := c-1}
\END
\end{programsc}

\caption{An Event-B specification of a machine \texttt{m1} with
  additional events \texttt{IL\_in} and \texttt{IL\_out} to model the
  behaviour of cars entering and leaving the island. The variables
  \texttt{a}, \texttt{b}, and \texttt{c} keep track of the number of
  cars on the bridge going to the island, the number of cars on the
  island and the number of cars on the bridge going to the mainland,
  respectively.}%
\label{fig:ebm1}
\end{figure}

\paragraph{Refinement.}
Event-B supports the process of refinement allowing the system specification to be gradually constructed from an abstract specification into a concrete specification~\cite{abrial_refinement_2007}.
The first refinement step is represented by giving a specification of a concrete machine \texttt{m1}, shown in Figure~\ref{fig:ebm1}, which is declared to refine \texttt{m0} (line 2).  New variables are added to record the number of cars on the bridge going to the island, the number of cars on the island and the number of cars on the bridge going to the mainland (\texttt{a}, \texttt{b} and \texttt{c} respectively on line 5) and the invariants on lines 7--15 ensure that these function as expected.

The abstract machine, \texttt{m0}, is refined by \texttt{m1} in two ways.  First, new events are added to describe cars entering and leaving the island (\texttt{IL\_out} and \texttt{IL\_in} on lines 30--42). Both of these events are defined to be \texttt{convergent}, meaning that they must decrease the value of variant expression defined on line 16.  Second, the events \texttt{ML\_in} and \texttt{ML\_out} from \texttt{m0} are refined to utilise the new variables, \texttt{a}, \texttt{b} and \texttt{c}; their relationship with the abstract events from \texttt{m0} is noted explicitly on lines 24 and 44.

\subsection{Existing approaches to Event-B semantics}%
\label{sec:ebrelsem}
As noted earlier, the only semantic definitions that exist for Event-B have been given in terms of proof obligations which, due to their abstract nature, can inhibit the formal definition of modularisation constructs and interoperability. We discuss these approaches below.

Proof support for Event-B is provided by its Eclipse-based IDE, the {\Rodin} Platform~\cite{abrial2010rodin}. For each of the specification components listed in Figure~\ref{fig1: syntax}, the {\Rodin} Platform  generates a series of proof obligations which can then be proved automatically and/or interactively using its proving interface.
{\Rodin} supports refinement by reusing proofs carried out in an abstract model to prove properties about the concrete model.

This notion of associating a specification with the set of proof obligations it generates can be seen as a form of semantics for those specifications, and this is the approach taken by Hallerstede in his semantics for Event-B~\cite{hallerstede_purpose_2008}.  Note that this is not a property-oriented specification, since proof obligations should be theorems of the underlying mathematical context, not axioms that could be used to specify a set of models or define a theory. In essence, the proof obligations are used to characterise rather than fix the semantics of Event-B models. By not associating Event-B with any behavioural model, this approach has the advantage of making an Event-B specification independent of any model of computation.  Hallerstede asserts that this facilitates interoperability, since there is no predetermined semantic model to complicate the mapping.

However, this approach also has some disadvantages.  For larger specifications it offers no explicit name-management, and has no formal concept of modularisation.  For example, the semantics of an Event-B machine at the end of a long refinement chain would be a large series of proof obligations collected along the way from different machines and contexts, and with no explicit mention of named events.
In a previous empirical study of a large corpus of Event-B specifications we have found a tendency for monolithic development, with repeated sub-specifications analogous to software code clones~\cite{farrell_clones_2017}.
We hypothesise that a more fully developed set of specification-building operators would facilitate the construction of large Event-B specifications in a more modular manner.

Abrial presented a model of Event-B machines in terms of states (variable-to-value mappings) and traces~\cite[Ch.\ 14]{abrial_modeling_2010}.  While Abrial intends these as a justification for the proof obligations, we have used a version of these models in our construction of the Event-B institution in Section~\ref{sec:evtmod}, and we thus define the semantics in this context.
Since the resulting proof obligations will then be part of the theory associated with the model, we anticipate that our approach will be complementary to that of Hallerstede.

There have been other approaches to defining a semantics for elements of the Event-B language, principally in the context of mapping it to other formalisms.  It has been shown that a shallow embedding of the Event-B logic in Isabelle/HOL  provides a sound semantics for the mathematical language of Event-B but not for the full modelling language~\cite{schmalz_logic_2010}. Moreover, Event-B is a special form of Back's Action systems which has been formally defined using Dijkstra's weakest precondition predicate transformers~\cite{back1990refinement}. This accounts for the refinement calculus used but not for a full semantic definition of the Event-B modelling language itself.

Other related work includes a CSP semantics for Event-B refinement~\cite{schneider_behavioural_2014}. This behavioural semantics defines the relationship between the first and last machines in a refinement chain but does not give a semantics for the internal components of these machines.
In previous work, we presented an institution-theoretic approach to  interoperability between Event-B (using the {\iEVT} institution) and CSP (using the {\iCSPCASL} institution)~\cite{farrell_csp_2017}. This work supported heterogeneous specification by constraining the models of {\iEVT}-specifications by the models of {\iCSPCASL} ones.

The approach presented in this paper treats Event-B as a complete language rather than as a calculus whose semantics can be given in terms of proof obligations.  Our semantics commits to the theory of institutions, and gets in return a standard set of specification-building operators, along with a generic  framework for interoperability with other formalisms, albeit at the price of including them in the institutional framework.

\section{\texorpdfstring{${\iEVT}$}{EVT}: The Institution for Event-B}%
\label{sec:iEVT}

This section introduces the {\iEVT} institution, which will play the role of the semantic domains for our semantics: that is, our semantics will map Event-B specifications into structured specifications in the {\iEVT} institution.  The definitions in this section then ensure that these structured specifications have a meaningful semantics.

To define the {\iEVT} institution, Sections~\ref{sec:evt} through~\ref{sec:evtsat} give definitions for its basic elements corresponding to Definition~\ref{def:ins}, namely
$\mathbf{Sign_{{\iEVT}}}$, $\mathbf{Sen_{{\iEVT}}}$, $\mathbf{Mod_{{\iEVT}}}$ and the satisfaction relation, $\models_{{\iEVT}}$.  We then prove that the institution-theoretic satisfaction condition holds for {\iEVT}.

A structured specification in {\iEVT} represents an Event-B machine.
We handle the Event-B contexts using the institution for first-order predicate logic with equality ({\iFOPEQ} in Definition~\ref{def:fopeq}), and in Section~\ref{sec:comor} we define an institution comorphism from {\iFOPEQ} to {\iEVT} to enable this.   Finally, in Section~\ref{sec:push} we show that {\iEVT} exhibits good behaviour with respect to modularisation by proving properties relating to pushouts and amalgamation in {\iEVT}.  This allows us to use the generic specification-building operators of institutions in {\iEVT}, and we discuss the pragmatics of this in Section~\ref{sec:prag}.

This work was originally presented in~\cite{farrell_institution_2017} and it forms the basis of our definition of a translational semantics in the subsequent sections. 

\subsection{\texorpdfstring{$\mathbf{Sign}_{\iEVT}$}{Sign-EVT}, the Category of \texorpdfstring{$\iEVT$}{EVT}-signatures}%
\label{sec:evt}
Signatures describe the vocabulary that specifications written in a particular institution can have. Here, we define what it means to be an {\iEVT}-signature (Definition~\ref{def:evtsig}) and an {\iEVT}-signature morphism (Definition~\ref{def:evtsigmor}). We show that $\mathbf{Sign}_{{\iEVT}}$ is indeed a category (Lemma~\ref{lem:signevt}) which is a necessary requirement for the definition of the {\iEVT} institution.
\begin{defi}[\textbf{{\iEVT}-Signature}]%
\label{def:evtsig}
A \textbf{signature} in {\iEVT} is a five-tuple of sets
$\Sigma_{{\iEVT}} = \langle S, \Omega, \Pi, E, V \rangle$ where
\begin{itemize}
\item The first three components form a standard {\iFOPEQ}-signature, $\langle S, \Omega, \Pi \rangle$. 
\item $E$ is a function from event names to their status. We write elements of
  this function as tuples of the form $(e\mapsto s)$, for an event named $e$
  with a status $s$ which always belongs to the
  poset $\{\texttt{ordinary} < \texttt{anticipated} < \texttt{convergent}\}$.
\item $V$ is a set of sort-indexed variables.  We write elements of
  this set as tuples of the form $(v:s)$, for a variable named $v$
  with sort $s \in S$.

  In theory it is possible for two variables of different sorts to
  share the same name, but we advise against this as it may lead to
  confusing specifications.  We assume that $V$ does not
  contain any \emph{primed} variable names and that such names are not
  introduced via any signature morphisms.
\end{itemize}
Following the convention for Event-B~\cite{abrial_modeling_2010}, we
assume that every signature has an initial event, called
\texttt{Init}, whose status is always \texttt{ordinary}.

\end{defi}

\noindent\emph{Notation:}
We write $\Sigma$ in place of $\Sigma_{{\iEVT}}$ when describing a
signature over our institution for Event-B. For signatures
over institutions other than {\iEVT} we will use the subscript
notation where necessary; e.g., a signature over {\iFOPEQ} is denoted by
$\Sigma_{{\iFOPEQ}}$.  For a given signature $\Sigma$, we access its
individual components using a dot-notation, e.g.\ $\Sigma.S$ for the
set $S$ in the tuple $\Sigma$.
For any signature, $\Sigma.E$ is a function from event names to their
status, and we will access the domain of this function, $dom(\Sigma.E)$ when
we wish to deal with event names only.

\begin{defi}[\textbf{{\iEVT}-Signature Morphism}]%
\label{def:evtsigmor}
A \textbf{signature morphism} $\sigma: \Sigma \rightarrow \Sigma'$ is a five-tuple containing $\sigma_S$, $\sigma_\Omega$, $\sigma_\Pi$, $\sigma_E$ and $\sigma_V$.  Here $\sigma_S$, $\sigma_\Omega$, $\sigma_\Pi$ are the mappings taken from the corresponding signature morphism in {\iFOPEQ}, as follows 
\begin{itemize}
\item $\sigma_S: \Sigma.S \rightarrow \Sigma'.S$ is a function mapping sort names to sort names.
\item $\sigma_\Omega: \Sigma.\Omega \rightarrow \Sigma'.\Omega$ is a family of functions mapping operation names in $\Sigma.\Omega$, respecting the arities and result sorts.
\item $\sigma_\Pi: \Sigma.\Pi \rightarrow \Sigma'.\Pi$ is a family of functions mapping the predicate names in $\Sigma.\Pi$, respecting the arities.
\end{itemize}
and
\begin{itemize}
\item $\sigma_E: \Sigma.E \rightarrow \Sigma'.E$ is a function on
  event names and their status with the condition that
  for the initial event $\sigma_E ( \texttt{Init}\mapsto
  \texttt{ordinary} ) =  (\texttt{Init}\mapsto \texttt{ordinary} )$, while
  for non-initial events $\sigma_E$ preserves the status ordering given by the poset
  $\{\mathtt{ordinary} < \mathtt{anticipated} <
  \mathtt{convergent}\}$.

\item $\sigma_V: \Sigma.V \rightarrow \Sigma'.V$ is a function on
  sets of variable names whose sorts change according to $\sigma_S$. 
  Thus if some variable $v$ is renamed to $w$, we must have that
$\sigma(v:s) = (w: \sigma_S(s))$,
  where  $\sigma_S$ is the sort mapping as described above.

\end{itemize}
\end{defi}
\begin{restatable}{lem}{signevt}%
\label{lem:signevt}
{\iEVT}-signatures and signature morphisms define a category $\mathbf{Sign}_{\iEVT}$. The objects are signatures and the arrows are signature morphisms. Composition of arrows is performed componentwise.
\end{restatable}

In this section, we have defined the category of {\iEVT}-signatures, $\mathbf{Sign}_{{\iEVT}}$. Next, we define the functor $\mathbf{Sen}_{{\iEVT}}$ that yields {\iEVT}-sentences.

\subsection{The Functor \texorpdfstring{$\mathbf{Sen}_{\iEVT}$}{Sen-EVT}, yielding \texorpdfstring{${\iEVT}$}{EVT}-sentences}%
\label{sec:evtsen}
The second component of an institution is a functor called $\mathbf{Sen}$ that generates a set of sentences over a particular signature (Definition~\ref{def:ins}).
Here, we define what is meant by a $\Sigma_{\iEVT}$-Sentence (Definition~\ref{def:evtsen}) and we prove that $\mathbf{Sen}_{\iEVT}$ is indeed a functor (Lemma~\ref{lem:sen}).

\begin{defi}[\textbf{$\Sigma_{\iEVT}$-Sentence}]%
\label{def:evtsen}

Sentences over {\iEVT} are of the form $\langle e, \phi(\overline{x}, \overline{x}\prime) \rangle$. Here, $e$ is an event name in the domain of $\Sigma.E$ (recall that elements of $\Sigma.E$ are functions containing tuples of the form $(e \mapsto \texttt{status})$) and $\phi(\overline{x}, \overline{x}\prime)$ is an open {\iFOPEQ}-formula over these variables $\overline{x}$ from $\Sigma.V$ and the primed versions, $\overline{x}\prime$, of the variables.


\end{defi}

In the {\Rodin} Platform, Event-B machines are presented (syntactically sugared) as can be seen in abstract form on the left of Figure~\ref{fig:abseb}. The set of variables in a machine is denoted by $\overline{x}$ (line 2) and the invariants are denoted by $I(\overline{x})$ on line 3. The {\iEVT}-signature derived from this machine would look like: \[\Sigma = \langle S, \Omega, \Pi, \{(
\mathtt{Init}\mapsto \mathtt{ordinary}),  \ldots, (\mathtt{e}_i \mapsto \mathtt{convergent}), \ldots\}, \{\overline{x}\}\rangle\] where the {\iFOPEQ}-component of the signature, $\langle S, \Omega, \Pi \rangle$, is drawn from the ``seen'' context(s) on line 1.
\begin{figure}

\begin{minipage}{0.47\textwidth}
\begin{programs}
\MACHINE{m} \REFINES{a} \SEES{ctx}
  \VARIABLES $\overline{x}$
  \INVARIANTS $I(\overline{x})$
  \VARIANT $n(\overline{x})$
  \EVENTS
  \INITIALISATION \texttt{ordinary}
    \bkw{then} \bLabel{act-name}{BA(\overline{x}\sp{\prime})}
   \vdots
  \EVT{$e_i$} \texttt{convergent}
    \bkw{any} $\overline{p}$
    \bkw{when} \bLabel{guard-name}{G(\overline{x},\overline{p})}
    \bkw{with} \bLabel{witness-name}{W(\overline{x},\overline{p})}
    \bkw{then} \bLabel{act-name}{BA(\overline{x},\overline{p},\overline{x}\sp{\prime})}
   \vdots
\bkw{END}
\end{programs}

\end{minipage}
\quad\vline\quad
\begin{minipage}{0.47\textwidth}
\begin{programs}*
$\{\langle\text{$e$}, I(\overline{x}) \land I(\overline{x}\prime)\rangle \mid e \in dom(\Sigma.E)\}$

$\langle\texttt{Init}, BA(\overline{x}\sp{\prime})\rangle$

$\langle\text{e}_i, n(\overline{x}\prime) < n(\overline{x})\rangle$

{$\langle\text{e}, \exists\overline{p} \cdot G(\overline{x},\overline{p}) \land W(\overline{x},\overline{p}) \land BA(\overline{x},\overline{p},\overline{x}\sp{\prime})\rangle$}

\end{programs}
\end{minipage}
\caption{The elements of an Event-B machine specification as presented
  in {\Rodin} (left) and the corresponding sentences in the {\iEVT}
  institution (right).  The {\iEVT}-signature is not shown here, but would
  have to be constructed before the sentences.}%
\label{fig:abseb}
\end{figure}
The variant expression, denoted by $n(\overline{x})$ on line 4, is used for proving termination properties. As described in Section~\ref{sec:eb}, events that have a status of \texttt{anticipated} or \texttt{convergent} must not increase or must strictly decrease the variant expression respectively~\cite{abrial_modeling_2010}. Each machine has an \texttt{Initialisation} event (lines 6--7) whose action is interpreted as a predicate $BA(\overline{x})$. Events can have parameter(s) as given by the list of identifiers $\overline{p}$ on line 10. $G(\overline{x}, \overline{p})$ and  $W(\overline{x},\overline{p})$ are predicates that
represent the guard(s) and witness(es) respectively over the variables and parameter(s) (lines 11--12). Actions are interpreted as before-after predicates i.e. \texttt{x := x + 1} is interpreted as \texttt{x$\prime$ = x + 1}. Thus, the predicate $BA(\overline{x},\overline{p},\overline{x}\sp{\prime})$ on line 13 represents the action(s) over the parameter(s) $\overline{p}$ and the sets of variables $\overline{x}$ and $\overline{x}\sp{\prime}$. 

Formulae written in the mathematical language (such as the axioms that may appear in contexts) are interpreted as sentences over {\iFOPEQ}.  We can include these in specifications over {\iEVT} using the comorphism defined in Section~\ref{sec:comor}. We represent the Event-B invariant, variant and event predicates as sentences over {\iEVT}. These are illustrated alongside the Event-B predicates that they correspond to in Figure~\ref{fig:abseb}.
\begin{description}
\item[\textbf{Invariants}] For each Event-B invariant, $I(\overline{x})$, we form the
open {\iFOPEQ}-sentence $I(\overline{x}) \land
I(\overline{x}\prime)$ and pair it with each event name in $\Sigma$, since invariants apply to all events in a machine.
Thus, we form the set of {\iEVT} event sentences $\{\langle e, I(\overline{x}) \land I(\overline{x}\prime)\rangle \mid e \in dom(\Sigma.E)\}$.

\item[\textbf{Variants}] Firstly, we assume the existence of a suitable type for variant expressions and the usual interpretation of the predicates $<$ and $\leq$ over the integers in the signature. In particular, we assume that all specifications include a standard built-in library that defines mathematical operators, such as $<$ and $\leq$, which are axiomatised consistently with their Event-B counterparts. The variant expression applies to specific events, so we pair it with
an event name in order to meaningfully evaluate it. This
expression can be translated into an open {\iFOPEQ}-term, which we
denote by $n(\overline{x})$, and we use this to construct a formula
based on the status of the event(s) in the signature $\Sigma$.
\begin{itemize}
\item For each $( e \mapsto\texttt{anticipated}) \in \Sigma.E$
  we form the {\iEVT}-sentence $\langle e, n(\overline{x}\prime) \leq n(\overline{x})\rangle$
\item For each $( e \mapsto \texttt{convergent}) \in \Sigma.E$
  we form the {\iEVT}-sentence $\langle e, n(\overline{x}\prime) < n(\overline{x}) \rangle$
\end{itemize}

\item[\textbf{Events}] Event guard(s) and witnesses are also labelled predicates that can be
translated into open {\iFOPEQ}-formulae over the variables
$\overline{x}$ in $V$ and parameters $\overline{p}$. These are denoted
by $G(\overline{x},\overline{p})$ and $W(\overline{x},\overline{p})$
respectively.  In Event-B, actions are interpreted as before-after
predicates, thus they can be translated into open
{\iFOPEQ}-formulae denoted by $BA(\overline{x},\overline{p},
\overline{x}\prime)$. Thus for each event we build the formula
\[ \phi(\overline{x}, \overline{x}\prime) \quad=\quad \exists
\overline{p} \cdot G(\overline{x},\overline{p}) \land
W(\overline{x},\overline{p}) \land BA (\overline{x}, \overline{p},
\overline{x}\prime)\] where $\overline{p}$ are the event parameters. This
is used to build an {\iEVT}-sentence of the form $\langle e,
\phi(\overline{x}, \overline{x}\prime)\rangle$.

The
\texttt{Init} event, which is an Event-B sentence over only the after
variables denoted by $\overline{x}\prime$, is a special case. In this case, we form the {\iEVT}-sentence $\langle \texttt{Init}, \phi(\overline{x}\prime)\rangle$ where $\phi(\overline{x}\prime)$ is a predicate over the after values of the variables as assigned by the \texttt{Init} event.

\end{description}
As a further example, Figure~\ref{fig:evtsentl} contains the {\iEVT}-sentences corresponding to the abstract Event-B specification on lines 1--29 of Figure~\ref{fig:ebm0} excluding those invariants which are used for typing purposes.

\begin{figure}
\begin{center}
\begin{minipage}{0.7\textwidth}
\begin{programs}
$\{\langle$ $e$, d > 0 $\rangle\mid e \in \{\text{Init, ML\_out, ML\_in}\}\}$
$\{\langle$ $e$, n $\leq$ d $\land$ n$\prime$ $\leq$ d $\rangle\mid e \in \{\text{Init, ML\_out, ML\_in}\}\}$

$\langle$ Init, n$\prime$ = 0 $\rangle$
$\langle$ ML\_out, n < d $\land$ n$\prime$ = n + 1 $\rangle$
$\langle$ ML\_in, n > 0 $\land$ n$\prime$ = n - 1 $\rangle$
\end{programs}

\end{minipage}
\end{center}
\caption{These are the {\iEVT}-sentences corresponding to the abstract Event-B specification as illustrated on lines 1--29 of Figure~\ref{fig:ebm0}. Note that here the constant, $d$, is interpreted as a 0-ary operator.}%
\label{fig:evtsentl}
\end{figure}

\begin{restatable}{lem}{lemsen}%
\label{lem:sen}
There is a functor $\mathbf{Sen}_{\iEVT}: \mathbf{Sign}_{\iEVT} \rightarrow \mathbf{Set}$ generating for each {\iEVT}-signature $\Sigma$ a set of {\iEVT}-sentences which is an object in the category \textbf{Set} and for each {\iEVT}-signature morphism $\sigma: \Sigma_1 \rightarrow \Sigma_2$ (morphisms in the category $\mathbf{Sign}_{\iEVT}$) a function $Sen(\sigma): Sen(\Sigma_1) \rightarrow Sen(\Sigma_2)$ (morphisms in the category \textbf{Set}) translating {\iEVT}-sentences.
\end{restatable}
In this section, we have defined what is meant by an {\iEVT}-sentence. Next, we define {\iEVT}-models and the $\mathbf{Mod}_{{\iEVT}}$ functor.

\subsection{The Functor \texorpdfstring{$\mathbf{Mod}_{\iEVT}$}{Mod-EVT}, yielding \texorpdfstring{${\iEVT}$}{EVT}-models}%
\label{sec:evtmod}

Our
construction of {\iEVT}-models is based on Event-B mathematical models
as described by Abrial~\cite[Ch.\ 14]{abrial_modeling_2010}.  In these models
the state is represented as a sequence of variable-values and models
are defined over before and after states.
We interpret these states as sets of variable-to-value mappings in our definition
of {\iEVT}-models and so we define a $\Sigma$-state of an algebra $A$ in Definition~\ref{def:state}.
\begin{defi}[\textbf{$\Sigma$-\textit{State}}$_A$]%
\label{def:state}
For any given {\iEVT}-signature $\Sigma$ we define a
\textbf{$\Sigma$-state} of an algebra $A$ as a set of (sort appropriate) variable-to-value mappings whose domain is the set of sort-indexed variable names
$\Sigma.V$. Note that $A$ is a model for the underlying {\iFOPEQ} signature associated with $\Sigma$.
We define the set $State_{A}$ as the set of all such  $\Sigma$-states. By
``sort appropriate'' we mean that for any variable $x$ of sort $s$ in
$V$, the corresponding value for $x$ should be drawn from $\lvert A\rvert _s$, the
carrier set of $s$ given by a {\iFOPEQ}-model $A$.
\end{defi}
Essentially, a $\Sigma$-state$_A$ corresponds to valuations in the
$\Sigma_{\iFOPEQ}$-model, $A$, of the variables in $\Sigma.V$.
\begin{defi}[\textbf{$\Sigma_{{\iEVT}}$-\textbf{Model}}]%
\label{def:evtmod}
Given $ \Sigma = \langle S, \Omega, \Pi, E, V \rangle$,
$\mathbf{Mod}_{\iEVT}(\Sigma)$ provides a category of {\iEVT}-models, where an
{\iEVT}-\textbf{model} over $\Sigma$ is a tuple $\langle A, L, R \rangle$ where
\begin{itemize}
\item $A$ is a $\Sigma_{{\iFOPEQ}}$-model.
\item $L \subseteq State_{A}$ is the non-empty initialising set that provides the
states after the \texttt{Init} event. We note that trivial models are excluded as the initialising set $L$ is never empty. We can see this because even in the extreme cases where there are no predicates or variables in the $\mathtt{Init}$ event, $L$ is the singleton containing the empty map $(L=\{\{\}\}$).
\item For every event name, $e$, in the domain of $E$, denoted by $e\in dom(E)$, other than \texttt{Init}, we define $R.e \subseteq State_{A} \times State_{A}$ where for each pair of states $\langle s, s\prime\rangle$ in $R.e$, $s$ provides values for the variables $x$ in $V$, and $s\prime$ provides values for their primed versions $x\prime$.
Then
$R = \{R.e \mid e \in dom(E) \ and \ e \neq \texttt{Init}\}$.
\end{itemize}
\end{defi}

\noindent
Intuitively, a model over $\Sigma$ maps every event name $e$ in $(\Sigma.E)$ to a  set of variable-to-value mappings over the carriers corresponding to the sorts of each of the variables $x \in \Sigma.V$ and their primed versions $x\prime$.
For example, given the event \texttt{e} on the left of Figure~\ref{fig:rex}, with natural number
variable \texttt{x} and boolean variable \texttt{y}, we
construct the variable-to-value mappings on the right of Figure~\ref{fig:rex}.
\begin{figure}
\hspace{-50pt}
\begin{minipage}{0.42\textwidth}
\begin{programs}
  \EVT{e}
    \bkw{when} \bLabel{grd1}{x<2}
    \bkw{then} \bLabel{act1}{x := x + 1}
         \bLabel{act2}{y := \texttt{false}}
\end{programs}
\end{minipage}
\quad
\begin{minipage}{0.5\textwidth}
{\[R_e = \left\{\begin{array}{llll}
  \{ \maplet{x}{0},& \maplet{y}{false},&
  \maplet{x'}{1},& \maplet{y'}{false} \},
\\
\{ \maplet{x}{0},& \maplet{y}{true},&
\maplet{x'}{1},& \maplet{y'}{false} \},
\\
\{ \maplet{x}{1},& \maplet{y}{false},&
\maplet{x'}{2},& \maplet{y'}{false} \},
\\
\{ \maplet{x}{1},& \maplet{y}{true},&
\maplet{x'}{2},& \maplet{y'}{false} \}
\end{array}\right\}\]}
\end{minipage}
\caption{The left hand side shows an example of an Event-B event, $e$,
  with natural number variable $x$ and boolean variable $y$. When
  $x<2$, the event increments the value of $x$ and changes $y$ to
  \texttt{false}. The right hand side enumerates $R_e$, the corresponding set
  of variable-to-value mappings for the event on the left.}%
\label{fig:rex}
\end{figure}
The notation used in this example
is interpreted as $\maplet {\textit{variable name}}
{\textit{value}}$ which is a data state where the
value is drawn from the carrier set corresponding to the sort of the
variable name given in $\Sigma.V$.

The proof requirement here is to show that
$\mathbf{Mod}_{{\iEVT}}(\Sigma)$ forms a category for a given
{\iEVT}-signature $\Sigma$.
\begin{restatable}{lem}{evtmodcatlem}%
\label{lem:evtmodcat} For any {\iEVT}-signature $\Sigma$ there is a category of {\iEVT}-models $\mathbf{Mod}_{\iEVT}(\Sigma)$ where the objects in the category are {\iEVT}-models and the arrows are {\iEVT}-model morphisms.
\end{restatable}


Model reducts are central to the definition of an institution (Definition~\ref{def:ins}) and in Definition~\ref{def:evtred}, we define the {\iEVT}-model reduct. Then, in Lemma~\ref{lem:modfun}, we prove that the {\iEVT}-model reduct is a functor.

\begin{defi}[\textbf{{\iEVT}-model reduct}]%
\label{def:evtred}
The reduct of an {\iEVT}-model $M=
\langle A, L,
R \rangle$ along an {\iEVT}-signature morphism $\sigma: \Sigma
\rightarrow \Sigma'$ is given by $M\rvert _\sigma = \langle A\rvert _\sigma,
L\rvert_\sigma, R\rvert_\sigma\rangle$. Here $A\rvert _\sigma$ is the reduct of the
{\iFOPEQ}-component of the {\iEVT}-model along the {\iFOPEQ}-components of $\sigma$. $L\rvert_\sigma$ and $R\rvert _\sigma$ are based on the reduction of the states of $A$ along $\sigma$, i.e., for every $\Sigma'$-state $s$ of $A$, that is for every sorted map $s : \Sigma'.V \rightarrow \lvert A\rvert$, $s\rvert _{\sigma}$ is the map $\Sigma'.V \rightarrow \lvert A\rvert$ given by the composition $\sigma_{V} ; s$. This extends in the usual manner from states to sets of states and to relations on states.
\end{defi}

\begin{restatable}{lem}{lemmodfun}%
\label{lem:modfun}
For each {\iEVT}-signature morphism $\sigma: \Sigma_1 \rightarrow \Sigma_2$ the {\iEVT}-model reduct is a functor $\textbf{Mod}(\sigma)$ from $\Sigma_2$-models to $\Sigma_1$-models.
\end{restatable}



\begin{restatable}{lem}{lemmodfunmor}
There is a functor $\mathbf{Mod}_{\iEVT}$ yielding a category $\mathbf{Mod}(\Sigma)$ of {\iEVT}-models for each {\iEVT}-signature $\Sigma$, and for each {\iEVT}-signature morphism $\sigma: \Sigma_1 \rightarrow \Sigma_2$ a functor $\mathbf{Mod}(\sigma)$ from $\Sigma_2$-models to $\Sigma_1$-models.
\end{restatable}

In this section, we defined {\iEVT}-models, model reducts and the $\mathbf{Mod}_{{\iEVT}}$ functor. Next, we describe the satisfaction relation, $\models_{{\iEVT}}$, in {\iEVT} and prove that {\iEVT} is a valid institution.
\subsection{The Satisfaction Relation for \texorpdfstring{${\iEVT}$}{EVT}}%
\label{sec:evtsat}
All institutions are equipped with a satisfaction relation to evaluate if a particular model satisfies a sentence. Here, we define the satisfaction relation for {\iEVT}.
In order to define the satisfaction relation for {\iEVT}, we describe
an embedding from {\iEVT}-signatures and models to
{\iFOPEQ}-signatures and models.

\smallskip\noindent
Given an {\iEVT}-signature $\Sigma = \langle S, \Omega, \Pi, E, V \rangle$ we form the following two {\iFOPEQ}-signatures:
\begin{itemize}
\item $\Sigma^{(V,V\prime)}_{{\iFOPEQ}} = \langle S, \Omega\cup V \cup V', \Pi\rangle$  where
  $V$ and $V'$ are the variables and their primed versions,
  respectively, that are drawn from the {\iEVT}-signature, and
  represented as $0$-ary operators with unchanged sort.
The intuition
here is that the set of variable-to-value mappings for the free
variables in an {\iEVT}-signature $\Sigma$ are represented as expansions of $A$
along the signature inclusion obtained by adding a distinguished $0$-ary operation
symbol to the corresponding {\iFOPEQ}-signature for each of the variables $x \in V$ and
their primed versions.

\item Similarly, for the initial state and its variables, we construct
  the signature $\Sigma^{(V')}_{{\iFOPEQ}}= \langle S,
  \Omega\cup V', \Pi\rangle$.
\end{itemize}

\smallskip\noindent
Given the {\iEVT} $\Sigma$-model $\langle A,L,R \rangle$, we construct
the {\iFOPEQ}-models:
\begin{itemize}
\item For every pair of states $\langle s,s' \rangle$, we form the
  $\Sigma^{(V,V')}_{{\iFOPEQ}}$-model expansion $A^{(s, s')}$,
  which is the {\iFOPEQ}-component $A$ of the {\iEVT}-model, with
 $s$ and $s'$ added as interpretations for the new operators that
correspond to the variables from $V$ and $V'$ respectively.
\item For each initial state $s' \in L$ we construct the
  $\Sigma^{(V')}_{{\iFOPEQ}}$-model expansion $A^{(s')}$ analogously.
  \end{itemize}

\noindent
For any {\iEVT}-sentence over $\Sigma$ of the form
$\langle e, \phi(\overline{x}, \overline{x}') \rangle$, we create a
corresponding {\iFOPEQ}-formula by replacing the free variables with
their corresponding operator symbols.  We write this (closed) formula
as $\groundevt$.

Given these {\iFOPEQ}-signatures and models, we now define the satisfaction relation for {\iEVT} in Definition~\ref{def:sat}.
\begin{defi}[\textbf{Satisfaction Relation}]%
\label{def:sat}
We define the satisfaction relation in {\iEVT} as follows.  We specify satisfaction for initial events over $L$ and non-initial events over $R$.
For any {\iEVT}-model $\langle A,L,R \rangle$ and
      {\iEVT}-sentence $\langle e, \phi(\overline{x}, \overline{x}\prime)
      \rangle$, where $e$ is an event name other than \texttt{Init}, we define:

\centerline{$
\langle A,L,R \rangle \models_{\Sigma}
  \langle e, \phi(\overline{x}, \overline{x}') \rangle
\quad \iff \quad
\forall \langle s,s' \rangle \in R.e \cdot A^{(s, s\prime)}
\models_{{\Sigma^{(V,V')}_{\iFOPEQ}}} \groundevt
$}

Similarly, we evaluate the satisfaction relation of {\iEVT}-sentences of the form $\langle \texttt{Init}, \phi(\overline{x}')\rangle $ as follows:

\centerline{$\langle A,L,R \rangle
\models_{\Sigma}
\langle \mathtt{Init}, \phi(\overline{x}') \rangle
\quad \iff \quad
\forall s'  \in L \cdot A^{(s\prime)}  \models_{{\Sigma^{(V')}_{\iFOPEQ}}} \groundevtinit$}



\end{defi}
\begin{thm}[\textbf{Satisfaction Condition}]%
\label{thm:evtsat}
Given {\iEVT} signatures $\Sigma_1$ and $\Sigma_2$, a signature morphism $\sigma: \Sigma_1 \rightarrow \Sigma_2$, a $\Sigma_2$-model $M_2$ and a $\Sigma_1$-sentence $\psi_1$, the following satisfaction condition holds:
\begin{center}
$ \mathbf{Mod}(\sigma)(M_2)\models_{{\iEVT}_{\Sigma_1}} \psi_1 \quad \iff \quad M_2 \models_{{\iEVT}_{\Sigma_2}} \mathbf{Sen}(\sigma)(\psi_1)$
\end{center}
\end{thm}

\begin{proof}
We must prove this for {\iEVT}-sentences of the form $\langle e, \phi(\overline{x}, \overline{x}')\rangle$ as defined in Definition~\ref{def:evtsen}, with the satisfaction relation, $\models_{{\iEVT}}$, as defined in Definition~\ref{def:sat}. Let $M_{2}$ be the {\iEVT}-model $\langle A_{2}, L_{2}, R_{2} \rangle$, and let $\psi_{1}$ be the {\iEVT}-sentence $\langle e, \phi(\overline{x}, \overline{x}') \rangle$, then the satisfaction condition is equivalent to
\begin{center}
$ \forall \langle s, s' \rangle \in R_2\rvert _{\sigma}.e \cdot  (A_2\rvert _{\sigma})^{(s, s\prime)}\rvert _\sigma \models_{{\iFOPEQ}_{\Sigma^{(V_1,V_1\prime)}_{\iFOPEQ}}} \groundevt$\\

$ \iff \forall \langle s, s' \rangle \in R_2.\sigma_{E}(e) \cdot A_2^{(s, s\prime)}  \models_{{\iFOPEQ}_{\Sigma^{(V_2,V_2\prime)}_{\iFOPEQ}}} Sen(\sigma)(\groundevt)$
\end{center}
Here, validity follows from the validity of
satisfaction in {\iFOPEQ}. We prove the result for initial
events in the same way.
Since validity follows from that  of {\iFOPEQ}, we thus conclude that {\iEVT} preserves the satisfaction condition required of an institution.
\end{proof}
In this section we have presented the satisfaction relation in {\iEVT} and showed that {\iEVT} preserves the axiomatic property of an institution (satisfaction condition in Theorem~\ref{thm:evtsat}). As the satisfaction relation in {\iEVT} (Definition~\ref{def:sat}) relied on a {\iFOPEQ} embedding of signatures and models, we discuss the relationship between {\iFOPEQ} and {\iEVT} in Section~\ref{sec:comor}.

\subsection{Specifying Event-B contexts by relating \texorpdfstring{${\iFOPEQ}$}{FOPEQ} and \texorpdfstring{${\iEVT}$}{EVT}}%
\label{sec:comor}

Initially, we defined the relationship between {\iFOPEQ} and {\iEVT}
to be a duplex institution formed from a restricted version of {\iEVT}
({\iEVT}$_{res}$) and {\iFOPEQ} where {\iEVT}$_{res}$ is the
institution {\iEVT} but does not contain any {\iFOPEQ}
components. Duplex institutions are constructed by enriching one
institution by the sentences of
another ({\iEVT}$_{res}$ is enriched by sentences from {\iFOPEQ}) using an institution semi-morphism~\cite{goguen_institutions:_1992}~\cite[\S 10.1]{sanella_foundations_2012}. This approach would allow us to
constrain {\iEVT}$_{res}$ by {\iFOPEQ} and thus facilitate the use of
{\iFOPEQ}-sentences in an elegant way.  However, duplex institutions
are not supported in {\Hets}, the only tool-set currently available for institutions~\cite{grumberg_heterogeneous_2007},
and therefore we opt for a comorphism (Definition~\ref{def:inscomor}) which embeds the simpler
institution {\iFOPEQ} into the more complex institution {\iEVT}~\cite[\S 10.4]{sanella_foundations_2012}.

\begin{defi}[\textbf{The institution comorphism} $\rho: {\iFOPEQ} \rightarrow {\iEVT}$]%
\label{def:comorevtfopeq}
We define $\rho: {\iFOPEQ} \rightarrow {\iEVT}$ to be an institution comorphism
  composed of:
\begin{itemize}
\item The functor $\rho^{Sign}: \mathbf{Sign}_{{\iFOPEQ}} \rightarrow
  \mathbf{Sign}_{{\iEVT}}$ which takes as input a {\iFOPEQ}-signature of the form $\langle S, \Omega, \Pi \rangle$ and extends it with the set  $E=\{(\texttt{Init} \mapsto \texttt{ordinary})\}$ and an empty set of variable names $V$. Here, $\rho^{Sign}(\sigma)$ works as $\sigma$ on $S$, $\Omega$ and $\Pi$, it is the identity on the \texttt{Init} event and the empty function on the empty set of variable names. 
\item The natural transformation  $\rho^{Sen}: \mathbf{Sen}_{{\iFOPEQ}} \rightarrow \rho^{Sign};\mathbf{Sen}_{{\iEVT}}$ which treats any closed {\iFOPEQ}-sentence (given by $\phi$) as an Event-B invariant sentence, to form the set of {\iEVT}-sentences $\{ \langle e, \phi\rangle \mid e \in dom(\Sigma .E)\}$. As there are no variables in the signature, and $\phi$ is any closed formula, we do not require $\phi$ to be over the variables $\overline{x}$ and $\overline{x}\prime$.
\item The natural transformation $\rho^{Mod}: (\rho^{Sign})^{op} ;\mathbf{Mod}_{{\iEVT}} \rightarrow \mathbf{Mod}_{{\iFOPEQ}}$
is such that for any {\iFOPEQ}-signature $\Sigma$ and {\iEVT}-model $\langle A, L, \emptyset\rangle$ for $\rho^{Sign}(\Sigma)$ then

\centering$\rho^{Mod}_{\Sigma}(\langle A, L, \emptyset\rangle) = A$

\end{itemize}
\end{defi}

\noindent
Next, we prove that $\rho:{\iFOPEQ}\rightarrow{\iEVT}$ meets the axiomatic requirements of an institution comorphism.
\begin{thm}\label{thm:comorevtfopeq}
The institution comorphism $\rho$ is defined such that for any $ \Sigma \in \lvert \mathbf{Sign}_{{\iFOPEQ}}\rvert$, the translations $\rho^{Sen}_{\Sigma}: \mathbf{Sen}_{{\iFOPEQ}}(\Sigma) \rightarrow \mathbf{Sen}_{{\iEVT}}(\rho^{Sign}(\Sigma))$ and $\rho^{Mod}_{\Sigma}: \mathbf{Mod}_{{\iEVT}}(\rho^{Sign}(\Sigma)) \rightarrow \mathbf{Mod}_{{\iFOPEQ}}(\Sigma)$ preserve the satisfaction relation.
 That is, for any $\psi \in \mathbf{Sen}_{{\iFOPEQ}}(\Sigma)$ and $M' \in \lvert\mathbf{Mod}_{{\iEVT}}(\rho^{Sign}(\Sigma))\rvert$
\begin{equation*}
\rho^{Mod}_{\Sigma}(M') \models_{{\iFOPEQ}_{\Sigma}} \psi \; \iff \; M' \models_{{\iEVT}_{\rho^{Sign}(\Sigma)}} \rho^{Sen}_{\Sigma}(\psi)\label{eqn:comor}\end{equation*}
\end{thm}

\begin{proof} By Definition~\ref{def:comorevtfopeq}, $M' = \langle A, L, \emptyset \rangle$, $\rho^{Mod}_{\Sigma}(M') = A$ and
  $\rho^{Sen}_{\Sigma}(\psi) = \{ \langle e, \psi\rangle \mid e \in dom(\Sigma .E)\}$. Therefore, the above equivalence becomes
\[A \models_{{\iFOPEQ}_{\Sigma}} \psi
\; \iff \;
M' \models_{{\iEVT}_{\rho^{Sign}(\Sigma)}} \{ \langle e, \psi\rangle \mid e \in dom(\Sigma .E)\}
\]
Then, by the definition of the satisfaction relation in {\iEVT} (Definition~\ref{def:sat}) we get the following for each $\psi$ in $\{ \langle e, \psi\rangle \mid e \in dom(\Sigma .E)\}$.
\[
A \models_{{\iFOPEQ}_{\Sigma}} \psi
\; \iff \;
\forall s\prime \in L \cdot A^{( s\prime)} \models_{{\iFOPEQ}_{(\rho^{Sign}(\Sigma))^{(V')}_{{\iFOPEQ}}}} \psi
\]

We deduce that $\Sigma = (\rho^{Sign}(\Sigma))^{V'}_{{\iFOPEQ}}$,
since there are no variable names in $V'$ and thus no new operator symbols are added to the signature. As there are no variable names in $V'$, we have $L=\{\{\}\}$, so we can conclude that $A^{(s\prime)} = A$.  
Thus the satisfaction condition holds.
\end{proof}

\paragraph{Using the comorphism.}
The purpose of this comorphism is to allow us to use {\iFOPEQ} specifications in {\iEVT}.
Thus, given the institution comorphism $\rho:{\iFOPEQ} \rightarrow {\iEVT}$,
for any $\Sigma$-specification written over {\iFOPEQ} we can use the
specification-building operator \[\_\_ \texttt{ \ikw{with} } \rho:
Spec_{{\iFOPEQ}}(\Sigma) \rightarrow
Spec_{{\iEVT}}(\rho^{Sign}(\Sigma))\]
to interpret this as a specification over
{\iEVT}~\cite[\S 10.4]{sanella_foundations_2012}. This results in
a specification with just the \texttt{Init} event and
no variables, containing {\iFOPEQ}-sentences that hold in the initial state.
This process is used to represent contexts, specifically
their axioms, which are written over {\iFOPEQ} as sentences over
{\iEVT}.

We note that there is a dual nature to the relationship between {\iFOPEQ} and {\iEVT}. The first being the use of {\iFOPEQ} to evaluate the {\iEVT}-satisfaction condition via an embedding as presented in Section~\ref{sec:evtsat}. The second is the comorphism that we have described above which allows us to project Event-B context sentences into {\iEVT} and we illustrate this in our  semantics in Section~\ref{sec:transsemantics}. This double relationship offers a number of interesting consequences. In a way, {\iEVT} is tied to {\iFOPEQ} in terms of the evaluation of the satisfaction relation. However, the use of the comorphism to represent context sentences means that it is possible to replace this use of {\iFOPEQ} with another logic so long as a comorphism can be defined to {\iEVT}. This second relationship illustrates the plug `n' play nature of the {\iEVT} institution.

In cases where a specification is enriched with new events, then the axioms and invariants should also apply to these new events. Our use of {\iEVT}-sentences that define invariants mediates this as they are applied to all events in the specification when evaluating the satisfaction condition.
In order for an institution to correctly utilise the specification-building operators, we must prove properties with regard to pushouts and amalgamation. We prove that {\iEVT} meets these requirements in Section~\ref{sec:push}.

\subsection{Enabling specification-building via pushouts and amalgamation}%
\label{sec:push}

We ensure that the {\iEVT} institution has good modularity properties
by proving that {\iEVT} admits the amalgamation property: all pushouts
in \textbf{Sign}$_{{\iEVT}}$ exist and every pushout diagram in
\textbf{Sign}$_{{\iEVT}}$ admits \textit{weak} model amalgamation~\cite[\S 4.4]{sanella_foundations_2012}. Note that the weak amalgamation property corresponds to the usual definition of amalgamation (Definition~\ref{def:amalg}) with the uniqueness constraint dropped~\cite[\S 4.4]{sanella_foundations_2012}. This requirement was actually first referred to as ``exactness'' and defined by Diaconescu et al.~\cite{diaconescu_logical}. In~\cite{sanella_foundations_2012},
structured specifications are built by consecutive application of these
specification-building operators. The meaning of structured specifications is obtained by applying functions that correspond to the specification-building operators on the associated signatures and model classes. In particular, the signature (resp.\ model class) corresponding to a structured  specification that uses the \ikw{and} specification-building operator is obtained using the pushout (resp.\ amalgamation) construction shown below~\cite[\S 5]{sanella_foundations_2012}.

\begin{restatable}{prop}{pushoutprop}
Pushouts exist in \textbf{Sign}$_{{\iEVT}}$.\label{prop:push}\end{restatable}
In Definition~\ref{def:amalg} we present the definition of amalgamation as outlined by Sannella and Tarlecki and use it to structure the proof of amalgamation in {\iEVT}~\cite[\S 4.4]{sanella_foundations_2012}.
\begin{defi}[\textbf{Amalgamation}]%
\label{def:amalg}
Let $\mathbf{INS} = \langle \mathbf{Sign}, \mathbf{Sen}, \mathbf{Mod}, \langle \models_{\Sigma} \rangle_{\Sigma \in \lvert\mathbf{Sign}\rvert}\rangle$ be an institution. The following diagram in $\mathbf{Sign}$
\begin{center}
	\begin{tikzpicture}[scale=0.5]
	\node (P0) at (90:2cm) {$\Sigma$};
	\node (P1) at (90+90:2.5cm) {$\Sigma_1$} ;
	\node (P3) at (90+3.75*72:2.5cm) {$\Sigma_2$};
	\node (P2) at (90+2*90:2cm) {$\Sigma'$};
	\path[commutative diagrams/.cd, every arrow, every label]
	(P0) edge node [swap]{$\sigma_1$} (P1)
	(P1) edge node  {$\sigma_1'$} (P2)
	(P3) edge node[swap] {$\sigma_2'$} (P2)
	(P0) edge node {$\sigma_2$} (P3);
	\end{tikzpicture}
\end{center}
admits amalgamation if
\begin{itemize}
\item for any two models $M_1 \in \lvert \mathbf{Mod}(\Sigma_1) \rvert$ and $M_2 \in \lvert \mathbf{Mod}(\Sigma_2) \rvert$, there exists a unique model $M' \in \lvert \mathbf{Mod}(\Sigma') \rvert$ (amalgamation of $M_1$ and $M_2$) such that $M'\rvert _{\sigma_1'} = M_1$ and $M'\rvert _{\sigma_2'} = M_2$.
\item for any two model morphisms $f_1:M_{11} \rightarrow M_{12}$ in $\mathbf{Mod}(\Sigma_1)$ and $f_2:M_{21} \rightarrow M_{22}$ in $\mathbf{Mod}(\Sigma_2)$ such that $f_1\rvert_{\sigma_1} = f_2\rvert_{\sigma_2}$, there exists a unique model morphism $f':M_1' \rightarrow M_2'$ in $\mathbf{Mod}(\Sigma')$ (amalgamation of $f_1$ and $f_2$) such that $f'\rvert_{\sigma_1'} = f_1$ and $f'\rvert_{\sigma_2'} = f_2$.
\end{itemize}
The institution $\mathbf{INS}$ has the \textit{amalgamation} property if all pushouts exist in $\mathbf{Sign}$ and every pushout diagram in $\mathbf{Sign}$ admits amalgamation~\cite[\S 4.4]{sanella_foundations_2012}.
\end{defi}

\begin{restatable}{prop}{weakamalgprop}
Every pushout diagram in \textbf{Sign}$_{{\iEVT}}$ admits weak model amalgamation.

\end{restatable}
%
So far, we have defined our institution for Event-B, {\iEVT}, outlined the comorphism from {\iFOPEQ} to {\iEVT}, and shown that {\iEVT} meets the requirements needed for correct use of the specification-building operators. Next, we discuss the pragmatics of specification-building in {\iEVT} (Section~\ref{sec:prag}).

\subsection{Some pragmatics of Specification Building in \texorpdfstring{${\iEVT}$}{EVT}}%
\label{sec:prag}
We represent an Event-B specification, that is composed of machines and contexts, as a structured specification that is built from flat specifications (presentations in Definition~\ref{def:pres}) over {\iEVT} using the specification-building operators. Technically, {\iEVT} allows for loose specifications which are not possible in Event-B, however, we do not exploit this in our examples that follow.

Recall that, for any
signature $\Sigma$, a $\Sigma$-presentation is a set of
$\Sigma$-sentences. A model of a $\Sigma$-presentation is a
$\Sigma$-model that satisfies all of the sentences in the presentation~\cite{goguen_institutions:_1992}.  Thus, for a presentation in {\iEVT}, model
components corresponding to an event must satisfy all of the sentences
specifying that event.  This incorporates the standard semantics of the
\emph{extends} operator for events in Event-B where the extending
event implicitly contains all the parameters, guards and actions of the
extended event but can include additional parameters, guards and actions~\cite{abrial_refinement_2007}.

An interesting aspect is that if a variable is not assigned a value, within
an action, then a model for the event may associate a new value with
this variable.  Some languages deal with this using a \emph{frame
  condition}, asserting implicitly that values for unmodified
variables do not change.  In Event-B such a condition would cause
complications when combining presentations, since variables
unreferenced in one event will be constrained not to change, and this
may contradict an action for them in the other event.  As far as we
can tell, the informal semantics for the Event-B language does not
require a frame condition, and we have not included one in our definition.

\section{A Formal Semantics for Event-B}%
\label{sec:transsemantics}
In the previous sections, we described the Event-B formal specification language and the {\iEVT} institution which provides the target for our translational semantics. This section is the centrepiece of this paper and shows how we map from the constructs of the Event-B language into the semantic domains given by {\iEVT}.
Our objective is to define a formal semantics for Event-B by representing Event-B specifications as structured specifications over {\iEVT}, our institution for Event-B~\cite{ farrell_institution_2017}. Our approach is structured around the three constituent sub-languages of Event-B that we identified in Section~\ref{sec:bg}.

\paragraph{Exploiting the three-layer model:}%
\label{sec:eventb}
We utilise our three-layer model of the Event-B language that was presented in Figure~\ref{fig: ebstruc} in order to translate Event-B into the institutional framework as follows.
\begin{itemize}
\item At the base of Figure~\ref{fig: ebstruc} is the Event-B mathematical language. The institution of first-order predicate logic with equality, {\iFOPEQ}, is embedded via an institution comorphism (Definition~\ref{def:comorevtfopeq}) into the institution for Event-B, {\iEVT}~\cite{farrell_institution_2017}. The semantics that we define translates the constructs of this mathematical language into corresponding constructs over {\iFOPEQ}.
\item At the next level is the Event-B infrastructure, which consists of those language elements used to define variables, invariants, variants and events. These are translated into sentences over {\iEVT}.
\item At the topmost level is the Event-B superstructure which deals with the definition of Event-B machines and contexts, as well as their relationships (\texttt{refines}, \texttt{sees}, \texttt{extends}). These are translated into structured specifications over {\iEVT} using the specification-building operators.
\end{itemize}

\noindent
Closely following the syntax defined in Figure~\ref{fig1: syntax}, we
define the semantics of each of the Event-B language elements by
describing a mechanism to translate them into {\iEVT}-sentences. In
order to carry out such a translation we first extract the
corresponding {\iEVT}-signature $\Sigma = \langle S, \Omega, \Pi, E, V
\rangle$ from any given Event-B specification. Contexts can be
represented entirely by the underlying mathematical language and thus
translated into specifications over {\iFOPEQ}, and embedded using the
comorphism.

In Section~\ref{sec:interface}, we define an interface (Figure~\ref{fig: magic}) in order to facilitate the use of some {\iFOPEQ}
operations and functions.  We provide functions for
extracting the signature of an Event-B specification in Section~\ref{sec:extract}. In a bottom-up fashion, Sections~\ref{sec:sen} and~\ref{sec:super} define the semantics of the Event-B infrastructure and
superstructure languages, respectively.

\subsection{The \texorpdfstring{${\iFOPEQ}$}{FOPEQ} Interface}%
\label{sec:interface}

Just as the Event-B formalism is built on an underlying mathematical
language, our institution for Event-B is founded on {\iFOPEQ}, the
institution for first-order predicate logic with equality. In Figure~\ref{fig: magic}, we define a {\iFOPEQ} interface in order to
facilitate the use of its operations and functions in our
definitions. The description of these operations is outlined in Figure~\ref{fig: magic} and not specified further.

The function $\mathbb{P}_\Sigma$ described in Figure~\ref{fig: magic} takes a labelled predicate and outputs a {\iFOPEQ}-sentence ($\Sigma$\textit{-formula}).
The function $\mathbb{T}_\Sigma$ takes an expression and returns a $\Sigma$\textit{-term}. These functions are used to translate Event-B predicates and expressions into $\Sigma$\textit{-formulae} and $\Sigma$\textit{-terms} respectively.

The purpose of the function $\mathbb{M}$ is to take two lists of identifiers and a list of labelled predicates, and use these to form the {\iFOPEQ} signature $\langle S,\Omega,\Pi \rangle$. The reason for this is that when extracting a signature from a context (described in Figure~\ref{fig: sigextract}) carrier sets are interpreted as sorts and used to form $S$. The constants and axioms are used to form $\Omega$ and $\Pi$. We assume that $\mathbb{M}$ provides this translation.

\begin{figure}
\centering
\noindent\framebox[\columnwidth]{\parbox{0.97\columnwidth}{
\textbf{{\iFOPEQ} Operations}
\begin{itemize}
\item $\mathtt{F.and} : \Sigma\textit{-formula}^*  \rightarrow \Sigma\textit{-formula}$\\
This corresponds to the logical conjunction ($\land$) of a set of formulae in {\iFOPEQ}.
\item $\mathtt{F.lt} : \Sigma\textit{-term} \times \Sigma\textit{-term} \rightarrow \Sigma\textit{-formula}$\\
This operation takes two terms and returns a formula corresponding to arithmetic less than ($<$).
\item $\mathtt{F.leq} : \Sigma\textit{-term}\times \Sigma\textit{-term} \rightarrow \Sigma\textit{-formula}$\\
This operation takes two terms and returns a formula corresponding to arithmetic less than or equal to ($\leq$).
\item $\mathtt{F.exists} : VarName^* \times \Sigma\textit{-formula} \rightarrow \Sigma\textit{-formula}$\\
This operation takes a sequence of $VarName$s and a formula and returns a formula corresponding to the existential quantification of the $VarName$s over the input formula.
\item $\mathtt{F.}\iota : VarName ^* \rightarrow \Sigma\textit{-formula} \rightarrow \Sigma\textit{-formula}$\\
This operation takes a list of $VarName$s and a formula and returns the input formula with the names of all of the free variables (as given by the list of $VarName$s) primed.
\end{itemize}
\smallskip
\textbf{{\iFOPEQ} Functions}
\begin{itemize}
\item $\mathbb{P}_\Sigma: LabelledPred \rightarrow \Sigma\textit{-formula}$
\item $\mathbb{T}_\Sigma: Expression\rightarrow \Sigma\textit{-term}$
\item $\mathbb{M}: SetName ^* \times ConstName^*  \times LabelledPred^* \rightarrow \lvert \textbf{Sign}_{{\iFOPEQ}}\rvert$
\end{itemize}
}}
\caption{The {\iFOPEQ} interface provides access to a range of operations and functions which we assume to exist. These are used throughout our definitions in Figures~\ref{fig: sigextract},~\ref{fig:sentrans} and~\ref{fig:sbotrans}.}%
\label{fig: magic}
\end{figure}
For simplicity, we assume that it is possible to use Event-B identifiers in {\iFOPEQ} and {\iEVT}. Also, when we reference a $\Sigma\textit{-formula}$ in Figure~\ref{fig: magic} we mean a possibly open {\iFOPEQ}-formula over the signature given by $\Sigma$. We only return a closed {\iFOPEQ}-formula when applying $\mathbb{P}_{\Sigma}$ to axiom sentences since they form closed predicates in Figure~\ref{fig:sentrans}.

\subsection{Extracting the Signature}%
\label{sec:extract}

\begin{figure}
\footnotesize
\textit{For an Event-B specification $SP$, we form an environment $\xi = \mathbb{D}\BR{SP}\xi_0$ where  $\xi_0$ is the empty environment.}\\

$\begin{array}{ll@{\;}l}
\bullet &Env & = (MachineName \cup ContextName) \rightarrow |\mathbf{Sign}|  \quad \Comment{An environment maps names to signatures}\\
\\
\bullet & \mathbb{D}:& Specification \rightarrow Env \rightarrow Env
\Comment{Process a list of machine/context definitions}\\[2pt]
&\mathbb{D} & \BR{\langle ~ \rangle}\; \xi \;=\; \xi\\
&\mathbb{D} & \BR{hd :: tl}\; \xi \;=\; \mathbb{D} \;(\BR{tl}) \;(\mathbb{D} \;\BR{hd} \;\xi)\\
\\
\bullet & \mathbb{D}: & Machine \rightarrow Env \rightarrow Env
\Comment{Extract and store the signature for one machine}\\[2pt]
&\mathbb{D} & \left\llbracket\begin{array}{l}
\texttt{machine } m \\
\texttt{refines }  a \\
\texttt{sees }  ctx_1, \ldots,ctx_n  \\
mbody\\
\texttt{end}\\
\end{array}\right\rrbracket\; \xi
\;=\;
\xi  \cup \{ \llbracket m \rrbracket \mapsto (\langle S, \Omega, \Pi, E, V\rangle \ \ikw{$\cup$} \ r( \xi \BR{a}))\}\\ 
&&where \\
&&  \qquad \langle S, \Omega, \Pi \rangle \;=\;
\{(\xi\; \BR{ctx_1}) \ \ikw{$\cup$} \ \ldots  \ \ikw{$\cup$} \ (\xi \;\BR{ctx_n})\}
\Comment{Include signatures from `seen' contexts}\\
&& \qquad  \langle E, V, RA \rangle \;=\;
\mathbb{D} \BR{mbody}
 \Comment{Collect names from machine body}\\
&&and \\
&& \qquad r \;:\; |\textbf{Sign}_{\iEVT}| \rightarrow |\textbf{Sign}_{\iEVT}| \Comment{Signature for the abstract machine (less the refined events)}\\
&& \qquad r (\xi \BR{a}) \;=\;
\texttt{let } \Sigma_a = \xi \BR{a} \texttt{ in } \langle \Sigma_a.S,\, \Sigma_a.\Omega,\,  \Sigma_a.\Pi,\,  RA \ndres \Sigma_a.E,\, \Sigma_a.V \rangle \\
\\
\bullet & \mathbb{D}: & MachineBody  \rightarrow
(\setof{EventName \times Stat} \times \setof{VarName} \times \setof{EventName} )\\ 
&&\Comment{Extract signature elements from machine-body} \\[-2ex]
&\mathbb{D} & \left\llbracket\begin{array}{l}
\texttt{variables }  v_1, \ldots, v_n \\
\texttt{invariants }  i_1, \ldots, i_n \\
\texttt{theorems }  t_1, \ldots, t_n \\
\texttt{variant }  n  \\
\texttt{events }  e_{init} \ e_1,\ldots,e_n \\
\end{array}\right\rrbracket
\;=\;
\langle E, V, RA \rangle \\
&&where \\
&& \qquad  E \;=\; \{\emph{def}\,\BR{e_{init}},\emph{def}\,\BR{e_1}, \ldots, \emph{def}\,\BR{e_n}\}
\Comment{Names of events defined here}\\
&& \qquad  V \;=\; \langle \BR{v_1}, \ldots, \BR{v_n} \rangle
\Comment{Names of variables declared here}\\
&& \qquad RA \;=\; \emph{ref}\,\BR{e_{init}} \cup \emph{ref}\,\BR{e_1} \cup \ldots \cup \emph{ref}\,\BR{e_n}
\Comment{Names of (abstract) events refined here}\\
\\
\bullet & \emph{def}: & Event \rightarrow (EventName \times Stat)
\Comment{Extract event name \& status from an event definition}\\[2pt]
&\emph{def} & \BR{\texttt{event }e, \texttt{status } s, \texttt{refines } e_1, \ldots, e_n,  \cdots \texttt{end}} \;=\; (e\mapsto s )\\ 
&\emph{def} & \BR{e_{init}} \;=\; \langle \mathtt{Init}, \mathtt{ordinary} \rangle\\
\\
\bullet & \emph{ref}: &Event \rightarrow \setof{EventName}
\Comment{Extract names of refined events from an event definition}\\[2pt]
&\emph{ref}& \BR{\texttt{event } e, \texttt{status } s, \texttt{refines } e_1, \ldots, e_n,  \cdots \texttt{end}} 
\;=\;
\{\BR{e_1}, \ldots, \BR{e_n}\}\\
&\emph{ref}&\BR{e_{init}} = \{\BR{e_{init}}\}\\
\\

\bullet & \mathbb{D}: & Context \rightarrow Env \rightarrow  Env
\Comment{Extract and store the signature for one context}\\[2pt]
&\mathbb{D} & \left\llbracket\begin{array}{l}
\texttt{context } ctx \\
\texttt{extends } ctx_1,\ldots,ctx_n \\
cbody\\
\texttt{end}
\end{array}\right\rrbracket \xi
\;=\;
  \xi \;\cup\; \{\BR{ctx} \mapsto \begin{array}{l}
  (\mathbb{D}\BR{cbody} \mathbin{\ikw{$\cup$}} \xi \BR{ctx_1} \mathbin{\ikw{$\cup$}}
    \ldots \mathbin{\ikw{$\cup$}} \xi \BR{ctx_n})\}\\
  \end{array}
\\
\\
\bullet & \mathbb{D}: & ContextBody \rightarrow  \lvert\textbf{Sign}_{{\iFOPEQ}}\rvert
\Comment{Extract the {\iFOPEQ} signature from a context body}
\\[2pt]
&\mathbb{D} & \left\llbracket\begin{array}{l}
\texttt{sets } s_1,\ldots,s_n \\
\texttt{constants } c_1, \ldots, c_n\\
\texttt{axioms }  a_1, \ldots, a_n \\
\texttt{theorems } t_1, \ldots, t_n \\
\end{array}\right\rrbracket
\;=\; \langle S, \Omega, \Pi \rangle
\Comment{Sorts, operations, predicates}\\
&& where \\
&& \qquad  \langle S, \Omega, \Pi \rangle  = \mathbb{M}((\BR{s_1},\ldots,\BR{s_n}),\; (\BR{c_1}, \ldots, \BR{c_n}),\; (\BR{a_1}, \ldots, \BR{a_m}))\\
\end{array}$
\caption{These auxiliary functions are used to extract a signature from an Event-B specification, which is a  list of context or machine definitions.  The extraction of {\iFOPEQ}-signatures uses the interface functions described in Figure~\ref{fig: magic}.}%
\label{fig: sigextract}
\end{figure}

Before we form sentences in any institution, we must first specify the corresponding signature.
To aid this process, we define $Env$ to be the type of all
  environments, which are just mappings from machine/context names to
  their corresponding signatures. In all subsequent definitions we use
  $\xi$ to denote an environment of type $Env$. We define the overloaded function $\mathbb{D}$ in Figure~\ref{fig: sigextract} to build the environment from a given Event-B specification, which is a list of machine and context definitions.

The function $\mathit{def}$ (Figure~\ref{fig: sigextract}) extracts the function $($event name $\mapsto$ status$)$ for each event in the machine and these functions form the $E$ component of the signature. Each event name is tagged with a status in order to correctly form variant sentences which will be discussed in Section~\ref{sec:sen}. The function $\mathit{ref}$ (Figure~\ref{fig: sigextract}) forms the set of events that a particular concrete event refines. We use this function to remove the names of the refined abstract events, and the status that each is paired with, from the abstract machine's signature, before we combine it with the concrete signature. We use the domain anti-restriction operator $\ndres$ to express this, where for any set $S$ and function $T$, we have
\[S\ndres T = \{x\mapsto y \mid x \mapsto y \in T \land x \notin S\}\]

In order to illustrate this signature extraction in practice, Figure~\ref{fig:tlsig} contains the signature that was extracted using these functions from \texttt{m1} in Figure~\ref{fig:ebm1}. Notice how the variables and events in the signature include those from the abstract machine, \texttt{m0} in Figure~\ref{fig:ebm0} (lines 8--29), and the constant from the context (lines 1--7) has been included as a $0$-ary operator using the {\iFOPEQ} interface. We note that, internally, \textsf{Rodin} distinguishes typing and non-typing invariants and we have used the typing invariants on lines 7--9 of Figure~\ref{fig:ebm1} and line 13 of Figure~\ref{fig:ebm0} to infer the types of the variables, and we assume an appropriate
  mapping to types in {\iFOPEQ}.

\begin{figure}
\begin{program}
$\Sigma_{\texttt{m1}} = \langle S, \Omega, \Pi, E, V \rangle$
  where
  $S = \{\mathbb{N}\}$,
  $\Omega = \{0:\mathbb{N},d:\mathbb{N}\}$,
  $\Pi = \{>:\mathbb{N}\times\mathbb{N}\}$,
  $E = \{(\texttt{Init} \mapsto \texttt{ordinary}), (\texttt{ML\_in} \mapsto \texttt{ordinary}), (\texttt{ML\_out} \mapsto \texttt{ordinary})$,
         $(\texttt{Il\_in} \mapsto \texttt{convergent}), (\texttt{Il\_out} \mapsto \texttt{convergent})\}$,
  $V = \{\texttt{n:}\mathbb{N}, \texttt{a:}\mathbb{N}, \texttt{b:}\mathbb{N}, \texttt{c:}\mathbb{N}\}$
\end{program}
\caption{An example of the signature extraction: this is the result of applying the functions from Figure~\ref{fig: sigextract} to the Event-B machine specification of \texttt{m1} in Figure~\ref{fig:ebm1}.}%
\label{fig:tlsig}
\end{figure}

Once the environment has been formed, we can then define a systematic translation from specifications in Event-B to structured specifications over {\iEVT}. We take a bottom-up approach to this translation which comprises two parts.
\begin{itemize}
\item The first mapping (in Figure~\ref{fig:sentrans}) that we define is from the Event-B infrastructure sentences (invariants, variants, events and axioms) to sentences over {\iEVT} (for invariants, variants and events) and sentences over {\iFOPEQ} (for axioms) in Section~\ref{sec:sen}.
\item The second mapping (in Figure~\ref{fig:sbotrans}) that we provide is from the superstructure components of an Event-B specification to structured specifications over {\iEVT} (for machines) and structured specifications over {\iFOPEQ} (for contexts) in Section~\ref{sec:super}.

\end{itemize}

\begin{figure*}
{\renewcommand{\baselinestretch}{1.0}
\footnotesize
$\begin{array}{ll@{\;}l}
\bullet & \mathbb{S}_\Sigma: & MachineBody \rightarrow Sen_{{\iEVT}}(\Sigma)
\Comment{Build sentences from a machine body}\\[2pt]
&\mathbb{S}_\Sigma & \left\llbracket\begin{array}{l}
\texttt{variables }  v_1, \ldots, v_n \\
\texttt{invariants }  i_1,\ldots,i_n \\
\texttt{theorems }  t_1,\ldots,t_n \\
\texttt{variant }  n  \\
\texttt{events }  e_{init}, \ e_1,\ldots,e_n \\
\end{array}\right\rrbracket
\;=\;
\left(\begin{array}{ll}
& \mathbb{I}_\Sigma\BR{i_1} \cup \ldots \cup \mathbb{I}_\Sigma\BR{i_n}
\\
\cup & \mathbb{V}_\Sigma \BR{n} \\
\cup & \mathbb{E}_\Sigma \BR{e_{init}} \\
\cup&  \mathbb{E}_\Sigma\BR{e_1} \cup \ldots \cup \mathbb{E}_\Sigma\BR{e_n}
\end{array}\right)\\
\\
\bullet & \mathbb{I}_\Sigma: & LabelledPred \rightarrow Sen_{{\iEVT}}(\Sigma)
\Comment{Invariant sentences}\\[2pt]
&\mathbb{I}_\Sigma &\BR{i} \;=\; \{\langle \BR{e},\;
\mathtt{F.and}(\mathbb{P}_\Sigma\BR{i}, \mathtt{F.}\iota (\Sigma.V)(\mathbb{P}_\Sigma\BR{i}))\rangle \mid e \in dom(\Sigma.E)\} \\
\\
\bullet & \mathbb{V}_\Sigma: & Expression \rightarrow Sen_{{\iEVT}}(\Sigma)
\Comment{Variant can't increase for non-ord.\ events}\\[2pt]
&\mathbb{V}_\Sigma &\BR{n}
\;=\;
\begin{array}[t]{lr@{\;\mid\;}l}
& \{\langle \BR{e},\;
\mathtt{F.lt}(\mathtt{F.}\iota(\Sigma.V)(\mathbb{T}_\Sigma \BR{n}), \mathbb{T}_\Sigma \BR{n}) \rangle & (e \mapsto \mathtt{convergent}) \in \Sigma.E\} \\
\cup & \{\langle \BR{e},\;
\mathtt{F.leq}(\mathtt{F.}\iota(\Sigma.V) (\mathbb{T}_\Sigma\BR{n}), \mathbb{T}_\Sigma \BR{n})\rangle & (e \mapsto \mathtt{anticipated}) \in \Sigma.E\}
\end{array}\\
\\

\bullet & \mathbb{E}_\Sigma: & InitEvent \rightarrow Sen_{{\iEVT}}(\Sigma)
\Comment{Initial event: get sentences from actions}\\[2pt]
&\mathbb{E}_\Sigma& \left\llbracket\begin{array}{l}
\texttt{event Initialisation}    \\
\texttt{status ordinary}   \\
\texttt{then } act_1, \ldots, act_n\\
\texttt{end}\\
\end{array}\right\rrbracket
\;=\;
\{\langle \mathtt{Init},\; BA \rangle \}\\
&& where\\
&& \qquad BA \;=\; \mathtt{F.and}(\mathbb{P}_\Sigma\BR{act_1},\ldots,\mathbb{P}_\Sigma\BR{act_n}) \\
\\

\end{array}$

$\begin{array}{ll@{\;}l}
\bullet & \mathbb{E}_\Sigma: & Event \rightarrow Sen_{{\iEVT}}(\Sigma)
\Comment{Non-initial event: get sentences from event body}\\[2pt]
&\mathbb{E}_\Sigma& \left\llbracket\begin{array}{l}
\texttt{event } e   \\
\texttt{status} \  s   \\
\texttt{refines } e_1, \ldots, e_n \\
ebody\\
\texttt{end}\\
\end{array}\right\rrbracket
\;=\;
\{\langle \BR{e},\; \mathbb{F}_\Sigma\BR{ebody} \rangle \}\\
\\
\bullet & \mathbb{F}_\Sigma: & EventBody \rightarrow \Sigma\textit{-formula}
\Comment{Build a {\iFOPEQ} formula for an event definition}\\[2pt]
&\mathbb{F}_\Sigma & \left\llbracket\begin{array}{l}
\texttt{any} \  p_1, \ldots, p_n  \\
\texttt{where} \ grd_1,\ldots, grd_n   \\
\texttt{with} \ w_1,\ldots, w_n \\
\texttt{then} \ act_1,\ldots,act_n \\
\end{array}\right\rrbracket
\;=\;
\begin{array}[t]{l}
\mathtt{F.exists} (p,\; \mathtt{F.and}(G, W, BA))\\[3pt]
\qquad\Comment{Formula is existentially quantified }\\
\qquad \CommentLN{over event parameters $p$}\\
\end{array}\\
&& where\\
&& \qquad p \;=\; \langle \BR{p_1}, \ldots, \BR{p_n} \rangle \Comment{List of parameters}\\
&& \qquad G \;=\; \mathtt{F.and}(\mathbb{P}_\Sigma\BR{grd_1},\ldots,\mathbb{P}_\Sigma\BR{grd_n}) \Comment{Guards}\\
&& \qquad W \;=\; \mathtt{F.and}(\mathbb{P}_\Sigma\BR{w_1}, \ldots,\mathbb{P}_\Sigma\BR{w_n}) \Comment{Witnesses}\\
&& \qquad BA \;=\; \mathtt{F.and}(\mathbb{P}_\Sigma\BR{act_1},\ldots,\mathbb{P}_\Sigma\BR{act_n}) \Comment{Actions}\\
\\
\bullet & \mathbb{S}_\Sigma:& ContextBody \rightarrow Sen_{{\iFOPEQ}}(\Sigma)
\qquad \Comment{Context: get {\iFOPEQ} sentences from axioms}\\[2pt]
&\mathbb{S}_\Sigma & \left\llbracket\begin{array}{l}
\texttt{sets } s_1,\ldots,s_n \\
\texttt{constants } c_1, \ldots, c_n\\
\texttt{axioms }  a_1,\ldots,a_n  \\
\texttt{theorems } t_1,\ldots,t_n  \\
\end{array}\right\rrbracket
\;=\;
\{\mathbb{P}_\Sigma \BR{a_1}, \ldots, \mathbb{P}_\Sigma \BR{a_n} \}
 \\
\end{array}$
}
\caption{The semantics for the Event-B infrastructure sub-language is provided by translating Event-B machines and contexts into sets of sentences.
For machines we generate sentences over {\iEVT}, denoted $Sen_{{\iEVT}}(\Sigma)$, while for contexts we generate sentences over {\iFOPEQ}, denoted $Sen_{{\iFOPEQ}}(\Sigma)$. We use the interface operations and functions described in Figure~\ref{fig: magic} throughout this translation.   }%
\label{fig:sentrans}
\end{figure*}
\subsection{Defining the Semantics of Event-B Infrastructure Sentences}%
\label{sec:sen}
In this section, we define a translation from Event-B infrastructure sentences to sentences over {\iEVT}. We translate the axiom sentences that are found in Event-B contexts to sentences over {\iFOPEQ} as they form part of the underlying Event-B mathematical language as shown in Figure~\ref{fig: ebstruc}. We refer the reader to Section~\ref{sec:evtsen}, where we described how to translate Event-B variants, invariants and events into {\iEVT}-sentences, and the equations described here mirrors this translation.

We define an overloaded meaning function, $\mathbb{S}_\Sigma$, for
specifications in Figure~\ref{fig:sentrans}. $\mathbb{S}_\Sigma$ takes
as input a specification and returns a set of sentences over {\iEVT}
($Sen_{{\iEVT}}(\Sigma)$) for machines and a set of sentences over
{\iFOPEQ} ($Sen_{{\iFOPEQ}}(\Sigma)$) for contexts. When applying
$\mathbb{S}_\Sigma$ to a machine (resp.\ context) we also define
functions for processing invariants, variants and events
(resp.\ axioms). These are given by $\mathbb{I}_\Sigma,
\mathbb{V}_\Sigma$ and $\mathbb{E}_\Sigma$. Axioms are predicates that
can be translated into closed {\iFOPEQ}-formulae using the function
$\mathbb{P}_\Sigma$ which is defined in the interface in Figure~\ref{fig: magic}.

Given a list of invariants $i_1, \ldots, i_n$ we define the function
$\mathbb{I}_\Sigma$ in Figure~\ref{fig:sentrans}. Each invariant, $i$,
is a \textit{LabelledPred} from which we form the following open
{\iFOPEQ}-sentence:
$\mathtt{F.and}(\mathbb{P}_\Sigma\BR{i},\mathtt{F.}\iota(\Sigma.V)(\mathbb{P}_\Sigma
\BR{i}))$. Each invariant sentence is paired with each event name in
the signature to form the set $\{ \langle e,
\mathtt{F.and}(\mathbb{P}_\Sigma\BR{i},\mathtt{F.}\iota(\Sigma.V)(\mathbb{P}_\Sigma
\BR{i}))\rangle \mid e \in dom(\Sigma.E)\}$ and this ensures that the
invariant is applied to all events when evaluating the satisfaction
condition.

The use of conjunction here for invariants ($\mathtt{F.and}$) rather
than implication may seem an unusual choice.  Typically (e.g. Abrial's
INV proof obligation,~\cite[\S5.2.1]{abrial_modeling_2010}), one demands
that if the invariant holds in the before-state then it should also
hold in the after-state.  However, since our sentences are effectively
indexed by events, and our semantics are stratified over events, we
are always working in a context where a given event has occurred.
Thus, in this (localised) context, we specify that the invariant must hold
in \emph{both} the before and after states.  Since separate
sentences are conjoined when grouped, using an implication here would
also allow undesirable models when such sentences are combined. Specifically, if implication were used between the invariants in the before and after states then models that may not obey the invariant in the before state would be valid (i.e. $\mathit{false} \Rightarrow \mathit{true}$). In our context, this is incorrect because we are assuming that a given event has happened which is only possible if the invariant is true in the before state. Further, we use logical conjunction to combine sentences so that the combination imposes restrictions on the set of allowable models. Including implication causes issues here because it would essentially create a sentence with multiple options (if-then branches) that would be carried through to the models.

Given a variant expression $n$, we define the function $\mathbb{V}_\Sigma$ in Figure~\ref{fig:sentrans}. The variant is only relevant for specific events so we pair it with an event name in order to meaningfully evaluate the variant expression. An event whose status is \texttt{convergent} must strictly decrease the variant expression~\cite{abrial_modeling_2010}. An event whose status is \texttt{anticipated} must not increase the variant expression. This expression can be translated into an open {\iFOPEQ}-term using the function $\mathbb{T}_\Sigma$ as described in Figure~\ref{fig: magic}. From this we  form a formula based on the status of the event(s) in the signature $\Sigma$ as described in Section~\ref{sec:sen}. Event-B machines are only permitted to have one variant~\cite{abrial2010rodin}.

In Figure~\ref{fig:sentrans} we define the function $\mathbb{E}_\Sigma$ to process a given event definition. Event guard(s) and witnesses are predicates that can be translated via $\mathbb{P}_\Sigma$ into open {\iFOPEQ}-formulae denoted by $G$ and $W$ respectively in Figure~\ref{fig:sentrans}. In Event-B, actions are interpreted as before-after predicates e.g. \texttt{x := x + 1} is interpreted as \texttt{x$\prime$ = x + 1}. Therefore, actions can also be translated via $\mathbb{P}_\Sigma$ into open {\iFOPEQ}-formulae denoted by $BA$ in Figure~\ref{fig:sentrans}. Thus the semantics of an $EventBody$ definition is given by the function $\mathbb{F}_\Sigma$ which returns the formula $\mathtt{F.exists} (p, \mathtt{F.and}(G, W, BA))$ where $p$ is a list of the event parameters. 

A context can exist independently of a machine and is written as a specification over {\iFOPEQ}. Thus, we translate an axiom sentence directly as a {\iFOPEQ}-sentence which is a closed $\Sigma$-formula using the function $\mathbb{P}_\Sigma$ given in Figure~\ref{fig: magic}. Axiom sentences are closed {\iFOPEQ}-formulae (elements of $Sen_{{\iFOPEQ}}(\Sigma)$) which are interpreted as valid sentences in {\iEVT} using the comorphism $\rho$.

Figure~\ref{fig:sentrans} describes our equations for the Event-B infrastructure constructs.  Next we show the relevant equations for the Event-B superstructure constructs.

\subsection{Defining the Semantics of Event-B Superstructure Sentences}%
\label{sec:super}

Based on the syntax defined in Figure~\ref{fig1: syntax} we have
identified the core constructs that form the Event-B superstructure
language:
\begin{itemize}
\item A context \texttt{extends} other contexts
\item A machine \texttt{sees} contexts and
  \texttt{refines} another (more abstract) machine
\item An individual event \texttt{refines} other events from the
  abstract machine
\end{itemize}

\noindent
In this section, we define a semantics for the Event-B superstructure language using specification-building operators. In Figure~\ref{fig:sbotrans} we define the  function $\mathbb{B}$ to translate Event-B specifications written using the superstructure language to structured specifications over {\iEVT} that use the specification-building operators defined in the theory of institutions~\cite[\S 5.1]{sanella_foundations_2012}.
We translate a specification as described by Figure~\ref{fig1: syntax} into a structured specification over the institution {\iEVT}.

\begin{figure*}
{\renewcommand{\baselinestretch}{1.0}
\footnotesize
\textit{The semantics of an Event-B specification $SP$ are given by
$\mathbb{B}\BR{SP}\xi$, where $\xi = \mathbb{D}\BR{SP}\xi_0$ is the environment defined by Figure~\ref{fig: sigextract}}.\\[3ex]

$\begin{array}{ll@{\;}l}
\bullet & \mathbb{B}:& Specification \rightarrow Env \rightarrow \lvert\textbf{Spec}\rvert^*
\Comment{Process specifications in an environment,}\\[2pt]
&\mathbb{B}& \BR{\langle~\rangle}\;  \xi \;=\; \langle~\rangle
\CommentLN{build a list of structured specifications}\\
&\mathbb{B}& \BR{hd :: tl}\;  \xi
\;=\;
(\mathbb{B}\; \BR{hd}\; \xi) :: (\mathbb{B}\; \BR{tl}\; \xi)\\
\vspace{5pt}\\
\bullet & \mathbb{B}:& Machine \rightarrow Env \rightarrow \lvert\textbf{Spec}_{\iEVT}\rvert
\Comment{Build an {\iEVT} structured specification for one machine}\\[2pt]
&\mathbb{B}&
\left\llbracket\begin{array}{l}
\texttt{machine } m \\
\texttt{refines }  a \\
\texttt{sees }  ctx_1, \ldots, ctx_n \\
mbody \\
\texttt{end}\\
\end{array}\right\rrbracket\; \xi
\;=\;
\begin{array}{l}
\left\langle\begin{array}{l}
\Sigma, \;
\left[\begin{array}{l}
\SPEC{\BR{m}}{\iEVT} \\
\quad\CommentNF{Include contexts using the comorphism $\rho$:}\\
\quad(\BR{ctx_1} \mathtt{\;\ikw{and}\;} \ldots \mathtt{\;\ikw{and}\;} \BR{ctx_n}) \mathtt{\;\ikw{with}\;} \rho \\
\quad\Comment{Sentences from the refined machine (if any):}\\
\quad(\mathtt{\ikw{and}\;} \mathbb{A}_{\Sigma} \BR{mbody}\BR{a} \xi) \\
\mathtt{\ikw{then}} \\
\quad \mathbb{S}_\Sigma\BR{mbody}\\
\end{array}\right] \\
\end{array}\right\rangle\\[2pt]
\quad
where \ \Sigma = \xi\BR{m}. \\
\end{array}\\
\\
\bullet & \mathbb{A}_\Sigma: & MachineBody \rightarrow EventName \rightarrow Env \rightarrow |\mathbf{Spec}_{\iEVT}|
\Comment{Extract any relevant specification from }\\
&& \CommentLN{the refined (abstract) machine}\\[-2ex]
&\mathbb{A}_\Sigma& \left\llbracket\begin{array}{l}
\texttt{variables }  v_1, \ldots, v_n \\
\texttt{invariants }  i_1,\ldots,i_n \\
\texttt{theorems }  t_1,\ldots,t_n \\
\texttt{variant }  n  \\
\texttt{events }  e_{init}, \ e_1,\ldots,e_n \\
\end{array}\right\rrbracket \BR{a}\; \xi
\;=\;
\begin{array}[t]{l}
\mathbb{I}_\Sigma\BR{i_1} \mathtt{\;\ikw{and}\;}\ldots \mathtt{\;\ikw{and}\;}\mathbb{I}_\Sigma\BR{i_n} \\
\mathtt{\ikw{and}\;}
\mathbb{R}_\Sigma\BR{e_1}\BR{a}\xi
\mathtt{\;\ikw{and}\;}\ldots \mathtt{\;\ikw{and}\;}\mathbb{R}_\Sigma\BR{e_n}\BR{a}\xi\\[3pt]
\Comment{Conjoin sentences from each event definition}\\
\end{array}\\
\\
\end{array}$

$\begin{array}{ll@{\;}l}

\bullet & \mathbb{R}_\Sigma: & Event\rightarrow EventName \rightarrow Env \rightarrow |\mathbf{Spec}_{\iEVT}|
\Comment{Extract specification from one refined event}\\[2pt]
&\mathbb{R}_\Sigma&
\left \llbracket \begin{array}{l}
\texttt{event } e_c   \\
\texttt{status} \  s  \\
\texttt{refines} \  e_{1}, \ldots, e_{n}  \\
ebody\\
\texttt{end}\\
\end{array} \right \rrbracket\; \BR{a}\; \xi
\;=\;
\begin{array}{l} \mathtt{let}\\
\qquad \Sigma_a = \xi \BR{a},
 \qquad\Comment{Signature of abstract machine}\\
\qquad \CommentNF{Use $\Sigma_h$, $\sigma_h$ to select only refined events:} \\
\qquad \Sigma_h = \langle \Sigma_a.S,\, \Sigma_a.\Omega, \,\Sigma_a.\Pi,\\
\qquad \hfill \{\BR{e_1}, \ldots, \BR{e_n}\} \lhd \Sigma_a.E, \, \Sigma_a.V \rangle,\\
\qquad \sigma_h : \Sigma_h \hookrightarrow \Sigma_a,\\
\qquad \CommentNF{Use $\sigma_m$ to reassign refined event sentences to $e_c$:}\\
\qquad \sigma_m: \Sigma_h \rightarrow \Sigma\\
\qquad \sigma_m = \left\langle\begin{array}{l}
 \Sigma_h.S \hookrightarrow\Sigma.S,\;
 \Sigma_h.\Omega \hookrightarrow\Sigma.\Omega,
 \Sigma_h.\Pi\hookrightarrow\Sigma.\Pi,\\
 \Sigma_h.E \mapsto \{\BR{e_c}\} \lhd \Sigma.E,\\
 \Sigma_h.V \hookrightarrow\Sigma.V
\end{array}\right\rangle\\
\mathtt{in }\\
\qquad (\BR{a} \mathtt{\;\ikw{hide\;via}\;} \sigma_h)\mathtt{\;\ikw{with}\;} \sigma_m.
\end{array}\\
\\
\bullet & \mathbb{B}:& Context \rightarrow Env \rightarrow \lvert\textbf{Spec}_{\iFOPEQ}\rvert
\Comment{Build a {\iFOPEQ} structured specification for one context}\\[2pt]
&\mathbb{B} & \left\llbracket\begin{array}{l}
\texttt{context } ctx \\
\texttt{extends } ctx_1,\ldots,ctx_n \\
cbody\\
\texttt{end}
\end{array}\right\rrbracket \xi
\;=\;
\begin{array}{l}
\left\langle\begin{array}{l}
\Sigma, \;
\left[\begin{array}{l}
\SPEC{\BR{ctx}}{\iFOPEQ} \\
\qquad \BR{ctx_1} \mathtt{\;\ikw{and}\;} \ldots \mathtt{\;\ikw{and}\;}  \BR{ctx_n}\\
\mathtt{\ikw{then}}\\
\qquad \mathbb{S}_\Sigma \BR{cbody}
\end{array}\right]\\
\end{array}\right\rangle \\
\\[-1.75ex]
\qquad where \ \Sigma = \xi \BR{ctx}.
\end{array}
\end{array}$}
\caption{The semantics for the Event-B superstructure sub-language is defined by translating Event-B specifications into structured specifications over {\iEVT} using the function $\mathbb{B}$ and the specification-building operators defined in the theory of institutions (Table~\ref{table: sbo}). Note that $\textbf{Spec}$ denotes the category of specifications as per~\cite[\S 5.5]{sanella_foundations_2012}.}%
\label{fig:sbotrans}
\end{figure*}

The construct that enables a context to extend others is used in Event-B to add more details to a context. Since a context in Event-B only refers to elements of the {\iFOPEQ} component of an {\iEVT} signature it is straightforward to translate this using the specification-building operator \ikw{then}.
As outlined in Table~\ref{table: sbo}, \ikw{then} is used to enrich the signature with new sorts/operations etc.~\cite[\S 5.2]{sanella_foundations_2012}.

A context can extend and be extended by more than one context. A context that extends multiple others contains all constants and axioms of the contexts that it is extending and the additional specification that it adds to extend these contexts~\cite{abrial2010rodin}.
To give a semantics for context extension using the specification-building operators we use the \ikw{and} operator to merge all extended contexts and the \ikw{then} operator to incorporate the extending context itself.
The specification-building operator \ikw{and} gives the sum of two specifications that can be written over different signatures (Table~\ref{table: sbo}).
It is the most straightforward way to combine specifications over different signatures~\cite[\S 5.2]{sanella_foundations_2012}.

In Event-B, a machine \texttt{sees} one or more contexts. This construct is used to add a context(s) to a machine so that the machine can refer to elements of the context(s). We know that the relationship between {\iFOPEQ} and {\iEVT} is that of an institution comorphism (Definition~\ref{def:comorevtfopeq})~\cite{farrell_institution_2017}. This enables us to directly use {\iFOPEQ}-sentences, as given by the context in this case, in an {\iEVT} specification. We use the specification-building operation \ikw{with} $\rho$ which indicates  translation by an institution comorphism $\rho$~\cite[\S 10.4]{sanella_foundations_2012}. The resulting machine specification is heterogeneous as it links two institutions, {\iEVT} and {\iFOPEQ}, by the institution comorphism.

An Event-B machine can refine at most one other machine and there are two types of machine refinement: superposition and data refinement~\cite{abrial2010rodin}. The specification-building operator, \ikw{then}, can account for both of these types of refinement because either new signature components or constraints on the data (gluing invariants) are added to the specification. In Figure~\ref{fig:sbotrans} the function $\mathbb{A}_{\Sigma}$ is used to process the events in the concrete machine which refine those in the abstract machine.

Event refinement in Event-B is superposition refinement~\cite{abrial2010rodin}. By superposition refinement all of the components of the corresponding abstract event are implicitly included in the refined version. This approach is useful for gradually adding more detail to the event. In {\iEVT}, we have not prohibited multiple definitions of the same event name. When there are multiple definitions we combine them by taking the conjunction of their respective formulae which will constrain the model. As mentioned above, when refining an abstract machine we use the function $\mathbb{A}_{\Sigma}$ in Figure~\ref{fig:sbotrans} to process the events in the concrete machine which refine those in the abstract machine. $\mathbb{A}_{\Sigma}$ in turn calls the function $\mathbb{R}_{\Sigma}$ (Figure~\ref{fig:sbotrans}).

$\mathbb{R}_{\Sigma}$ restricts the event component of the abstract
machine signature to those events contained in the \texttt{refines}
clause of the event definition using domain restriction
($\dres$). $\mathbb{A}_{\Sigma}$ also extracts the invariant sentences
from the abstract machine. This new signature is included in the
abstract machine via the signature morphism $\sigma_h$. We then form
the signature morphism $\sigma_m$ by letting it be the identity on the sort,
operation, predicate and variable components of $\Sigma_h$.
For events, $\sigma_m$ maps each of the abstract event signature components that are being refined by the concrete event $e_c$ to the signature component corresponding to $e_c$. The resulting specification uses \ikw{hide via} and \ikw{with} to apply these signature morphisms ($\sigma_h$ and $\sigma_m$) correctly. Note that we could have used the \ikw{reveal} specification-building operator rather than \ikw{hide via} to achieve the same result~\cite[\S 5.2]{sanella_foundations_2012}.

Thus, we have formalised an institution-based  semantics for the three layers of the Event-B language.
In previous empirical work, we discovered that the Event-B language is lacking a well-defined set of modularisation constructs~\cite{farrell_clones_2017}. It is clear that the specification-building operators available in the theory of institutions can provide a solution to this. We validate our  semantics and show how these specifications can be modularised using the specification-building operators in the next sections.

\section{The Semantics in Action: a worked example}%
\label{sec:impl}

In order to validate our  semantics we developed a tool to take Event-B specifications and generate the corresponding {\iEVT}-specifications, following the definitions given in Section~\ref{sec:transsemantics}.
The tool is called {\EBtoEVT} and takes as input an Event-B specification in the XML format used by {\Rodin}, and generates as output the {\iEVT}-specifications, syntactically sugared in a {\Hets}-like notation.
{\EBtoEVT} is written in Haskell~\cite{thompson1999haskell}, a choice that allows our code to closely mirror the function definitions, and is compatible with the {\Hets} code-base. Specifically, \textsc{Hets} facilitates heterogeneous specification and proof using formalisms that have been encoded as institutions.
To illustrate our approach, Figure~\ref{fig:source} shows the
associated Haskell code snippet of the $\mathbb{F}$ function from
Figure~\ref{fig:sentrans}. The reader  need only refer back to the
definition of $\mathbb{F}_\Sigma$ in Figure~\ref{fig:sentrans} to see their correspondence. Since refinement is central to systems developed using Event-B, we demonstrate how a chain of Event-B refinement steps is captured using our translational semantics. We illustrate this via the application of our {\EBtoEVT} tool to the running example in this paper.

\begin{figure}
\begin{programsc}
bbEventDefines ::[(Ident, Ident)]-> [EventDef] -> [Sentence]
bbEventDefines xx' evtList =
  let pp eBody = (eany eBody)
      gg eBody = F.and (ewhere eBody)
      ww eBody = F.and (ewith eBody)
      ba eBody = F.and (ethen eBody)
      bbEvent eBody = (F.exists (xx' ++(pp eBody)) (F.and [gg eBody, ww eBody, ba eBody]))
      bbEventDef evt = (Sentence (ename evt) (bbEvent (ebody evt)))
  in (map bbEventDef evtList)
\end{programsc}
\caption{The Haskell implementation of the $\mathbb{F}$ function from Figure~\ref{fig:sentrans} as applied to a list of Event-B event definitions.}%
\label{fig:source}
\end{figure}

\subsection{The Abstract Model}

We illustrate the use of {\EBtoEVT} using the cars on a bridge example as introduced in Section~\ref{sec:bg}. This Event-B specification has three refinement steps, resulting in four machines and three contexts. The final refinement step results in quite a large specification and so we omit it here but more details can be found in~\cite{farrell_phdthesis}.
We present the resultant {\iEVT}-specification corresponding to each of the first two refinement steps.

Figure~\ref{fig:ebm0} introduced an abstract Event-B specification and using our {\EBtoEVT} translator tool we obtain the corresponding {\iEVT}-specification that is presented in Figure~\ref{fig:evtm0}.

\begin{description}
\item[Lines 1--5] This is a {\iCASL}-specification that describes the context from Figure~\ref{fig:ebm0}. Note that {\iFOPEQ} is a sub-logic of {\iCASL}~\cite{caslref_2004} and we use {\iCASL} in its place here because it is available in {\Hets}~\cite{grumberg_heterogeneous_2007}.  The constant $d$ is represented as an operation of the appropriate type and the non-typing axioms are included as predicates. This is achieved via the function $\mathbb{P}$ that was defined in the {\iFOPEQ} interface (Figure~\ref{fig: magic}).
\item[Lines 6--20] This is the {\iEVT}-specification that is generated corresponding to the abstract machine in Figure~\ref{fig:ebm0} (lines 8--29). The machine variable, $n$, is represented as an operation and the context specification is included using the \ikw{then} specification-building operator on lines 7 and 8. This corresponds to the application of the function $\mathbb{B}$ (Figure~\ref{fig:sbotrans}) to the superstructure components of the machine.
\end{description}

With regard to the infrastructure translation described in Figure~\ref{fig:sentrans}, note that the invariant and event sentences have
been syntactically sugared in order to simplify the notation.  In
this and future examples, we use
\bkw{thenAct} in place of the \bkw{then} Event-B keyword to
distinguish it from the \ikw{then} specification-building operator.

\subsection{The First Refinement}
In Figure~\ref{fig:ebm1} we presented a refinement of the abstract specification that was introduced in Figure~\ref{fig:ebm0}. Using our {\EBtoEVT} translator tool, we obtain the {\iEVT}-specification corresponding to this machine as shown in Figure~\ref{fig:evtref1}.

\begin{description}
\item[Lines 1--3] The application of $\mathbb{B}$ to the
  \texttt{refines} and \texttt{sees} clauses of Figure~\ref{fig:ebm1} (lines 2--3) results in the application of the specification-building operators (\ikw{and}, \ikw{then}) to include the abstract machine and context specifications into the generated {\iEVT}-specification.
\item[Lines 4--9] The new variables and (non-theorem) invariants are included as operations and predicates respectively. For readability, we use a simple notation for the variant expression on line 9. The reader should understand that the semantics of this is evaluated as described in Figure~\ref{fig:sentrans}.
\item[Lines 10--31] $\mathbb{R}_{\Sigma}$ returns (\texttt{m$o$ \ikw {hide via}} $\sigma_h$) \texttt{\ikw{with}} $\sigma_m$ for each refined event. The event names have remained the same during event refinement in this example, and since events with the same name are implicitly merged we can omit this notation here.
\end{description}

\begin{figure}
\begin{minipage}{0.45\textwidth}
\begin{programs}
\SPECH{cd}
  \ikw{sort} $\mathbb{N}$
  \ikw{ops}  d:$\mathbb{N}$
  .  d $> 0$
\ikw{end}

\SPECH{m0}
  \textsc{cd}
  \ikw{then}
    ops  n:$\mathbb{N}$
    .  $n\mathbin{\leq}d$
    \EVENTS
      \INITIALISATION
        \bkw{thenAct}  n := 0
      \EVT{ML\_out}{ordinary}
        \bkw{when} $n<d$
        \bkw{thenAct}  n := n+1
      \EVT{ML\_in}{ordinary}
        \bkw{when} $n>0$
        \bkw{thenAct}  n := n-1
\ikw{end}
\end{programs}

\caption{A structured specification in the {\iEVT} institution (syntactically
  sugared)  as generated by {\EBtoEVT}.  This specification corresponds
  to the Event-B specification of the abstract machine given in Figure~\ref{fig:ebm0}.}%
\label{fig:evtm0}
\end{minipage}
\qquad
\begin{minipage}{0.45\textwidth}
\begin{programs}
\SPECH{m1}
  \textsc{m0} \ikw{and} \textsc{cd}
    \ikw{then}
    \ikw{ops}  a:$\mathbb{N}$
         b:$\mathbb{N}$
         c:$\mathbb{N}$
    .  $n=a+b+c$
       $a=0 \lor c=0$
    \bkw{variant} $2*a+b$
    \EVENTS
      \INITIALISATION
        \bkw{thenAct}  a := 0
                 b := 0
                 c := 0
      \EVT{ML\_out}{ordinary}
        \bkw{when} $a+b<d$
              $c=0$
        \bkw{thenAct} a := a+1
      \EVT{IL\_in}{convergent}
        \bkw{when}  $a>0 $
        \bkw{thenAct} a := a-1
                b := b+1
      \EVT{IL\_out}{convergent}
        \bkw{when} $0<b $
              $a=0$
        \bkw{thenAct} b := b-1
                c := c+1
      \EVT{ML\_in}{ordinary}
        \bkw{when}  $c>0$
        \bkw{thenAct}  c := c-1
\ikw{end}

\end{programs}
\caption{An {\iEVT} specification corresponding to the results of the
  first refinement step, as represented by the Event-B specification in
  Figure~\ref{fig:ebm1}.}%
\label{fig:evtref1}

\end{minipage}
\end{figure}

\begin{figure}
\begin{minipage}{0.47\textwidth}
\begin{programsc}
\CONTEXT{Color}
  \SETS Color
  \CONSTANTS
    red, green
  \AXIOMS
    \bLabel{axm4}{Color = \{green,red\}}
    \bLabel{axm3}{green \neq red}
\END

\MACHINE{m2}
  \REFINES{m1}
  \SEES{cd, Color}
  \VARIABLES a, b, c, ml\_tl, il\_tl,
    il\_pass, ml\_pass
  \INVARIANTS
   \bLabel{inv1}{ml\_tl \in \{red,green\}}
   \bLabel{inv2}{il\_tl \in \{red,green\}}
   \bLabel{inv3}{ml\_tl=green \Rightarrow c=0}
   \bLabel{inv12}{ml\_tl=green \Rightarrow a+b+c<d}
   \bLabel{inv4}{il\_tl=green \Rightarrow a=0}
   \bLabel{inv11}{il\_tl=green \Rightarrow b>0}
   \bLabel{inv6}{il\_pass\in\{0,1\}}
   \bLabel{inv7}{ml\_pass\in\{0,1\}}
   \bLabel{inv8}{ml\_tl=red \Rightarrow ml\_pass=1}
   \bLabel{inv9}{il\_tl=red \Rightarrow il\_pass=1}
   \bLabel{inv5}{il\_tl=red \lor ml\_tl=red}
   \bLabel{thm2}{0\mathbin{\geq}a \Rightarrow a=0} theorem
   \bLabel{thm3}{0\mathbin{\geq}b \Rightarrow b=0} theorem
   \bLabel{thm4}{0\mathbin{\geq}c \Rightarrow c=0} theorem
   \bLabel{thm5}{\lnot(d\mathbin{\leq}0)} theorem
   \bLabel{thm6}{b+1\mathbin{\geq}d \land \lnot(b+1=d)}
          $ \Rightarrow \lnot(b<d)$ theorem
   \bLabel{thm7}{b\mathbin{\leq}1 \land \lnot(b=1)}
          $ \Rightarrow \lnot(b>0)$ theorem
   \bLabel{thm1}{(ml\_tl= green \land a+b+1<d)}
          $\lor (ml\_tl=green \land a+b+1=d)$
          $\lor (il\_tl=green \land b>1)$
          $\lor (il\_tl=green \land b=1)$
          $\lor (ml\_tl=red \land a+b<d$
              $\land c=0 \land il\_pass=1)$
          $\lor (il\_tl=red \land 0<b \land a=0 $
              $\land ml\_pass=1)$
          $\lor 0<a \lor 0<c$ theorem
  \VARIANT ml\_pass + il\_pass
  \EVENTS
    \INITIALISATION
      \bkw{then} \bLabel{act2}{a := 0}
           \bLabel{act3}{b := 0}
           \bLabel{act4}{c := 0}
           \bLabel{act1}{ml\_tl := red}
           \bLabel{act5}{il\_tl := red}
           \bLabel{act6}{ml\_pass := 1}
           \bLabel{act7}{il\_pass := 1}

\end{programsc}
\end{minipage}
\qquad
\begin{minipage}{0.45\textwidth}
\begin{programsc}\setcounter{@@lineno}{52}
    \EVT{ML\_out1}{ordinary}
      \REFINES{ML\_out}
      \bkw{when} \bLabel{grd1}{ml\_tl = green}
           \bLabel{grd2}{a+b+1<d}
      \bkw{then} \bLabel{act1}{a := a+1}
           \bLabel{act2}{ml\_pass := 1}
    \EVT{ML\_out2}{ordinary}
      \REFINES{ML\_out}
        \bkw{when} \bLabel{grd1}{ml\_tl=green}
             \bLabel{grd2}{a+b+1=d}
        \bkw{then} \bLabel{act1}{a := a+1}
             \bLabel{act2}{ml\_tl := red}
             \bLabel{act3}{ml\_pass := 1}
     \EVT{IL\_out1}{ordinary}
       \REFINES{IL\_out}
       \bkw{when} \bLabel{grd1}{il\_tl=green}
            \bLabel{grd2}{b>1}
       \bkw{then} \bLabel{act1}{b := b-1}
            \bLabel{act2}{c := c+1}
            \bLabel{act3}{il\_pass := 1}
     \EVT{IL\_out2}{ordinary}
       \REFINES{IL\_out}
       \bkw{when} \bLabel{grd1}{il\_tl=green}
            \bLabel{grd2}{b=1}
       \bkw{then} \bLabel{act1}{b := b-1}
            \bLabel{act2}{il\_tl := red}
            \bLabel{act3}{c := c+1}
            \bLabel{act4}{il\_pass := 1}
     \EVT{ML\_tl\_green}{convergent}
       \bkw{when} \bLabel{grd1}{ml\_tl=red}
            \bLabel{grd2}{a+b<d}
            \bLabel{grd3}{c=0}
            \bLabel{grd4}{il\_pass=1}
       \bkw{then} \bLabel{act1}{ml\_tl := green}
            \bLabel{act2}{il\_tl := red}
            \bLabel{act3}{ml\_pass := 0}
     \EVT{IL\_tl\_green}{convergent}
       \bkw{when} \bLabel{grd1}{il\_tl=red}
            \bLabel{grd2}{0<b}
            \bLabel{grd3}{a=0}
            \bLabel{grd4}{ml\_pass=1}
       \bkw{then} \bLabel{act1}{il\_tl := green}
            \bLabel{act2}{ml\_tl := red}
            \bLabel{act3}{il\_pass := 0}
      \EVT{IL\_in}{ordinary}
        \REFINES{IL\_in}
        \bkw{when} \bLabel{grd11}{0<a}
        \bkw{then} \bLabel{act11}{a := a-1}
             \bLabel{act12}{b := b+1}
      \EVT{ML\_in}{ordinary}
        \REFINES{ML\_in}
        \bkw{when} \bLabel{grd1}{0<c}
        \bkw{then} \bLabel{act1}{c := c+1}
\END
\end{programsc}
\end{minipage}
\caption{An Event-B specification of a machine \texttt{m2} that
  refines the Event-B machine in Figure~\ref{fig:ebm1} by adding new
  events \texttt{Ml\_tl\_green} and \texttt{Il\_tl\_green}. The
  context \texttt{Color} on lines 1--8 adds a new data type which is
  used by the \texttt{ml\_tl} and \texttt{il\_tl} traffic light variables.}%
\label{fig:ebm2}
\end{figure}

\begin{figure}
\begin{minipage}{0.45\textwidth}
\begin{programsc}
\SPECH{COLOR}
  \ikw{sorts}  Color
  \ikw{ops}  red: Color
       green: Color
  .  Color = \{green,red\}
     green $\neq$ red
\ikw{end}

\SPECH{m2}
  \textsc{m1} \ikw{and} \textsc{cd} \ikw{and} \textsc{COLOR} \ikw{and}
  (\textsc{m1} \ikw{ hide via } $\sigma_{h_{\texttt{ ML\_out}}})$ \ikw{with} $ \sigma_{m_{\texttt{ ML\_out}\mapsto\texttt{ML\_out1}}}$ \ikw{and}
  (\textsc{m1} \ikw{ hide via } $\sigma_{h_{\texttt{ ML\_out}}})$ \ikw{with} $ \sigma_{m_{\texttt{ ML\_out}\mapsto\texttt{ML\_out2}}}$ \ikw{and}
  (\textsc{m1} \ikw{ hide via } $\sigma_{h_{\texttt{ IL\_out}}})$ \ikw{with} $ \sigma_{m_{\texttt{ IL\_out}\mapsto\texttt{IL\_out1}}}$ \ikw{and}
  (\textsc{m1} \ikw{ hide via } $\sigma_{h_{\texttt{ IL\_out}}})$ \ikw{with} $ \sigma_{m_{\texttt{ IL\_out}\mapsto\texttt{IL\_out2}}}$
  \ikw{then}
    \ikw{ops}  ml\_tl : \{red,green\}
      il\_tl :  \{red,green\}
      il\_pass :  \{0,1\}
      ml\_pass :  \{0,1\}
    .  ml\_tl=green $\Rightarrow$ c=0
       ml\_tl=green $\Rightarrow$ a+b+c<d
       il\_tl=green $\Rightarrow$ a=0
       il\_tl=green $\Rightarrow$ b>0
       ml\_tl=red $\Rightarrow$ ml\_pass=1
       il\_tl=red $\Rightarrow$ il\_pass=1
       il\_tl=red $\lor$ ml\_tl=red
    \bkw{variant}  ml\_pass+il\_pass
    \EVENTS
      \INITIALISATION
        \bkw{thenAct}  a := 0
                 b := 0
                 c := 0
                 ml\_tl := red
                 il\_tl := red
                 ml\_pass := 1
                 il\_pass := 1
    \EVT{ML\_out1}{ordinary}
        \bkw{when}  ml\_tl=green
              a+b+1<d
        \bkw{thenAct}  a := a+1
                 ml\_pass := 1

\end{programsc}
\end{minipage}
\qquad
\begin{minipage}{0.45\textwidth}
\begin{programsc}\setcounter{@@lineno}{40}
    \EVT{ML\_out2}{ordinary}
        \bkw{when}  ml\_tl=green
              a+b+1=d
        \bkw{thenAct}  a := a+1
                 ml\_tl := red
                 ml\_pass := 1
    \EVT{IL\_out1}{ordinary}
        \bkw{when}  il\_tl=green
              b>1
        \bkw{thenAct}  b := b-1
                 c := c+1
                 il\_pass := 1
    \EVT{IL\_out2}{ordinary}
        \bkw{when}  il\_tl=green
              b=1
        \bkw{thenAct}  b := b-1
                 il\_tl := red
                 c := c+1
                 il\_pass := 1
    \EVT{ML\_tl\_green}{anticipated}
        \bkw{when}  ml\_tl=red
              a+b<d
              c=0
              il\_pass=1
        \bkw{thenAct}  ml\_tl := green
                 il\_tl := red
                 ml\_pass := 0
    \EVT{IL\_tl\_green}{anticipated}
        \bkw{when}  il\_tl=red
              0<b
              a=0
              ml\_pass=1
        \bkw{thenAct}  il\_tl := green
                 ml\_tl := red
                 il\_pass := 0
    \EVT{IL\_in}{ordinary}
        \bkw{when}  0<a
        \bkw{thenAct}  a := a-1
                 b := b+1
    \EVT{ML\_in}{ordinary}
        \bkw{when}  0<c
        \bkw{thenAct} c := c-1
\ikw{end}
\end{programsc}
\end{minipage}
\caption{A structured specification in the {\iEVT} institution corresponding to
  the Event-B specification of Figure~\ref{fig:ebm2}, representing the
  results of the second refinement step.}%
\label{fig:evtref2}
\end{figure}

\subsection{The Second Refinement}

Figure~\ref{fig:ebm2} contains an Event-B specification that shows the
second refinement step, which adds a substantial amount of new
behaviour to the specification. The first addition is a context (lines
1--8) containing colours which provide values for the new variables
(\texttt{ml\_tl} and \texttt{il\_tl}). These are used to control the
behaviour of a pair of traffic lights (line 12). The new variables
\texttt{ml\_pass} and \texttt{il\_pass} on line 13 are used to record whether or
not a car has passed through the mainland or island traffic light
respectively. The objective of these variables is to ensure that the
traffic lights do not change so rapidly that no car can pass through
them~\cite{abrial_modeling_2010}.  Events for these lights were added
to the machine and the current events were modified to account for
this behaviour. In particular, each of the abstract events
(\texttt{ML\_out} and \texttt{IL\_out} in Figure~\ref{fig:ebm1}) are
refined by two events (\texttt{ML\_out1}, \texttt{ML\_out2},
\texttt{IL\_out1} and \texttt{IL\_out2} in Figure~\ref{fig:ebm2}).

As can be seen from Figure~\ref{fig:ebm2}, there is a good degree of
repetition in the Event-B specification resulting from common cases
for the numeric variables.  These can be represented much more
succinctly as a structured specification in our institution as shown in Figure~\ref{fig:evtref2}.
\begin{description}
\item[Lines 1--7] This {\iCASL} specification specifies the new data type $\texttt{Color}$. This corresponds to the context described in the Event-B model in Figure~\ref{fig:ebm2}.
\item[Lines 8--14] The superstructure component of the translation includes the abstract machine (\texttt{m$0$}), and the contexts (\texttt{cd} and \texttt{COLOR}). In this case the names of the events are changed in the refinement step. The event \texttt{ML\_out} is decomposed into the events \texttt{Ml\_out1} and \texttt{ML\_out2}. The event \texttt{IL\_out} follows a similar decomposition into \texttt{IL\_out1} and \texttt{IL\_out2}. We account for this refinement using the signature morphisms and specification-building operators on lines 10--13. We have not explicitly described these signature morphisms but their functionality should be obvious from their respective subscripts.
\item[Line 26] As before, we use a simple variant notation for the variant expression.
\item[Lines 27--83] The events are translated as before for the remainder of the specification.
\end{description}
This example illustrates the use of {\EBtoEVT} as a means for bridging
the gap between the {\Rodin} and {\Hets} software eco-systems. The
source code for {\EBtoEVT} along with more examples of its use can be found on our git repository\footnote{\url{https://github.com/mariefarrell/EB2EVT.git}}.

\section{Enhanced Refinement and Modularisation in {\iEVT}}%
\label{sec:refmod}

Our formalisation of Event-B in terms of the theory of institutions offers a number of benefits. In this section, we provide a discussion of refinement in this setting and how we can exploit the use of specification-building operators for modularisation of Event-B specifications.
\subsection{Representing Refinement Explicitly}
As is evidenced by our running example and Section~\ref{sec:eb}, refinement is a central aspect of the Event-B methodology, therefore any formalisation of Event-B must be capable of capturing refinement. The theory of institutions equips us with a basic notion of refinement as \emph{model-class inclusion} where the class of models of a specification comprises the models that satisfy the specification. We formally define what is meant by \textit{model-class} in Definition~\ref{def:modclass}.

\begin{defi}[Model-Class]%
\label{def:modclass}
For any $\Phi \subseteq\mathbf{Sen}(\Sigma)$ of $\Sigma$-sentences, the class $Mod_{\Sigma}(\Phi) \subseteq\lvert\mathbf{Mod}(\Sigma)\rvert$ of models of $\Phi$ is defined as the class of all $\Sigma$-models that satisfy all the sentences in $\Phi$.
\end{defi}
Note that we omit the $\Sigma$ subscript where the signature is clear from the context.
With regard to institution-theoretic refinement, the class of models of the concrete specification is a subset of the class of models of the abstract specification~\cite[\S 7.1]{sanella_foundations_2012}.

We consider two cases: (1) when the signatures are the same and (2) when the signatures are different. In the case where the signatures are the same, refinement is denoted as:\\
\centerline{$SP_A \sqsubseteq SP_C \quad\iff\quad Mod(SP_C) \subseteq Mod(SP_A)$}
where $SP_A$ is an abstract specification that refines ($\sqsubseteq$) to a concrete specification $SP_C$, $Mod(SP_A)$ and $Mod(SP_C)$ denote the class of models of the abstract and concrete specifications, respectively. This means that, given specifications with the same signature, the concrete specification should not exhibit any model that was not possible in the abstract specification. This equates to the definition of general refinement outlined in~\cite{reeves_general_2008,reeves_general_2008-1}.

In terms of Event-B, data refinement is not supported in this way because it would involve a change of signature. Instead, refinement occurs between machines by strengthening the invariants and the guards but does not provide any new variables or events. In this case, the specification is made more deterministic by further constraining the invariants and guards. This increase in determinism is supported by the refinement calculi~\cite{back1990refinement, back_refinement_1998, morgan_refinement_1988, morris_theoretical_1987}.

When the signatures are different and related by a signature morphism $\sigma: Sig[SP_A] \rightarrow Sig[SP_C]$ then we can use the corresponding model morphism in order to express refinement. This model morphism is used to interpret the concrete specification as containing only the signature items from the abstract specification. Here, refinement is the model-class inclusion of the models of the concrete specification restricted into the class of models of the abstract specification using the model morphism. In this case we write:
\[SP_A \sqsubseteq SP_C \quad\iff\quad Mod(\sigma)(SP_C) \subseteq Mod(SP_A)\]
where $Mod(\sigma)(SP_C)$ is the model morphism applied to the model-class of the concrete specification $SP_C$. This interprets each of the models of $SP_C$ as models of $SP_A$ before a refinement relationship is determined. Note that we could have alternatively represented this kind of refinement using the \ikw{hide via} specification-building operator to transform $SP_C$ into a specification over the signature $Sig[SP_A]$ using the appropriate morphism.

In this scenario, both data refinement and superposition refinement are supported directly in {\iEVT}, since the refined or added variables and events can be removed, using a suitable model morphism as outlined above, before refinement is completed. Schneider et al.\ used a similar notion of hiding when they provided a CSP semantics for Event-B refinement~\cite{schneider_behavioural_2014}.

\begin{figure}
\begin{programs}
\REFS{REF0}{M0}{M1}
   ML\_in $\mapsto$ ML\_in, ML\_out $\mapsto$ ML\_out
\ikw{end}

\REFS{REF1A}{M1}{M2}
   ML\_in $ \mapsto$ ML\_in, ML\_out $\mapsto$ ML\_out1,
   IL\_in $ \mapsto$ IL\_in, IL\_out $\mapsto$ IL\_out1
\ikw{end}

\REFS{REF1B}{M1}{M2}
   ML\_in $\mapsto$ ML\_in, ML\_out $\mapsto$ ML\_out2,
   IL\_in $\mapsto$ IL\_in, IL\_out $\mapsto$ IL\_out2,
\ikw{end} \end{programs}
\caption{A specification in the {\iEVT} institution that explicitly
  defines the refinement relationships between the abstract machine
  and its two concrete refinements. Note that, for readability, we omit the status mappings since all statuses are \texttt{ordinary} in this case and remain as so after signature morphism.}%
\label{refinement}
\end{figure}

Figure~\ref{refinement} uses the refinement syntax available in {\Hets} to specify each of the refinement steps in the cars on a bridge example.  Note that this allows us to express the refinement as a separate entity, whereas in Event-B the refinement step is woven into the definition of the refined machine.
In Figure~\ref{refinement}, lines 1--3 define the refinement step between the abstract specification \texttt{\textsc{m0}} and the concrete specification \textsc{m1}. In this step the events are refined but no renaming or statuses are changed so it is very straightforward.
Lines 4--11 define the refinement step between the specification \texttt{\textsc{m1}} and the concrete specification \texttt{\textsc{m2}}. There are two separate refinements (\texttt{REF1A} and \texttt{REF1B}) because the events \texttt{ML\_out} and \texttt{IL\_out} are split into two further events each as presented in Figure~\ref{fig:evtref2}.

The example that we presented in this paper does not make use of data refinement in Event-B but if it did, it would also be possible to capture this kind of refinement in the {\Hets}-style notation above. An example of this was presented in our previous work~\cite{farrell_institution_2017}.

\subsection{Modularisation using Specification-Building Operators}%
\label{sec:modular}

During the evolution of the Event-B formalism from Classical-B, certain facilities for the reuse of machine specifications disappeared. These include the modularisation properties supplied by the keywords \texttt{INCLUDES} and \texttt{USES} which facilitated the use of an existing machine in other developments~\cite{silva_supporting_2009}. Since then there have been numerous attempts at regaining such modularisation features~\cite{Silva2012,gondal_feature_2009,gondal_composing_2011,hoang2011survey,Butler2009, poppleton2008composition}. Modularisation in Event-B does not require re-engineering the language in the same way as Object-Z or VVSL, but rather, modifications are made by building plugins for {\Rodin}.

Modularisation, in terms of decomposition, in Event-B was first described by Abrial as the act of cutting a large system of events into smaller pieces that can each be refined independently of the others~\cite{abrial_refinement_2007}. This offers the advantage of separation of concerns even though there still needs to be a link between the parts. The central idea is that the specifications constructed via decomposition and further refinement may be easily recombined resulting in a specification that could have been obtained without using any decomposition techniques in the first place. Three primary types of Event-B decomposition have been identified~\cite{hoang2011survey}. These are \textit{shared variable}~\cite{abrial2009event}, \textit{shared event}~\cite{Butler2009} and \textit{modularisation}~\cite{iliasov_modules_2010}. Another related approach is \textit{generic instantiation}~\cite{abrial_refinement_2007}.

It is clear, from observations of industrial projects and the sheer volume of {\Rodin} plugins developed for modularisation in Event-B, that there is an underlying demand for modularisation~\cite{iliasov_modules_2010,hoang2011survey}.
In previous work, we assembled and analysed a corpus of Event-B specifications in light of their metrics and the incidence of specification clones throughout the corpus~\cite{farrell_clones_2017}. This work supported the claim that the Event-B formal specification language is in need of a set of well-defined modularisation constructs and we postulated that the specification-building operators of the theory of institutions could fill this void.

We have used the specification-building operators throughout our definitions in Section~\ref{sec:transsemantics} as a way of naturally adding modularisation to these specifications and this shines through in the example that we have shown above. Notably, we can use these specification-building operators to further modularise the specifications in our illustrative example. In particular, the {\iEVT}-specification in Figure~\ref{fig:ebm1} could be further modularised as illustrated in Figure~\ref{fig: modularm1}.

\begin{figure}
\begin{minipage}{0.35\textwidth}
\begin{programsc}
\SPECH{datam0}
  \textsc{cd} \ikw{then}
    ops n:$\mathbb{N}$
    . n $\leq$ d
    \EVENTS
      \INITIALISATION
        \bkw{thenAct} n := 0
\ikw{end}

\SPECH{datam1}
  \textsc{datam0} \ikw{then}
    \ikw{ops} a, b, c : $\mathbb{N}$
    . n = a + b + c
      a = 0 $\lor$ c = 0
    \bkw{variant} 2$*$a+b
    \EVENTS
      \INITIALISATION
        \bkw{thenAct}  a := 0
                 b := 0
                 c := 0
\ikw{end}
\end{programsc}
\end{minipage}
\qquad
\begin{minipage}{0.55\textwidth}
\begin{programsc}\setcounter{@@lineno}{20}
\SPECH{inout}
  ops v1, v2 : $\mathbb{N}$
  \EVENTS
    \EVT{out}{ordinary}
      \bkw{when} v1 = 0
      \bkw{thenAct} v2 := v2 + 1
    \EVT{in}{ordinary}
      \bkw{when} v1 > 0
      \bkw{thenAct} v1 := v1 + 1
\ikw{end}

\SPECH{m1}
  \textsc{datam1} \ikw{and} \textsc{m0} \ikw{and}
  \textsc{inout} \ikw{with} \{$\langle$out, ordinary$\rangle$ $\mapsto$ $\langle$ML\_out, ordinary$\rangle$,
               $\langle$in, ordinary$\rangle$ $\mapsto$ $\langle$ML\_in, ordinary$\rangle$,
                v1 $\mapsto$ c, v2 $\mapsto$ a\} \ikw{and}
  \textsc{inout} \ikw{with} \{$\langle$out, ordinary$\rangle$ $\mapsto$ $\langle$IL\_out, convergent$\rangle$,
              $\langle$in, ordinary$\rangle$ $\mapsto$ $\langle$IL\_in, convergent$\rangle$,
               v1 $\mapsto$ a, v2 $\mapsto$ c\}
   \ikw{then}
     \EVENTS
       \EVT{ML\_out}{ordinary}
         \bkw{when} a + b < d
       \EVT{IL\_in}{convergent}
         \bkw{thenAct} b:= b + 1
       \EVT{IL\_out}{convergent}
         \bkw{when} 0 < b
         \bkw{thenAct} b := b + 1
\ikw{end}
\end{programsc}
\end{minipage}
\caption{A structured specification in the {\iEVT} institution that is an even more
  modular version of the structured specification in Figure~\ref{fig:evtref1}. }%
\label{fig: modularm1}
\end{figure}

\begin{description}
\item[Lines 1--20] The specifications \texttt{\textsc{datam0}} (lines 1--8) and \texttt{\textsc{datam1}} (lines 9--20) capture the data parts of the \texttt{\textsc{m0}} and \texttt{\textsc{m1}} specifications and their \texttt{Initialisation} events (see Figures~\ref{fig:evtm0} and~\ref{fig:evtref1} for correspondence).
\item[Lines 21--30] Since the \texttt{ML$\_$in} and \texttt{ML$\_$out} events are quite similar, modulo renaming of variables, to the \texttt{IL$\_$in} and \texttt{IL$\_$out} events we construct the \texttt{\textsc{inout}} specification. This defines an \texttt{out} (lines 24--26) and an \texttt{in} event (lines 27--29) using natural number variables \texttt{v1} and \texttt{v2} (introduced on line 22).
\item[Lines 31--48] This is a modular specification of \texttt{\textsc{m1}} that instantiates two copies of the \texttt{\textsc{inout}} specification with the appropriate signature morphisms mapping the \texttt{out} event to \texttt{ML$\_$out} in the first instance and to \texttt{IL$\_$out} in the second. Correspondingly, \texttt{in} maps to \texttt{ML$\_$in} and then to \texttt{IL$\_$in}. Further to this, the sentences after the \ikw{then} keyword (lines 40--48) add extra behaviour to these events that was not made explicit in the \texttt{\textsc{inout}} specification.
\end{description}

This specification is evidently more modular than the original version by modelling the simple behaviour of coming in and going out in a separate specification and then instantiating it twice with the appropriate events and variables renamed. This kind of modularisation of specifications is currently not available in the Event-B formal specification language, but our {\iEVT} institution allows us to enhance Event-B with these powerful modularisation constructs.

As already mentioned, a number of plugins exist with the objective of providing modularisation features for Event-B. However, although all of these plugins offer some form of modularisation for Event-B specifications, it is not clear how specifications developed utilising more than one of these plugins could be combined in practice, or, if this is even possible. 

\subsection{Capturing Current Event-B Modularisation}

In this subsection, we illustrate how the functionality of a selection
of the current \textsf{Rodin} modularisation plugins can be captured
using the specification-building operators and parametrisation made
available in the theory of institutions using {\iEVT}.
Most of these plugins embody a shared variable or shared event style
of modularisation (e.g~\cite{Silva2012,Butler2009,iliasov_modules_2010}), and this section
demonstrates how our semantics provides a theoretical foundation for
current modularisation approaches in Event-B.  A more thorough
analysis of the current modularisation approaches available in the
\textsf{Rodin} Platform can be found in~\cite{farrell_phdthesis}.

\subsubsection{Shared Variable and Shared Event Decomposition}

Modularisation is straightforward when the events in a machine each
reference their own, separate set of variables.  However, the more common
situation is for events to share variables, as is illustrated by the
machine \texttt{M} at the top of Figures~\ref{fig: sv} and~\ref{fig: se}.  In
machine \texttt{M}, each of the three variables $\{\mathtt{v1}, \mathtt{v2}, \mathtt{v3}\}$ is
referenced by two different events (we have omitted the event status and
variable sorts from this example to increase clarity).  In each case
we wish to split machine \texttt{M} into two equivalent sub-machines \texttt{M1} and
\texttt{M2}, and Figures~\ref{fig: sv} and~\ref{fig: se} show two
different approaches to achieving this in Event B.

Figure~\ref{fig: sv} illustrates the \textit{shared variable}
approach, where an Event-B machine is decomposed into sub-machines
by grouping events which share the same variables~\cite{abrial_refinement_2007}.
Since this may not produce a partition, the user must choose which
events to allocate to a sub-machine, e.g., in Figure~\ref{fig: sv} we
allocate events \texttt{e1} and \texttt{e2} to sub-machine \texttt{M1}
and events \texttt{e3} and \texttt{e4} to sub-machine \texttt{M2}.
Since \texttt{e2} and \texttt{e3} both
refer to \texttt{v2} in Figure~\ref{fig: sv}, `external' events
\texttt{e3\_e} and \texttt{e2\_e} are added to
abstractly describe the behaviour of the shared variable in the other
sub-machine.


\begin{figure}
\centering
\begin{minipage}{0.45\textwidth}
\centering
\includegraphics[scale=0.5]{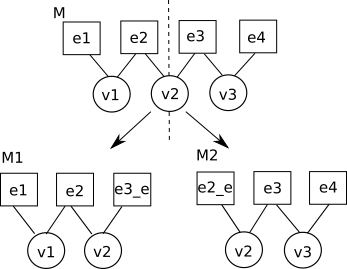}
\caption{The \emph{shared variable} decomposition of machine
  \texttt{M} into sub-machines \texttt{M1} and \texttt{M2}. Here we
  wish to partition machine \texttt{M} into two sets of events, as shown by
  the dotted line. }%
\label{fig: sv}
\end{minipage}
\qquad
\begin{minipage}{0.45\textwidth}
\begin{programsc}

\SPECH{M1}
  (\textsc{M} \ikw{hide via} $ \sigma_1$)
    \ikw{with} \{\,e3 $\mapsto$ e3\_e\,\}
\ikw{end}
  where $\sigma_1$ = \{\,v1 $\mapsto$ v1, v2 $\mapsto$ v2, e1 $\mapsto$ e1,
               e2 $\mapsto$ e2, e3 $\mapsto$ e3\,\}

\SPECH{M2}
  (\textsc{M} \ikw{hide via} $ \sigma_2$)
    \ikw{with} \{\,e2 $\mapsto$ e2\_e\,\}
\ikw{end}
  where $\sigma_2$ = \{\,v2 $\mapsto$ v2, v3 $\mapsto$ v3, e2 $\mapsto$ e2,
               e3 $\mapsto$ e3, e4 $\mapsto$ e4\,\}
\end{programsc}
\caption{Representing the \emph{shared variable} modularisation used
  in Figure~\ref{fig: sv} using specification-building operators.}%
\label{fig: sbodecomp}
\end{minipage}

\end{figure}

Figure~\ref{fig: sbodecomp} recasts the shared variable example of
Figure~\ref{fig: sv} using specification-building operators.
The specification-building operator \ikw{hide via}, used on lines 2
and 8 of Figure~\ref{fig: sbodecomp}, is used to hide auxiliary
signature items (variables and events in this case) by means of
omitting them from the signature morphisms $\sigma_1$ and
$\sigma_2$.  The two external events are created using the signature
morphisms defined in lines 3 and 9 of Figure~\ref{fig: sbodecomp}.

In contrast, Figure~\ref{fig: se} illustrates \textit{shared event}
decomposition, which partitions an Event-B machine
based on variables which participate in the same events~\cite{abrial_refinement_2007}.
In this example, a clean partition is not possible, and the user has
chosen to allocate variable \texttt{v1} to sub-machine \texttt{M1}
and variables \texttt{v2} and \texttt{v3} to sub-machine \texttt{M2}.
The problem now is caused by event \texttt{e2} which uses variables
\texttt{v1} and \texttt{v2} and thus cannot be cleanly allocated to
either machine.   The solution is to decompose event \texttt{e2} into
two events \texttt{e2{\_1}} and \texttt{e2{\_2}}. These new events
each correspond to a restricted version of the original event
\texttt{e2}, each confined to the guards and actions that refer to
each of the variable sets (i.e. \texttt{e2{\_1}} contains only the
guards and actions of \texttt{e2} that referenced \texttt{v1}).

\begin{figure}

\begin{minipage}{0.45\textwidth}
\centering
\includegraphics[scale=0.5]{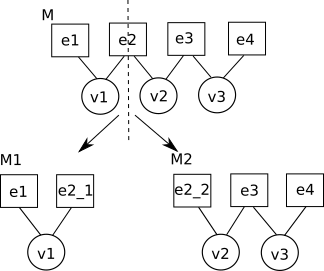}
\caption{The \emph{shared event} decomposition of machine \texttt{M}
  into sub-machines \texttt{M1} and \texttt{M2}. Here we
  wish to partition machine \texttt{M} into two sets of variables, as shown by
  the dotted line.}%
\label{fig: se}
\end{minipage}
\qquad
\begin{minipage}{0.45\textwidth}
\begin{programsc}
\SPECH{M1}
  (\textsc{M} \ikw{hide via} $\sigma_1$)
    \ikw{with} \{\,e2 $\mapsto$ e2\_1\,\}
\ikw{end}
  where
    $\sigma_1$ = \{\,v1 $\mapsto$ v1, e1 $\mapsto$ e1, e2 $\mapsto$ e2\,\}

\SPECH{M2}
  (\textsc{M} \ikw{hide via} $\sigma_2$)
    \ikw{with} \{\,e2 $\mapsto$ e2\_2\,\}
\ikw{end}
  where
    $\sigma_2$ = \{\,v2 $\mapsto$ v2, v3 $\mapsto$ v3, e2 $\mapsto$ e2,
           e3 $\mapsto$ e3, e4 $\mapsto$ e4\,\}

\end{programsc}
\caption{Representing the \emph{shared event} modularisation used
  in Figure~\ref{fig: se} using specification-building operators}%
\label{fig: sboparallel}
\end{minipage}

\end{figure}

Figure~\ref{fig: sboparallel} represents this decomposition using
specification-building operators.  In each case we use hiding with a
signature morphism to specify the signature components that we wish to
keep, and then the shared event is renamed.  The identity mappings for
the variables on lines 6, 12 and 13 have the effect of hiding those
variables that belong to the other machine.  Thus these variables can
still be used in constraints, but are no longer externally visible.
Since the machines \texttt{M1} and \texttt{M2} do not share any common
variables they can now be refined independently.

\subsubsection{Generic Instantiation}%
\label{sec:gen}

The \textit{generic instantiation} approach to re-use in Event B is
intended to facilitate the parametrisation of machines, so that a
refinement chain can be re-used~\cite{abrial_refinement_2007}.  The
idea is that a chain of existing machine refinements is made generic
with respect to all of the carrier sets, constants and axioms that
have been collected in the contexts used throughout. It is then
possible for a user to instantiate these sets and constants
differently in a new model.  In order to reuse the proofs of the old
model, its axioms must then be proven as theorems after instantiation
in the new model. This process facilitates the reuse of existing
Event-B models thus providing a more modular approach to specification
in Event-B.


Figure~\ref{fig:geninsplugin} is an example of using the generic
instantiation plugin for {\Rodin}~\cite{silva_supporting_2009}. Here, a
machine \texttt{mac} has been specified (in outline) on lines 6--10,
and it depends on a context \texttt{ctx1} specified on lines 1--5.
This is standard Event B syntax, and the machine is not marked as
genericised in any way.  The instantiation is shown on lines 11--20 of
Figure~\ref{fig:geninsplugin}, where a machine \texttt{maci} is
specified (again, in outline), and declared to instantiate
\texttt{mac}, by providing a context \texttt{ctx2} as an
implementation of \texttt{ctx1} (lines 12--13).  The details of the
instantiation are made explicit in the following lines, where the
elements of the context \texttt{ctx1} are replaced by corresponding
elements from \texttt{ctx2} (lines 14--16), and the variables and events of
\texttt{mac} are renamed (lines 17--19).

The generic instantiation plugin generates extra proof
obligations to ensure that the instantiation of the original machine
is a valid one and that if the instantiated machine refines an
existing model that this is a valid refinement. The machine and its
instance are linked via refinement and once the proof obligations have
been carried out successfully then the rest of the pattern refinement
chain can be instantiated.

\begin{figure}
\begin{minipage}{0.49\textwidth}
\begin{programsc}
\CONTEXT{ctx1}
  \SETS S$_1$,\ldots,S$_m$
  \CONSTANTS C$_1$,\ldots,C$_n$
  \AXIOMS Ax$_1$,\ldots,Ax$_p$
\END

\MACHINE{mac}
  \SEES{ctx1}
  \VARIABLES v$_1$,\ldots,v$_q$
  \EVENTS ev$_1$,\ldots,ev$_r$
\END

\INSTANTIATEDMACHINE{maci}
  \INSTANTIATES{mac}{ctx1}
  \SEES{ctx2} //\textit{contains instance properties}
  \REPLACE \hfill//\textit{replace parameters defined in ctx1}
    \SETS S$_1$ := ctx2\_S$_1$,\ldots, S$_m$ := ctx2\_S$_m$
    \CONSTANTS C$_1$ := ctx2\_C$_2$,\ldots, C$_n$ := ctx2\_C$_n$
  \RENAME \hfill// \textit{rename elements in mac}
    \VARIABLES v$_1$ := ctx2\_v$_1$,\ldots, v$_q$ := ctx2\_v$_q$
    \EVENTS ev$_1$ := ctx2\_ev$_1$,\ldots, ev$_r$ := ctx2\_ev$_r$
\END

\end{programsc}
\caption{The machine \texttt{maci} is defined as an instantiation of
the machine \texttt{mac} using the generic instantiation plugin.}%
\label{fig:geninsplugin}
\end{minipage}
\qquad
\begin{minipage}{0.45\textwidth}
\begin{programsc}
\SPECH{ctx1}
  \ikw{sorts} S$_1$,\ldots, S$_m$
  \ikw{ops} C$_1$,\ldots, C$_n$
  . Ax$_1$,\ldots,Ax$_p$
\ikw{end}

\SPECH{mac}
  \textsc{ctx1} \ikw{then}
    \ikw{ops} v$_1$,\ldots,v$_q$
    \EVENTS ev$_1$,\ldots,ev$_r$
\ikw{end}

\SPECH{maci}
  (\textsc{ctx2} \ikw{then} \textsc{mac})
   \ikw{with}
     \{S$_1$ $\mapsto$ ctx2\_S$_1$,\ldots, S$_m$ $\mapsto$ ctx2\_S$_m$,
      C$_1$ $\mapsto$ ctx2\_C$_2$,\ldots, C$_n$ $\mapsto$ ctx2\_C$_n$,
      v$_1$ $\mapsto$ ctx2\_v$_1$,\ldots, v$_q$ $\mapsto$ ctx2\_v$_q$,
      ev$_1$ $\mapsto$ ctx2\_ev$_1$,\ldots, ev$_r$ $\mapsto$ ctx2\_ev$_r$\}
\ikw{end}

\end{programsc}

\caption{A representation of generic instantiation using
  specification-building operators in {\iEVT}.}%
\label{fig:genins}
\end{minipage}
\end{figure}


There are many ways of implementing this using the specification-building operators, and we provide one example in Figure~\ref{fig:genins}. In fact, the theory of institutions is equipped with
a concept of parameterisation and it is easy to see how we could have
used this to represent the generic instantiation in Figure~\ref{fig:genins} by giving the context, \texttt{ctx1} as a parameter
to the specification for \texttt{mac} with the appropriate signature
morphisms in the {\iEVT} institution.  However, this would involve
changing the original specification of \texttt{mac}, and so we have
chosen a specification in Figure~\ref{fig:genins} that more closely
resembles the approach shown in Figure~\ref{fig:geninsplugin}.
The main elements of the instantiation are the inclusion of
\texttt{ctx2} and \texttt{mac} in \texttt{maci} (line 12), along
with a series of appropriate renamings (lines 13--17) to match those in
the original.


In fact, the possibilities provided by the specification-building
operators are much more general than those provided by the generic
instantiation plugin.  Using the generic instantiation plugin, only
user defined sets can be replaced so it is not possible, for example,
to replace $\mathbb{Z}$, $\mathbb{N}$ or $BOOL$.  Moreover, it is a
requirement that \emph{all} sets and contexts be replaced ensuring
that there are no uninstantiated parameters.  One of the main benefits
of using specification-building operators over {\iEVT} is that it is
now possible to replace any sets that a user wishes, once the
appropriate signature morphisms have been defined by the user.  It is
also possible to generalise machines in {\iEVT} using any signature
elements, including other machines, as parameterisation is not limited
to just sorts and constants.

\section{Conclusions and Future Work}%
\label{sec:conclude}
In this paper, we provide a detailed description of the Event-B institution, {\iEVT}. 
Our definition of {\iEVT} allows the restructuring of Event-B specifications using the standard specification-building operators for institutions~\cite[\S 5.1]{sanella_foundations_2012}. Thus {\iEVT} provides a means for writing down and splitting up the components of an Event-B system, facilitating increased modularity for Event-B specifications. We have also shown how the institutional framework subsumes and can explicitly represent Event-B refinement.

Moreover, we contribute a formal semantics for the Event-B formal specification language. We achieved this by constructing a three-layer model for the Event-B language that decomposed it into its mathematical, infrastructure and superstructure sub-languages.
The  semantics that we have presented for each of these three constituent languages and thus for Event-B itself is grounded in the theory of institutions using our institution for Event-B ({\iEVT}) and the institution for first-order predicate logic with equality ({\iFOPEQ}). We have described a Haskell implementation of this semantics in our {\EBtoEVT} translator tool and illustrated how it is used to generate syntactically sugared {\iEVT}-specifications. We demonstrate our approach using the cars on a bridge example taken from Abrial's book~\cite{abrial_modeling_2010}.

Further to this, we have discussed refinement and illustrated the use of the specification-building operators as modularisation constructs in the {\iEVT} institution.
Institutions are designed with modularity and interoperability in mind and so our approach does not inhibit using multiple modelling domains to evaluate the semantics of an Event-B model. In fact, we have provided scope for the interoperability of Event-B with other formalisms that have been described in this institutional framework. Future work includes defining further institution-theoretic relationships between {\iEVT} and other formalisms that have been defined in the theory of institutions. This will provide a solid, mathematical foundation for their combination.

This work examined how the theory of institutions can provide a gateway to formalism-independent modularisation constructs for Event-B. The Event-B specification language does not natively embody strong modularisation features and one might wonder whether this work could motivate a more drastic modification to the Event-B formalism. It is our view that this work provides access to stronger modularisation for Event-B without the need to modify the formalism itself. This is particularly important since Event-B is industrial-strength and widely used in safety-critical systems already so dramatic modifications to the original language at this stage might cause issues in uptake. Although this paper is focused on Event-B, we have demonstrated how such modulrisation capabilities can be added to a formal specification language using the theory of institutions and future work seeks to examine how this approach can be applied to other similar formal languages. We have also illustrated the kinds of modularisation capabilities that are useful in a formal specification language and this work should inspire the development of moduilarisation features for these formalisms.
\bibliographystyle{alphaurl}
\bibliography{bibliography}

\clearpage
\appendix
\section*{Appendix}\renewcommand{\thesection}{A}

\signevt*
\begin{proof}
Let $\Sigma = \langle S, \Omega, \Pi, E, V\rangle$ be an {\iEVT}-signature where $\langle S, \Omega, \Pi \rangle$ is a signature over {\iFOPEQ}, $\Sigma_{{\iFOPEQ}}$, the institution for first-order predicate logic with equality~\cite[\S 4.1]{sanella_foundations_2012}. $E$ is a set of tuples of the form $(event \ name \mapsto status )$ and $V$ is a set of sort-indexed variable names.

In a category, morphisms can be composed, their composition is associative and identity morphisms exist for every object in the category. We show that $\mathbf{Sign}_{{\iEVT}}$ preserves these three properties:
\begin{enumerate}[(\alph*)]
\item Composition of ${\iEVT}$-signature morphisms:\\
{\iEVT}-signature morphisms can be composed, the composition of $\sigma_S, \sigma_\Omega$ and $\sigma_\Pi$ is inherited from {\iFOPEQ}. Therefore, we only examine the composition of $\sigma_E$ and $\sigma_V$.
	\begin{itemize}
	\item[$\sigma_E$:] Event names do not have a sort or arity, so the only restrictions are (1) on functions of the form $( \texttt{Init}\mapsto \texttt{ordinary})$ and (2) that the ordering in the status poset is preserved. In both of these cases the composition is well-defined for $\sigma_E$.
	\item[$\sigma_V$:] Variable names are sort-indexed so $\sigma_V$ utilises $\sigma_S$ on these sorts.
	\begin{align*}
	\sigma_2(\sigma_1(v:s)) &= \sigma_2((\sigma_{1_{V}}(v): \sigma_{1_S}(s))) \\
	&= \sigma_{2_{V}}(\sigma_{1_{V}}(v)): \sigma_{2_S}(\sigma_{1_S}(s))
	\end{align*}

	\end{itemize}

	Let $\sigma_1: \Sigma_1 \rightarrow \Sigma_2$ and $\sigma_2: \Sigma_2 \rightarrow \Sigma_3$, then we prove that $\sigma_2 \circ \sigma_1$ is a morphism in the category of {\iEVT}-signatures.
	\begin{itemize}
	\item For all $(e_1 \mapsto st_1 ) \in \Sigma_1.E_1$, $\sigma_1(e_1\mapsto st_1) \in \Sigma_2.E_2$ and for all $(e_2\mapsto st_2) \in \Sigma_2.E_2$,  $\sigma_2(e_2\mapsto st_2 ) \in \Sigma_3.E_3$. Therefore $\sigma_2(\sigma_1( e_1\mapsto st_1 )) \in \Sigma_3.E_3$ so \[\forall ( e_1\mapsto st_1 ) \in \Sigma_1.E_1  \Rightarrow \sigma_2 \circ \sigma_1 (e_1 \mapsto st_1 ) \in \Sigma_3.E_3\]
		\item For all $(v_1:s_1) \in \Sigma_1.V_1$, $\sigma_1((v_1:s_1)) \in \Sigma_2.V_2$ and for all $(v_2:s_2) \in\Sigma_2.V_2$, $\sigma_2((v_2:s_2))  \in \Sigma_3.V_3$. Therefore $\sigma_2(\sigma_1((v_1:s_1))) \in \Sigma_3.V_3$ so
		\[\forall (v_1:s_1) \in \Sigma_1.V_1\Rightarrow \sigma_2 \circ \sigma_1 ((v_1:s_1)) \in \Sigma_3.V_3\]
\end{itemize}
Therefore, {\iEVT}-signature morphisms can be composed.\\

\item Composition of ${\iEVT}$-signature morphisms is associative:
\[(\sigma_3 \circ \sigma_2) \circ \sigma_1 = \sigma_3 \circ (\sigma_2 \circ \sigma_1)\]

	For $(e\mapsto st )\in \Sigma.E$:
	\begin{align*}
	\sigma_2 \circ \sigma_1( e\mapsto st ) &= \sigma_2(\sigma_1(e\mapsto st ))
	\end{align*}
	Then, by definition of composition
	\begin{align*}
	\sigma_3 \circ (\sigma_2 \circ \sigma_1)(e\mapsto st ) &= \sigma_3(\sigma_2(\sigma_1(e\mapsto st ))) \\
	&= \sigma_3 \circ \sigma_2 \circ(\sigma_1(e\mapsto st ))\\
	& = (\sigma_3 \circ \sigma_2) \circ \sigma_1( e\mapsto st )
	\end{align*}

Similarly, for $(v:s) \in \Sigma.V$,  \[\sigma_3 \circ (\sigma_2 \circ \sigma_1)((v:s)) = (\sigma_3 \circ \sigma_2) \circ \sigma_1((v:s))\]
Therefore, the composition of {\iEVT}-signature morphisms is associative.

\item Identity morphisms for ${\iEVT}$-signatures:\\
For any {\iEVT}-signature $\Sigma$, there exists an identity signature morphism $id_{\Sigma}: \Sigma \rightarrow \Sigma$.
$id_E$ and $id_V$ are such that $id_E(e\mapsto st ) = (e\mapsto st )$ and $id_V((v:s)) = (v:s)$. This morphism satisfies the following signature morphism condition \[( e\mapsto st)\in \Sigma.E \Rightarrow id_E(e\mapsto st ) \in \Sigma.E \ \ \land \ \ (v:s) \in \Sigma.V \Rightarrow id_V(v:s) \in \Sigma.V\]
\end{enumerate}
We have thus shown that $\mathbf{Sign}_{{\iEVT}}$ forms a category as instructed by the definition of an institution (Definition~\ref{def:ins}).
\end{proof}

\lemsen*

\begin{proof}
$\mathbf{Sen}_{\iEVT}$ is a functor therefore it is necessary to map the {\iEVT}-signature morphisms to corresponding functions over sentences. The functor maps morphisms to functions between sets of sentences respecting sort, arity and \texttt{Init} events. The image of a signature $\Sigma_i$ ($i \in \{1,2,..\}$) in $\mathbf{Sign}_{{\iEVT}}$ is an object $\mathbf{Sen}(\Sigma_i)$ in the category $\mathbf{Set}$. By the definition of the extension of the signature morphism to sentences~\cite[\S 4]{sanella_foundations_2012}, the translation of the object $\mathbf{Sen}(\Sigma_1)$ coincides with an object in $\mathbf{Sen}(\Sigma_2)$ with sort, operation, predicate, event and variable names translated with respect to the signature morphism $\sigma: \Sigma_1 \rightarrow \Sigma_2$. Thus, the image of a morphism in $\mathbf{Sign}_{{\iEVT}}$ is a function $Sen(\sigma): Sen(\Sigma_1) \rightarrow Sen(\Sigma_2)$ in the category $\mathbf{Set}$.

Then we prove that $\mathbf{Sen}$ preserves the composition of {\iEVT}-signature morphisms and identities as follows.
\begin{enumerate}[(\alph*)]
\item Composition of ${\iEVT}$-sentence morphisms:
\[Sen(\sigma_2 \circ \sigma_1) = Sen(\sigma_2) \circ Sen(\sigma_1)\]

Let $\Sigma_i$ ($i = 1..3$) be {\iEVT}-signatures and let $\sigma_1:\Sigma_1 \rightarrow \Sigma_2$ and $\sigma_2: \Sigma_2 \rightarrow \Sigma_3$ be {\iEVT}-signature morphisms. Given an {\iEVT}-sentence of the form $\langle e, \phi(\overline{x}, \overline{x}\prime)\rangle$, then by expanding each side of the following equivalence, since signature morphisms can be composed, we show that composition is preserved.
\begin{align*}
Sen(\sigma_2 \circ \sigma_1)(\langle e, \phi(\overline{x}, \overline{x}\prime)\rangle) &= Sen(\sigma_2) \circ Sen(\sigma_1) (\langle e, \phi(\overline{x}, \overline{x}\prime)\rangle)\\
\langle \sigma_2 \circ \sigma_1(e), \sigma_2 \circ \sigma_1(\phi(\overline{x}, \overline{x}\prime))\rangle &= Sen(\sigma_2)(\langle \sigma_1(e), \sigma_1(\phi(\overline{x}, \overline{x}\prime))\rangle)\\
\langle \sigma_2 (\sigma_1(e)), \sigma_2 (\sigma_1(\phi(\overline{x}, \overline{x}\prime)))\rangle &= \langle \sigma_2(\sigma_1(e)), \sigma_2(\sigma_1(\phi(\overline{x}, \overline{x}\prime)))\rangle
\end{align*}
Thus, composition is preserved.\\

\item Preservation of identities:\\
Let $id_{\Sigma_1}: \Sigma_1 \rightarrow \Sigma_1$ be an identity {\iEVT}-signature morphism as in Lemma~\ref{lem:signevt}. Since {\iEVT}-signature morphisms already preserve identity and $Sen(id_{\Sigma_1})$ is the application of the identity signature morphisms to every element of the sentence, then the identities are preserved. We can illustrate this as follows. Given a $\Sigma_1$-sentence $\langle e, \phi(\overline{x}, \overline{x}\prime)\rangle$,
\begin{align*}
& Sen(id_{\Sigma_1})(\langle e, \phi(\overline{x}, \overline{x}\prime)\rangle)\\
&= \langle id_{\Sigma_1}(e), id_{\Sigma_1}(\phi(\overline{x}, \overline{x}\prime))\rangle\\
&= \langle e, \phi(\overline{x}, \overline{x}\prime)\rangle
\end{align*}
and so, identities are preserved.
\end{enumerate}
Thus $\mathbf{Sen}_{{\iEVT}}$ is a functor.
\end{proof}

\evtmodcatlem*

\begin{proof}
We begin by describing {\iEVT}-model morphisms and then prove that composition and identities are preserved.

In {\iFOPEQ} a model morphism $h:A_1 \rightarrow A_2$ is a family of functions $h = \langle h_s:\lvert A_1\rvert _s \rightarrow \lvert A_2\rvert _s \rangle_{s \in S}$ which respects the sorts and arities of the operations and predicates. Recall from Definition~\ref{def:evtmod} that {\iEVT}-models have the form $\langle A,L, R \rangle$, therefore {\iEVT}-model morphisms are given by extending the corresponding {\iFOPEQ}-model morphisms for the $A$ component of the model to the initialising set $L$ and the relations in $R$.

Thus for each {\iEVT}-model morphism $\mu : \langle A_1,L_1, R_1 \rangle \rightarrow \langle A_2,L_2, R_2 \rangle$ there is an underlying {\iFOPEQ}-model morphism $h: A_1 \rightarrow A_2$, and we extend this to the states in the set $L_1$ and in the relation $R_1$.  That is, for any element \[\{x_1 \mapsto a_{1},\ldots,x_n \mapsto a_{n}, x_1\prime \mapsto a_{1}\prime,\ldots,x_n\prime \mapsto a_{n}\prime\} \in R_1.e\] in $R_1$ there is
\[ \{x_1 \mapsto h(a_{1}),\ldots,x_n \mapsto h(a_{n}),x_1\prime \mapsto h(a_{1}\prime),\ldots,x_n\prime \mapsto h(a_{n}\prime)\} \in R_2.e \]
  in $R_2$ where $x_1, \ldots, x_n, x_1\prime, \ldots, x_n\prime$ are variable names and their primed versions drawn from $V$. A similar construction follows for $L_1$.
The composition of model morphisms, their associativity and identity derives from that of {\iFOPEQ}.

\begin{enumerate}[(\alph*)]
\item Composition of ${\iEVT}$-model morphisms:\\
Let $M_i = \langle A_i,L_i, R_i \rangle$ be a model and $h_i : M_i \rightarrow M_{i+1}$ be an {\iEVT}-model morphism where $i \in \{1,2,3\}$. We can now show that the composition of {\iEVT}-model morphisms is associative as follows:

We know that
\[(h_3 \circ h_2) \circ h_1(A_1)  =  h_3 \circ (h_2 \circ h_1)(A_1)\] since model morphisms and thus their associativity for $A$ are inherited from {\iFOPEQ}.

Next, we consider $L_1$ and $R_1$ whose fundamental components are $\Sigma$-states (variable to value mappings). Signature morphisms on $\Sigma$-states follow the associated sort mappings in {\iFOPEQ} and thus model morphisms on $\Sigma$-states follow model morphisms on $A$ as above. In this way, compositionality of model morphisms is associative for $L_1$ and $R_1$.

\item Identity morphism for ${\iEVT}$-models:\\
For any {\iEVT}-model $M_i$ there exists an identity model morphism $h_{id} : M_i \rightarrow M_i$. If $M_i = \langle A_i,L_i, R_i\rangle$ then $h_{id}(M_i) = \langle A_i, L_i, R_i\rangle$.
\end{enumerate}
Thus $\mathbf{Mod}_{{\iEVT}}(\Sigma)$ forms a category.
\end{proof}

\lemmodfun*
Recall that each $\Sigma$-State$_{A}$ is a set of variable-to-value mappings of the form
\[\{x_1 \mapsto a_{1}, \ldots, x_n \mapsto a_{n}\}\]
where $x_1, \ldots, x_n \in \Sigma.V $ (Definition~\ref{def:state}).

\begin{proof}

Let $M_2 = \langle A_2,L_2, R_2 \rangle$ be a $\Sigma_2$-model. Then the reduct $M_2\rvert _{\sigma}$ collapses the {\iEVT}-model to only contain signature items supported by $\Sigma_1$ and consists of the tuple
$M_2\rvert _{\sigma} = \langle A_2\rvert _{\sigma}, L_2\rvert _{\sigma}, R_2\rvert _{\sigma} \rangle$ such that
\begin{itemize}
\item $A_2 \rvert_{\sigma}$ is the reduct of the {\iFOPEQ}-component of the {\iEVT}-model along the {\iFOPEQ}-components of $\sigma: \Sigma \rightarrow \Sigma'$.
\item $L_2\rvert _{\sigma}$ and $R_2\rvert _{\sigma}$ are based on the reduction of the states of $A_2$ along $\sigma$. In particular, given $e \in dom(E_1)$ and $e \neq \texttt{Init}$ and $R_2.\sigma(e) \in R_2$
\[R_2.\sigma(e) = \{s_1, \ldots, s_m\} \]
where each $s_i$ is a $\Sigma_2$-state$_{A_2}$ ($1 \leq i \leq m$) is of the form
\[\{\sigma(x_1) \mapsto a_1,\ldots, \sigma(x_n) \mapsto a_n, \sigma(x_1\prime) \mapsto a_1\prime,\ldots, \sigma(x_n\prime) \mapsto a_n\prime\}\]
with $x_1,\ldots,x_n \in \Sigma.V $ and $x_1\prime,\ldots,x_n\prime \in \Sigma.V^{\prime}$.

Then for each $e \in dom(E1)$, $e \neq \texttt{Init}$ and $R_2\rvert _{\sigma}.e \in R_2\rvert _{\sigma}$ there is
\[R_2\rvert_{\sigma}.e=\{s_1\rvert _{\sigma},\ldots, s_m\rvert _{\sigma}\}\]
where each $s_i\rvert _{\sigma}$ ($1\leq i \leq m$) is of the form
\[\{x_1 \mapsto a_1, \ldots, x_n \mapsto a_n, x_1\prime \mapsto a_1\prime, \ldots, x_n\prime \mapsto a_n\prime\}\]
\end{itemize}

\noindent
In order to prove that the model reduct is a functor, we show that it preserves composition and identities as follows:
\begin{enumerate}[(\alph*)]
\item Preservation of composition for ${\iEVT}$-model reducts:\\
Our objective here is to show that the reduct of a composition of two {\iEVT}-model morphisms is equal to the composition of the reducts of those {\iEVT}-model morphisms.
Given {\iEVT}-model morphisms $h_1: M_1 \rightarrow M_2$ and $h_2:  M_2 \rightarrow M_3$, then we show that
\[(h_2 \circ h_1)\rvert _{\sigma} = h_2\rvert _{\sigma} \circ h_1\rvert _{\sigma}\]
for some {\iEVT}-signature morphism $\sigma: \Sigma \rightarrow \Sigma'$.
Given an {\iEVT}-model of the form $M_1 = \langle A, L, R\rangle$ over $\Sigma$, then for any $R.e \in R$, as outlined above, of the form \[\{x_1 \mapsto a_1, \ldots, x\prime_n \mapsto a\prime_n\}_e\] Then $(h_2 \circ h_1)|_{\sigma}$ is defined as
\[(h_2 \circ h_1)\{\sigma(x_1) \mapsto a_1, \ldots, \sigma(x\prime_n) \mapsto a\prime_n\}_\sigma(e)\]
This is equal to
\[h_2 \{\sigma(x_1) \mapsto h_1(a_1), \ldots, \sigma(x\prime_n) \mapsto h_1(a\prime_n)\}_\sigma(e)\]
Since {\iEVT}-model morphisms can be composed, this is thus equal to $h_2\rvert _{\sigma} \circ h_1\rvert _{\sigma}$.

\item Preservation of identities for ${\iEVT}$-model reducts:\\
The reduct of the identity is the identity.
Let $id_{M_2}$ be an identity $\Sigma_2$-morphism then $id_{M_2}\rvert _{\sigma}$ is an identity $\Sigma_1$-morphism $h_1$ defined by $h_1(R.e) = id_{M_2}\rvert _{\sigma} (R.e) = R.e$ for any $R.e \in R$ and $ e $ is an event other than \texttt{Init}.
\end{enumerate}

\noindent
For the components belonging to $A$ these proofs follow the corresponding proofs in {\iFOPEQ}.\end{proof}

\lemmodfunmor*

\begin{proof}
For each $\sigma: \Sigma_1 \rightarrow \Sigma_2$ in $\mathbf{Sign}_{\iEVT}$ there is an arrow in $\mathbf{Sign}_{\iEVT}^{op}$ going in the opposite direction. By Lemma 4, the image of this arrow in $\mathbf{Sign}_{\iEVT}^{op}$ is $\mathbf{Mod}(\sigma): \mathbf{Mod}(\Sigma_2) \rightarrow \mathbf{Mod}(\Sigma_1)$ in $\mathbf{Cat}$. By Lemma 3, the image of a signature in $\mathbf{Sign}_{\iEVT}$ is an object $\mathbf{Mod}(\Sigma)$ in $\mathbf{Cat}$. Therefore, domain and co-domain of the image of an arrow are the images of the domain and co-domain respectively.

\begin{enumerate}[(\alph*)]
\item Preservation of composition:
\[\mathbf{Mod}(\sigma_2 \circ \sigma_1) = \mathbf{Mod}(\sigma_2) \circ \mathbf{Mod}(\sigma_1)\]
Let $\sigma_1: \Sigma_1 \rightarrow \Sigma_2$ and $\sigma_2: \Sigma_2 \rightarrow \Sigma_3$ be {\iEVT}-signature morphisms and let $M_i = \langle A_i, L_i, R_i \rangle$ be an {\iEVT}-model over $\Sigma_i$ and let $h_i$ be a $\Sigma_i$-model morphism with $i \in \{1,2,3\}$.
	\begin{itemize}
	\item $M_3\rvert _{\sigma_2 \circ \sigma_1} = (M_3\rvert _{\sigma_2})\rvert _{\sigma_1}$\\
	By definition of reduct
	$M_3\rvert _{\sigma_2} = \langle A_3, L_3, R_3 \rangle\rvert _{\sigma_2} = \langle A_2, L_2, R_2 \rangle = M_2$.\\
	Then
	$(M_3\rvert _{\sigma_2})\rvert _{\sigma_1} = M_2\rvert_{\sigma_1} = \langle A_2, L_2, R_2 \rangle\rvert _{\sigma_1} = \langle A_1, L_1, R_1 \rangle = M_1$.\\
	By composition of signature morphisms $\sigma_2 \circ \sigma_1: \Sigma_1 \rightarrow \Sigma_3$, then \[M_3\rvert _{\sigma_2 \circ \sigma_1}=\langle A_3, L_3, R_3 \rangle\rvert _{\sigma_2 \circ \sigma_1} = \langle A_1,L_1, R_1 \rangle = M_1\]
	Therefore $M_3\rvert _{\sigma_2 \circ \sigma_1} = (M_3\rvert _{\sigma_2})\rvert _{\sigma_1}$
	\item $h_3\rvert _{\sigma_2 \circ \sigma_1} = (h_3\rvert _{\sigma_2})\rvert _{\sigma_1}$\\
	Proof similar to above.
	\end{itemize}

\item Preservation of identities:\\
Let $id_{\Sigma_1}$ be an identity signature morphism as defined in Lemma~\ref{lem:signevt}. Since {\iEVT}-signature morphisms already preserve identity and $Mod(id_{\Sigma_1})$ is the application of the identity signature morphisms to every part of the {\iEVT}-model, the identities are preserved.
\qedhere
\end{enumerate}
\end{proof}

\pushoutprop*
\begin{proof}

Given two {\iEVT}-signature morphisms $\sigma_1: \Sigma \rightarrow \Sigma_1$
and $\sigma_2 : \Sigma \rightarrow \Sigma_2$ a pushout is a triple
$(\Sigma', \sigma_1',\sigma_2')$ that satisfies the universal
property: for all triples $(\Sigma'', \sigma_1'',\sigma_2'')$ there
exists a unique morphism $u: \Sigma' \rightarrow \Sigma''$ such that
the diagram below commutes. Our pushout construction
follows {\iFOPEQ} for the elements that {\iFOPEQ} has in common with
{\iEVT}. In $\textbf{Sign}_{{\iEVT}}$ the additional elements are $E$
and $V$ as presented below.

\begin{center}
	\begin{tikzpicture}[scale=0.4]
	\node (P0) at (90:2cm) {$\Sigma$};
	\node (P1) at (90+90:2.5cm) {$\Sigma_1$} ;
	\node (P2) at (90+2*90:2.5cm) {$\Sigma'$};
	\node (P3) at (90+3.75*72:2.5cm) {$\Sigma_2$};
	\node (P4) at (90+2*90:5cm) {$\Sigma''$};
	\path[commutative diagrams/.cd, every arrow, every label]
	(P0) edge node [swap]{$\sigma_1$} (P1)
	(P1) edge node  {$\sigma_1'$} (P2)
	(P3) edge node[swap] {$\sigma_2'$} (P2)
	(P0) edge node {$\sigma_2$} (P3)
	(P1) edge node [swap]{$\sigma_1''$} (P4)
	(P3) edge node {$\sigma_2''$} (P4)
	(P2) edge[dashed] node {$u$} (P4);
	\end{tikzpicture}
\end{center}
We base our constructions of the pushout in $E$ and $V$ on the canonical pushout in $\mathbf{Set}$.
\begin{itemize}
\item \textit{Set of functions of the form $($event name $\mapsto$ status$)$ in $E$:}
The set of all event names in the pushout is the pushout in $\mathbf{Set}$ (disjoint union)
on event names only.  Then, the status of an event in the pushout is the supremum of all event statuses that are mapped to it, according to the ordering given in Definition~\ref{def:evtsigmor}.
Since {\iEVT}-signature morphisms map $(\texttt{Init}\mapsto \texttt{ordinary})$ to $( \texttt{Init} \mapsto \texttt{ordinary})$ the pushout does likewise. The universality property for $E$ follows from that of $\mathbf{Set}$.

The canonical injections in $\mathbf{Set}$ for $E$ are denoted by $inl$ for inject left and $inr$ for inject right. Then, let $(e_1\mapsto status_1) = \sigma_1(e \mapsto status)$, $(e_2\mapsto status_2) = \sigma_2(e \mapsto status)$ and $status' = status_1 \sqcup status_2$ (where $\sqcup$ denotes the supremum).   
Thus, the pushout in $E$ is given by the formula $E_1 \dot{\cup} E_2 / \sim$ where $\sim$ is the least equivalence relation such that
   \[(inl(e_1)\mapsto status') \sim (inr(e_2)\mapsto status')\]

\item \textit{Set of sort-indexed variable names $V$:}
The set of sort-indexed variable names in the pushout is the pushout in {\iFOPEQ} for the sort components and the pushout in $\mathbf{Set}$ for the variable names.
This is a similar construction to the pushout for operation names in {\iFOPEQ} as these also have to follow the sort pushout.
Thus, the universality property for $V$ follows from that of $\mathbf{Set}$ and the {\iFOPEQ} pushout for sorts.Thus the pushout in $V$ is given by the formula $V_1 \dot{\cup} V_2 / \sim$ where $\sim$ is the least equivalence relation such that \[\sigma'_1 \circ \sigma_1(v:s) \sim \sigma'_2 \circ \sigma_2(v:s)\] for $(v:s)\in \Sigma.V$.
\qedhere
\end{itemize}
\end{proof}
\weakamalgprop*
We decompose this proposition into two further sub-propositions:

\begin{subprop}For $M_1 \in \lvert\textbf{Mod}(\Sigma_1)\rvert$ and $M_2 \in \lvert\textbf{Mod}(\Sigma_2)\rvert$ such that $M_1\rvert _{\sigma_1} = M_2\rvert _{\sigma_2}$, there exists an {\iEVT}-model (the amalgamation of $M_1$ and $M_2$) $ M' \in \lvert\textbf{Mod}(\Sigma')\rvert$ such that $M'\rvert _{\sigma'_1} = M_1$ and $M'\rvert _{\sigma'_2} = M_2$.
\end{subprop}

\begin{proof}
Given the {\iEVT}-models $M \in \lvert \mathbf{Mod}(\Sigma)\rvert, M_1\in \lvert \mathbf{Mod}(\Sigma_1)\rvert, M_2\in \lvert \mathbf{Mod}(\Sigma_2)\rvert$ and {\iEVT}-signature morphisms $\sigma_1: \Sigma \rightarrow \Sigma_1, \sigma_2: \Sigma: \rightarrow \Sigma_2$ in the commutative diagram below.
\begin{center}
	\begin{tikzpicture}[scale=0.6]
	\node (P0) at (90:4cm) {$M' = \langle A',L',R'\rangle$};
	\node (P1) at (90+72:5cm) {$M_1 = \langle A_1,L_1,R_1\rangle$} ;
	\node (P2) at (90+2*90:1cm) {$M = \langle A,L,R\rangle$};
	\node (P3) at (90+4*72:5cm) {$M_2 = \langle A_2,L_2,R_2\rangle$};
	\path[commutative diagrams/.cd, every arrow, every label]
	(P0) edge node [swap] {$Mod(\sigma_1')$} (P1)
	(P1) edge node[swap]  {$Mod(\sigma_1)$} (P2)
	(P3) edge node{$Mod(\sigma_2)$} (P2)
	(P0) edge node {$Mod(\sigma_2')$} (P3);
	\end{tikzpicture}
\end{center}
 We compute the model amalgamation (following the corresponding pushout diagram in $\mathbf{Sign}$ as illustrated in Proposition~\ref{prop:push}) $M' = M_1 \otimes M_2$ which is of the form $\langle A', L', R' \rangle$ where
$A' = A_1 \otimes A_2$ is the 
{\iFOPEQ}-model amalgamation of $A_1$ and $A_2$. We
construct the initialising set $L' = L_1 \otimes L_2$ by amalgamating the initialising sets $L_1$ and $L_2$ to
get the set of all possible combinations of variable mappings, while
respecting the amalgamations induced on variable names via the pushout
$\Sigma.V'$.

For example, suppose that the sort-indexed variable $x$ is in $\Sigma.V$. Then we can apply the {\iEVT}-signature morphisms $\sigma_1: \Sigma \rightarrow \Sigma_1$ and $\sigma_2: \Sigma \rightarrow \Sigma_2$ to get $\sigma_1(x) = x_1 \in \Sigma_1.V$ and $\sigma_2(x) = x_2 \in \Sigma_2.V$.
From Proposition~\ref{prop:push}, we know that the pushout $ (\Sigma_1.V \; \cup \; \Sigma_2.V) /{}_{\sim}$ contains the variable name $x'$ such that $\sigma_1'(x_1) = \sigma_2'(x_2) = x' \in \Sigma'.V$. We follow this pushout construction in order to construct the corresponding model amalgamation. A model $M \in \lvert \mathbf{Mod}(\Sigma)\rvert$ contains an initialising set $L$ which, in turn, contains maplets of the form $x \mapsto a$ where $a$ is a sort-appropriate value for the variable $x$. In light of the {\iEVT}-signature morphisms outlined above, we know that $M_1 \in \lvert \mathbf{Mod}(\Sigma_1)\rvert$ contains maplets of the form $x_1 \mapsto a_1$ and $M_2 \in \lvert \mathbf{Mod}(\Sigma_2)\rvert$ contains maplets of the form $x_2 \mapsto a_2$ where $a_1$ and $a_2$ are sort-appropriate values for $x_1$ and $x_2$ respectively. Then the amalgamation $M' \in \lvert \mathbf{Mod}(\Sigma) \rvert$ has the initialising set $L' = L_1 \otimes L_2$ which contains all maplets of the form $x' \mapsto a'$ where $a'$ is a sort-appropriate value for $x'$ drawn from $A'$ (the {\iFOPEQ}-amalgamation of $A_1$ and $A_2$). Note that initialising sets contain sets of variable-to-value mappings for all of the variables in the $V$ component of their signatures, as such, the amalgamation contains the above maplets in all possible combinations with those corresponding to any other variables that may be in the signature.

We construct the relation $R' = R_1 \otimes R_2$, which is the
amalgamation of $R_1$ and $R_2$, in a similar manner.
Specifically, starting from any $R.e = \{s_1,\ldots,s_m\} \in R$ where $s_1,\ldots,s_m$ are states of the form \[\{x_1 \mapsto a_1,\ldots, x_n\prime \mapsto a_n\prime\}\]
where $x_1,\ldots, x_n\prime$ are the variable names (and their primed versions drawn from $\Sigma.V$).
We construct the corresponding relation $R'.\sigma'(e)$ in $R'$ so that the diagram in Figure~\ref{fig:bigcomm} commutes.

In Figure~\ref{fig:bigcomm}, $h' = (h_1 + h_2)$ is the corresponding function over the carrier-sets in $M'$ obtained from {\iFOPEQ}, and $\sigma' = (\sigma_1' \; \circ \; \sigma_1)+(\sigma_2'\; \circ\; \sigma_2)$
is the mapping for variable and event names obtained from the corresponding construction in $\mathbf{Sign}_{\iEVT}$.
\end{proof}

\begin{figure}
\vspace{50pt}
\begin{tikzpicture}[scale=0.26]
  \node (P0) at (90:22cm)
        {$R'.\sigma'(e) = \{\sigma'(x_i) \mapsto h'(a_i)\} \in R'$};
\node (P1) at (70+72:22.5cm)
      {$R_1.\sigma_1(e) =\{\{\sigma_1(x_i) \mapsto h_1(a_i)\}\in R_1$} ;
\node (P3) at (110+4*72:22.5cm)
      {$R_2.\sigma_2(e) =\{\sigma_2(x_i) \mapsto h_2(a_i)\} \in R_2$};
      \node (P2) at (90:4cm)
            {$R.e =\{x_i \mapsto a_i\} \in R$};
\path[commutative diagrams/.cd, every arrow, every label]
(P0) edge node[swap] {$Mod(\sigma_1')$} (P1)
(P1) edge node[swap] {$Mod(\sigma_1)$} (P2)
(P3) edge node {$Mod(\sigma_2)$} (P2)
(P0) edge node {$Mod(\sigma_2')$} (P3);
\end{tikzpicture}
\caption{The construction of $R' = R_1 \otimes R_2$, the amalgamation of $R_1$ and $R_2$.}%
\label{fig:bigcomm}
\end{figure}

\begin{subprop}
  For any two {\iEVT}-model morphisms $f_1:M_{11} \rightarrow M_{12}$ in $\textbf{Mod}(\Sigma_1)$ and $f_2:M_{21} \rightarrow M_{22}$ in $\textbf{Mod}(\Sigma_2)$ such that $f_1\rvert _{\sigma_1} = f_2\rvert _{\sigma_2}$, there exists an {\iEVT}-model morphism (the amalgamation of $f_1$ and $f_2$) called $f':M'_1 \rightarrow M'_2$ in $\textbf{Mod}(\Sigma')$, such that $f'\rvert _{\sigma'_1} = f_1$ and $f'\rvert _{\sigma'_2} = f_2$.
\end{subprop}
\begin{proof}
Given the {\iEVT}-model morphisms $f_1$ and $f_2$ and their common reduct $f_0$, we construct $f'$ so that the following diagram commutes:
\begin{center}
	\begin{tikzpicture}[scale=0.7]
	\node (P0) at (90:4cm) {$f' : M_1' \rightarrow M_2'$};
	\node (P1) at (90+72:5cm) {$f_1 : M_{11} \rightarrow M_{12}$} ;
	\node (P2) at (90+2*90:1cm) {$f_0 : M_{01} \rightarrow M_{02}$};
	\node (P3) at (90+4*72:5cm) {$f_2 : M_{21} \rightarrow M_{22}$};
	\path[commutative diagrams/.cd, every arrow, every label]
	(P0) edge node [swap] {$Mod(\sigma_1')$} (P1)
	(P1) edge node[swap]  {$Mod(\sigma_1)$} (P2)
	(P3) edge node{$Mod(\sigma_2)$} (P2)
	(P0) edge node {$Mod(\sigma_2')$} (P3);
	\end{tikzpicture}
\end{center}

Since each {\iEVT}-model has a {\iFOPEQ} model as its first component, each of the {\iEVT}-model morphisms $f_0$, $f_1$, $f_2$ and $f'$ must have an underlying {\iFOPEQ}-model morphism, which we denote $f_0^-$, $f_1^-$, $f_2^-$ and $f^{'-}$ respectively. To build the amalgamation for {\iEVT}-models we must show how to extend these to cover the data states of the {\iEVT}-models. This {\iEVT}-model morphism follows the underlying {\iFOPEQ}-model morphism on sort carrier sets for the values in the data states.\smallskip

Given $R.e \in R$, suppose we start with any $f_0$-maplet of the form
\begin{align*}
& \{\ldots,\{x_1 \mapsto a_1, \ldots, x_n \mapsto a_n\}, \ldots\}\\
& \mapsto \{\ldots,\{ x_1 \mapsto f_0^-(a_{1}), \ldots, x_n \mapsto f_0^-(a_{n})\},\ldots\}_e \;\in\; f_0
\end{align*}

\noindent where $f_0^-$ is the underlying map on data types from the {\iFOPEQ}-model morphism. Then the original two functions in $f_1$ and $f_2$ must have maplets of the form
\begin{align*}
& \{\ldots,\{\sigma_1(x_1) \mapsto h_1(a_1), \ldots, \sigma_1(x_n\prime) \mapsto h_1(a_n\prime)\}, \ldots\}\\
& \mapsto \{\ldots,\{ \sigma_1(x_1) \mapsto f_1^-(h_1(a_{1})), \ldots, \sigma_1(x_n\prime) \mapsto f_1^-(h_1(a_{n}\prime))\},\ldots\} \;\in\; f_1
\end{align*}
and
\begin{align*}
& \{\ldots,\{\sigma_2(x_1) \mapsto h_2(a_1), \ldots, \sigma_2(x_n\prime) \mapsto h_2(a_n\prime)\}, \ldots\}\\
& \mapsto \{\ldots,\{ \sigma_2(x_1) \mapsto f_2^-(h_2(a_{1})), \ldots, \sigma_2(x_n\prime) \mapsto f_2^-(h_2(a_{n}\prime))\},\ldots\} \;\in\; f_2
\end{align*}
where $f_1^-$ and $f_2^-$ are again the data type maps from the underlying {\iFOPEQ}-model morphism, and $h_1$ and $h_2$ are obtained from $Mod(\sigma_1)$ and $Mod(\sigma_2)$.

We then can construct the elements of the {\iEVT}-model morphism $f'$, which is the amalgamation of $f_1$ and $f_2$, as $f'$-maplets of the form:

\begin{align*}
& \{\ldots,\{\sigma'(x_1) \mapsto h'(a_1), \ldots, \sigma'(x_n\prime) \mapsto h'(a_n\prime)\}, \ldots\}\\
& \mapsto \{\ldots,\{ \sigma'(x_1) \mapsto f'^-(h'(a_{1})), \ldots, \sigma'(x_n\prime) \mapsto f'^-(h'(a_{n}\prime))\},\ldots\} \;\in\; f'
\end{align*}
As before, $h' = (h_1 + h_2)$ is the corresponding function over the carrier-sets in $M'$ obtained from {\iFOPEQ}, and $\sigma' = (\sigma_1' \; \circ \; \sigma_1)+(\sigma_2'\; \circ\; \sigma_2)$ is
the mapping for variable and event names obtained from the corresponding construction in $\mathbf{Sign}$.  Here $f^{'-} = f_1^-+f_2^-$ is the model morphism amalgamation from the corresponding diagram for model morphisms in {\iFOPEQ}, which ensures that the data states are mapped to corresponding states in the model $M_2'$.
\end{proof}

\end{document}